%% file: ms.tex
\documentclass[10pt]{article}
\input{header}

\usepackage{fullpage, parskip}
\usepackage{stackengine}
\usepackage{nicefrac}
\usepackage{booktabs}

\newcommand{\includesupp}{1}

\newcommand{\switchref}[2]{%
  \if\includesupp1%
    #1%
  \else%
    #2%
  \fi%
}

\hypersetup{pdfborder = {0 0 0.5 [3 3]}, colorlinks = true, linkcolor = BrightRed, citecolor = SkyBlue}

\title{High-dimensional regression with outcomes of mixed-type using the multivariate spike-and-slab LASSO}
\author{Soham Ghosh\thanks{Dept.~of Statistics, University of Wisconsin--Madison. \texttt{sghosh39@wisc.edu}} \and Sameer K. Deshpande\thanks{Dept.~of Statistics, University of Wisconsin--Madison. \texttt{sameer.deshpande@wisc.edu}}.}

\begin{document}
\maketitle

\input{abstract}
\newpage

\section{Introduction}
\label{sec:introduction}
\input{intro}

\section{Modeling Mixed-type Responses}
\label{sec:method}
\input{method}

\section{The MCECM Algorithm for parameter estimation}
\label{sec:algorithm}
\input{short_algorithm}

\section{Theoretical Results}
\label{sec:theory}
\input{short_theory}

\section{Synthetic Experiments}
\label{sec:experiments}
\input{short_experiments}

\section{Real Data Analysis}
In this section, we demonstrate \texttt{mixed-mSSL}'s performance in settings with large $n$ (\Cref{sec:ckddata}), large $p$ (\Cref{sec:crcdata}), and large $q$ (\Cref{sec:hmscdata}).
\label{sec:realdata}
\subsection{Chronic Kidney Disease Data}
\label{sec:ckddata}

\input{ckddata}

\subsection{Colorectal Cancer Data}
\label{sec:crcdata}
\input{crcdata}

\subsection{Finnish Bird Data}
\label{sec:hmscdata}
\input{hmscdata}

\section{Discussion}
\label{sec:discussion}
\input{discussion}

\bibliographystyle{apalike} 
\bibliography{refs}

\newpage
\appendix
\begin{center}
{
\LARGE
\textbf{Supplementary Materials}
}
\end{center}

\renewcommand{\thefigure}{\thesection\arabic{figure}}
\renewcommand{\thetable}{\thesection\arabic{table}}
\renewcommand{\theequation}{\thesection\arabic{equation}}
\renewcommand{\thelemma}{\thesection\arabic{lemma}}
\renewcommand{\thetheorem}{\thesection\arabic{theorem}}

\setcounter{figure}{0}
\setcounter{equation}{0}
\setcounter{table}{0}
\setcounter{theorem}{0}
\section{Proofs}
\label{sec:proofs}
\input{proofs}

\setcounter{figure}{0}
\setcounter{equation}{0}
\setcounter{table}{0}
\setcounter{theorem}{0}
\section{Additional algorithmic details}
\label{sec:extraalgorithm}
This section collects the minimal derivations referenced in the main text: the indicator–augmentation calculus and the closed-form E–step updates used to compute the adaptive penalties. We omit definitions and high-level motivation already given in Section 3.

\input{extraalgorithm}

\setcounter{figure}{0}
\setcounter{equation}{0}
\setcounter{table}{0}
\setcounter{theorem}{0}
\section{Additional simulation details and results}
\label{sec:additional_experiments}
\input{additional_experiments}

\setcounter{figure}{0}
\setcounter{equation}{0}
\setcounter{table}{0}
\section{Additional details for the real data analyses}
\label{sec:additional_real}
\input{additional_real}

\end{document}

%% file: header.tex
\usepackage{amsmath, amssymb, amsthm}
\usepackage{algorithm}
\usepackage{algpseudocode}
\usepackage{amsmath, amssymb, amsthm}
\usepackage{float, graphicx, multirow,subcaption, setspace}
\usepackage{comment}
\usepackage{url}
\usepackage{enumitem}
\usepackage{tikz}
\usepackage{bm, bbm}
\usetikzlibrary{shapes.geometric, positioning}
\usetikzlibrary{quotes, angles}
\usepackage{rotating}
\usepackage{hyperref}
\usepackage{cleveref}
\usepackage{natbib}
\usepackage{mathtools}
\usepackage{booktabs}

\definecolor{SkyBlue}{RGB}{14, 118, 188}
\definecolor{BrightRed}{RGB}{223, 82, 78}
\definecolor{Green638}{RGB}{165,255,118} 

\newcommand{\R}{\mathbb{R}} 
\newcommand{\E}{\mathbb{E}} 
\def\P{\mathbb{P}} 

\newcommand{\ind}[1]{\mathbbm{1}\left( #1 \right)} 

\newcommand{\mvnormaldist}[3]{\mathcal{N}_{#1}\left(#2,#3\right)} 
\newcommand{\berndist}[1]{\textrm{Bernoulli}\left(#1\right)} 

\newcommand{\by}{\bm{y}}
\newcommand{\bx}{\bm{x}}
\newcommand{\bz}{\bm{z}}

\newcommand{\br}{\bm{r}}

\newcommand{\bY}{\bm{Y}}
\newcommand{\bX}{\bm{X}}
\newcommand{\bZ}{\bm{Z}}

\newcommand{\bmu}{\boldsymbol{\mu}}
\newcommand{\bdelta}{\boldsymbol{\delta}}

\theoremstyle{plain}
\newtheorem{theorem}{Theorem}

\newtheorem{lemma}{Lemma}
\newtheorem{proposition}{Proposition}

\theoremstyle{definition}

\newtheorem{ex}{Example}

\theoremstyle{plain}

%% file: abstract.tex
We consider a high-dimensional multi-outcome regression in which $q,$ possibly dependent, binary and continuous outcomes are regressed onto $p$ covariates.
We model the observed outcome vector as a partially observed \emph{latent} realization from a multivariate linear regression model.
Our goal is to estimate simultaneously a sparse matrix ($\bm{B}$) of latent regression coefficients (i.e., partial covariate effects) and a sparse latent residual precision matrix ($\Omega$), which induces partial correlations between the observed outcomes.
To this end, we specify continuous spike-and-slab priors on all entries of $\bm{B}$ and the off-diagonal elements of $\Omega$ and derive a Monte Carlo Expectation-Conditional Maximization algorithm to compute the maximum a posteriori estimate.
Under a set of mild assumptions, we derive the posterior contraction rate for our model in the high-dimensional regimes where both $p$ and $q$ diverge with the sample size $n$ and establish a sure screening property, which implies that, as $n$ increases, we can recover all truly non-zero elements of $\bm{B}$ with probability tending to one.
We demonstrate the excellent finite-sample properties of our proposed method, which we call \texttt{mixed-mSSL}, using extensive simulation studies and three applications spanning medicine to ecology.

%% file: intro.tex
Several scientific applications involve estimating the effects of numerous covariates on multiple, possibly interdependent outcomes.
When all outcomes are continuous, a multi-output linear regression model is a natural starting point.
In this setting, jointly modeling the regression coefficients and the residual covariance matrix confers many advantages over fitting separate models to each outcome: parameter estimates from the joint model are asymptotically more efficient \citep{Zellner1962_sur} and, in high-dimensional settings, joint models can detect smaller covariate effects \citep{Deshpande2019_mSSL} and produce more accurate predictions \citep{LiGhosh2023}.
But when outcomes are of \emph{mixed-type} --- that is, when some outcomes are continuous and others
are discrete; specifying joint models that capture the residual dependence between the outcomes is much more complicated. 
By expressing the vector of continuous and discrete observations as a partially observed realization from a sparse multi-output linear regression model, we show that many of the benefits of joint modeling can be realized in the mixed-type outcome setting. 

Specifically, for $n$ observations of a $p$-dimensional covariate vector $\bx$ and a $q$-dimensional outcome $\by$, we introduce a latent Gaussian variable $\bz_i \sim \mvnormaldist{q}{\bm{B}^{\top}\bx_i}{\Omega^{-1}}$ such that $\by_i = g(\bz_i),$ where $g$ is a deterministic function (defined formally in \Cref{sec:method}). This framework allows us to simultaneously estimate the $p \times q$ regression matrix $\bm{B}$ and the $q \times q$ latent residual precision matrix $\Omega$. In each of the following applications, our goal is to estimate $\bm{B}$ to characterize covariate effects on multiple outcomes and $\Omega$ to quantify residual cross-outcome dependence.
\begin{itemize}
\item{\textbf{Biomarkers for kidney disease:} Like \citet{CKDdata}, we seek the most important clinical predictors for two outcomes -- chronic kidney disease status (binary) and urine specific gravity (continuous) -- using data on $n=400$ patients and $p=24$ covariates.}
\item{\textbf{Colorectal cancer and the microbiome:} In a high-dimensional setting ($p=849$ bacterial species $> \ n =574 $ subjects) obtained through meta studies outlined in \citet{Wirbel2019meta} and \citet{Qin2012T2D}, we aim to identify gut microbes associated with three linked health indicators: colorectal cancer status (binary), Type-2 diabetes status (binary), and BMI (continuous). 
}
\item{\textbf{Finnish birds:} Following \citet{Lindstrom15}, we want to determine how climate covariates affect the presence-absence of $q=50$ most common Finnish bird species across $137$ locations and, from the residual dependencies, infer a network of inter-species interactions unexplained by climate.}
\end{itemize}

For the CRC and Finnish bird datasets, the total number of parameters $(pq+q(q-1)/2)$ exceeds $n$, so we assume sparsity in $\bm B$ and/or $\Omega$. We enforce this via \emph{element-wise} continuous spike-and-slab LASSO (\texttt{SSLASSO}) priors on entries of $\bm B$ and on the off-diagonals of $\Omega$ \citep{RockovaGeorge2018_ssl}. Intuitively, \texttt{SSLASSO} uses a continuous mixture of a sharp ``spike" penalty and a diffuse ``slab" penalty to perform automatic variable selection, aggressively shrinking noise while preserving large signals. Motivated by this setting, we propose \texttt{mixed-mSSL} to learn $(\bm B,\Omega)$ for mixed-type outcomes with large $p$ or $q$. The method extends \texttt{SSLASSO} and its multivariate variant \texttt{mSSL} \citep{Deshpande2019_mSSL} to mixed outcomes via a latent Gaussian formulation.

We develop a Monte Carlo Expected--Conditional Maximization (MCECM) algorithm that targets the maximum a posteriori (MAP) estimates of $(\bm{B}, \Omega)$.
Our algorithm alternates between a Monte Carlo E-step and two conditional maximization steps that sequentially update $\bm{B}$ and $\Omega$. 
In simulation studies, our \texttt{mixed-mSSL} algorithm demonstrated strong accuracy and specificity for joint variable and covariance selection, outperforming a recently proposed Bayesian competitor \citep[\texttt{mt-MBSP}]{wang2023twostep} and two baseline approaches, one that fits separate spike-and-slab lasso models (\texttt{sepSSL}) to each outcome and another using frequentist elastic-net penalized GLMs (\texttt{sepglm}). 

We derive posterior contraction rates for both  $\bm{B}$ and $\Omega$ in the regime where both the number of covariates $p$ and the number of outcomes $q$ grow with the sample size $n$. 
In doing so, our results significantly extend earlier results on posterior consistency for $\bm{B}$ like \citet{BAIGhoshJMVA}, who assumed that $q$ is fixed, and \citet{wang2023twostep}, who did not establish posterior contraction for $\Omega.$ 
Under an additional ``beta-min'' condition, we further show that \texttt{mixed-mSSL} can consistently recover the non-zero elements of $\bm{B}$ (\Cref{thm:surescreening}). 

On the motivating datasets, \texttt{mixed-mSSL} obtained promising empirical results. 
For the CKD data, it not only recovered known biomarkers from the original study \citep{CKDdata} but also uniquely identified random blood glucose as a risk factor, a finding corroborated in follow-up studies. 
Furthermore, it achieved the highest predictive accuracy (AUC and prediction error) on the CRC microbiome data and uncovered an ecologically plausible species interaction network in the Finnish bird data. 

The rest of the paper is organized as follows. 
In \Cref{sec:method}, we introduce the \texttt{mixed-mSSL} model in detail before deriving our MCECM algorithm for MAP estimation in \Cref{sec:algorithm}.
We state our main theoretical results in \Cref{sec:theory}, deferring all proofs to \switchref{\Cref{sec:proofs}}{Section S1 in the Supplementary Materials}. 
Then, in \Cref{sec:experiments}, we report the results from several synthetic data simulations comparing \texttt{mixed-mSSL} to several competing methods.
We apply \texttt{mixed-mSSL} to each of the three motivating datasets in \Cref{sec:realdata} and conclude with a discussion of potential extensions in \Cref{sec:discussion}.

%% file: method.tex
Suppose we have $n$ pairs $(\bx_{1}, \by_{1}), \ldots, (\bx_{n}, \by_{n})$ of covariates $\bx \in \R^{p}$ and outcomes $\by \in \R^{q_{c}} \times \{0,1\}^{q_{b}}.$ 
We model, for each $i = 1, \ldots, n$,  $\bz_{i} \vert \bm{B}, \Omega \sim \mvnormaldist{q}{\bm{B}^{\top}\bx_{i}}{\Omega^{-1}}$ such that $\by_{i} = g(\bz_{i}),$  where the function $g: \R^{q} \mapsto \R^{q_{c}} \times \{0,1\}^{q_{b}}$ operates element-wise and is given by:
\begin{equation}
\label{eq:g_function}
g(z_1, \ldots, z_{q_c}, z_{q_c+1}, \ldots, z_q) = (z_{1}, \ldots, z_{q_{c}}, \ind{z_{q_{c}+1} \geq 0}, \ldots, \ind{z_{q} \geq 0})^{\top}.
\end{equation}

That is, $g$ leaves the first $q_{c}$ components of its argument unchanged and thresholds the remaining $q_{b}$ components at zero, assuming $q=q_{c}+q_{b}$. Because only the \emph{sign} of each latent binary coordinate is observed, its scale is not identifiable. 
We therefore fix the residual variance of every binary latent component to be one;
that is, we constrain $\sigma^{2}_{k}: = (\Omega^{-1})_{k,k} = 1$ for $k = q_{c}+1, \ldots, q.$

Our proposed model generalizes \citet{Chib1998AnalysisOM}'s multivariate probit regression model.
It is also closely related to \citet{CanaleDunson2011}'s rounding-based approach to modeling count processes.
Our proposed model is similar to the one studied in \citet{wang2023twostep}, which introduces $q_{b}$ latent P\'{o}lya-Gamma \citep{Polson2013_polyagamma} variables (one for each binary outcome) that are related to a $q$-dimensional latent Gaussian $\bz$; see Equation 2.6 of \citet{wang2023twostep} for details. 
\citet{ekvall}'s \texttt{mmrr} procedure also models mixed-type outcomes using latent Gaussian variables.
However, their approach fundamentally differs from ours: while our model treats the observed outcomes as deterministic transformations of latent variables $(\by=g(\bz)),$ their method assumes canonical generalized linear model (GLM) links relating the conditional mean of each observed outcome directly to the latent variables (i.e., $\E(\bm{Y} \vert \bm{Z})=h^{-1}(\bm{Z})$ for a known link function $h$).

The precision matrix $\Omega$ captures residual dependence between the components of the latent variable $\bz,$ which induces dependence between the observed outcomes.
Although the model allows for arbitrary correlations between elements of the latent $\bz,$ it restricts the range of observable correlations between outcomes of mixed type (\Cref{lem:correlation_bounds}).

\begin{lemma}[Maximal observed correlation]\label{lem:correlation_bounds}
Fix an $\bx \in \R^{p}$ and let $\bY=g(\bZ)$ where $g$ is the function in \Cref{eq:g_function} and $\bZ \sim \mvnormaldist{q}{\bm{B}^{\top}\bx}{\Omega^{-1}}.$
For $k \neq k',$ the maximum absolute conditional correlation $\lvert \mathrm{Corr}(Y_{k}, Y_{k'} \vert \bX = \bx) \rvert$ is 1 if $Y_{k}$ and $Y_{k'}$ are both binary or both continuous and $\sqrt{2/\pi} \approx 0.8$ if one of $Y_{k}$ and $Y_{k'}$ is continuous and the other binary.\end{lemma}
We prove \Cref{lem:correlation_bounds} in \switchref{\Cref{sec:mixedcorr}}{Section S1.3 in the Supplementary Materials} using a similar argument to the proof of Lemma 2.1 in \citet{ekvall}.

\subsection{Prior specification}
\label{sec:priorspec}
When the total number of parameters in the latent model, $pq + q(q-1)/2,$ exceeds the sample size, not all the entries in $\bm{B}$ and $\Omega$ are likelihood-identified.
To make the estimation problem tractable, it is common to assume that both $\bm{B}$ and $\Omega$ are sparse.
Rather than specifying a prior fully supported on exactly sparse matrices, we follow \citet{Deshpande2019_mSSL} and specify continuous Laplacian spike-and-slab priors on the elements $\beta_{j,k}$ and $\omega_{k,k'}.$

Formally, we fix positive constants $0 < \lambda_{1} \ll \lambda_{0}$ and model the entries $\beta_{j,k}$ as being drawn from either a very diffuse $\textrm{Laplace}(\lambda_{1})$ slab or $\textrm{Laplace}(\lambda_{0})$ spike that is sharply concentrated around zero.
Letting $\theta \in [0,1]$ be the probability that each $\beta_{j,k}$ is drawn from the slab, our conditional prior density of $\bm{B}$ is 
$$
p(\bm{B} \vert \theta) = \prod_{j = 1}^{p}{\prod_{k = 1}^{q}{\left[\frac{\theta\lambda_{1}}{2}e^{-\lambda_{1}\lvert \beta_{j,k} \rvert} + \frac{(1-\theta)\lambda_{0}}{2}e^{-\lambda_{0}\lvert \beta_{j,k} \rvert}\right]}}.
$$

For $\Omega,$ it suffices to specify a prior for the diagonal elements $\omega_{k,k}$ and the off-diagonal elements in the upper triangle $\omega_{k,k'}$ with $k < k'.$
To this end, entirely analogously with $\bm{B},$ we fix two more positive constants $0 < \xi_{1} \ll \xi_{0};$ introduce a mixing proportion $\eta \in [0,1];$ and model the $\omega_{k,k'}$'s as conditionally independently drawn from a $\textrm{Laplace}(\xi_{1})$ slab with probability $\eta$ or a $\textrm{Laplace}(\xi_{0})$ spike with probability $(1-\eta).$
We further place independent $\textrm{Exponential}(\xi_{1})$ priors on the diagonal elements and truncate the prior to the positive definite cone, yielding the conditional density
\begin{align*}
p(\Omega \vert \eta) &\propto \ind{\Omega \succ 0} \times \prod_{k=1}^{q}{\xi_{1}e^{-\xi_1 \omega_{k,k}}} \times \prod_{1 \leq k < k' \leq q}{\left[\eta \xi_{1}e^{-\xi_{1}\lvert \omega_{k,k'} \rvert} + (1 - \eta)\xi_{0}e^{-\xi_{0}\lvert \omega_{k,k'}\rvert} \right]}
\end{align*}
To model our uncertainty about the proportion of entries in $\bm{B}$ and $\Omega$ drawn from the slab (i.e., their overall sparsity), we place independent Beta $(a_{\theta},b_{\theta})$ and Beta $(a_{\eta},b_{\eta})$ priors on $\theta$ and $\eta$ respectively where $a_{\theta}, b_{\theta}, a_{\eta},$ and $b_{\eta}$ are fixed positive constants.

%% file: short_algorithm.tex
We approximate the MAP of $\Xi=(\bm B,\theta,\Omega,\eta)$ via a Monte Carlo ECM (MCECM) routine. 
Let $\bX\in\R^{n\times p}$, $\bY\in\R^{n\times q}$, $\bZ\in\R^{n\times q}$ collect rows $\bx_i^\top,\by_i^\top,\bz_i^\top$. 
Throughout, we assume that the columns of $\bX$ are centered and scaled to have $\ell_{2}$ norm $\sqrt{n}.$
We partition $\by_i=(\by_i^{(C)},\by_i^{(B)})\in\R^{q_c}\times\{0,1\}^{q_b}$ and similarly partition $\bz_i=(\bz_i^{(C)},\bz_i^{(B)}).$
For $\by_i^{(B)}=(y_{i1}^{(B)},\dots,y_{iq_b}^{(B)})$, we define the orthant $\mathcal H(\by_i^{(B)})=\{\bz\in\R^{q_b}:y_{ik}^{(B)}=\mathbbm{1}(z_k>0)\ \forall k\}$.
Under our partially observed latent Gaussian model, $\bz_i\in\R^{q_c}\times\mathcal H(\by_i^{(B)})$. 
For $\Omega\succ0$, the log posterior of $\Xi$ is given as,
\begin{align}
\begin{split}
\label{eq:xi_log_posterior}
\log p(\Xi \vert \bY) &= \sum_{i = 1}^{n}{\log p(\by_{i} \vert \bm{B}, \Omega)} + \sum_{j = 1}^{p}{\sum_{k = 1}^{q}{\log\left[\theta\lambda_{1}e^{-\lambda_{1}\lvert \beta_{j,k} \rvert} + (1-\theta)\lambda_{0}e^{-\lambda_{0}\lvert \beta_{j,k}\rvert} \right]}} \\
&+ \sum_{1 \leq k < k' \leq q}{\log\left[\eta\xi_{1}e^{-\xi_{1}\lvert \omega_{k,k'} \rvert} + (1-\eta)\xi_{0}e^{-\xi_{0}\lvert \omega_{k,k'}\rvert} \right]} - \xi_{1}\sum_{k = 1}^{q}{\omega_{k,k}} \\
&+ (a_\theta -1)\log\theta + (b_\theta -1)\log(1-\theta) + (a_{\eta} - 1)\log\eta + (b_{\eta} - 1)\log(1-\eta).
\end{split}
\end{align}
This is hard to optimize due to the (i) non-concavity of $\log p(\bm{B}, \Omega \vert \theta, \eta)$ and (ii) the analytic intractability of $p(\by_{i} \vert \bm{B}, \Omega).$

\textbf{Augmentation and surrogate objective.}
We overcome the challenges of the non-concave penalties with an EM-like algorithm similar to \citet{Deshpande2019_mSSL}.
Specifically, we introduce indicators $\bdelta^{(\beta)}=\{\delta^{(\beta)}_{jk}\}$ and $\bdelta^{(\omega)}=\{\delta^{(\omega)}_{kk'}\}$ encoding whether or not the elements $\beta_{j,k}$ and $\omega_{k,k'}$ are drawn from their respective slab ($\delta = 1$) or spike ($\delta = 0$) distributions.
We iteratively compute and maximize a surrogate objective function:
\begin{align}
\begin{split}
\label{eq:ecm_surrogate}
F(\Xi) & = \E[\log p(\Xi, \bdelta^{(\beta)}, \bdelta^{(\omega)} \vert \Xi, \bY] \\
 & = \sum_{i = 1}^{n}{\log p(\by_{i} \vert \bm{B}, \Omega)} -\sum_{j=1}^{p}{\sum_{k=1}^{q}{\lambda^{\star}_{j,k}\lvert \beta_{j,k} \rvert }} - \sum_{1 \leq k < k' \leq q}{\xi^{\star}_{k,k'}\lvert \omega_{k,k'} \rvert} - \xi_{1}\sum_{k = 1}^{q}{\omega_{k,k}} \\
&+ \left(a_\theta -1 + \sum_{j = 1}^{p}{\sum_{k = 1}^{q}{{p^{\star}_{j,k}}}}\right)\log\theta + \left(b_\theta -1 + pq - \sum_{j = 1}^{p}{\sum_{k = 1}^{q}{{p^{\star}_{j,k}}}}\right)\log(1-\theta) \\
& + \left(a _{\eta} - 1 + \sum_{1 \leq k < k' \leq q}{q^{\star}_{k,k'}}\right)\log \eta \\
& + \left(b_{\eta} - 1 + \frac{q(q-1)}{2} - \sum_{1 \leq k < k' \leq q}{q^{\star}_{k,k'}}\right)\log(1-\eta),
\end{split}
\end{align}
where $p^\star_{jk}=\E[\delta^{(\beta)}_{jk}\vert\bm B,\theta],$ $q^\star_{kk'}=\E[\delta^{(\omega)}_{kk'}\vert\Omega,\eta],$ $\lambda^\star_{jk}=\lambda_1 p^\star_{jk}+\lambda_0(1-p^\star_{jk}),$ and $\xi^\star_{kk'}=\xi_1 q^\star_{kk'}+\xi_0(1-q^\star_{kk'})$ are available in closed form (see \switchref{\Cref{sec:slabprobs}}{Section S2.1 in the Supplementary Materials for more details}). 

\subsection{Monte Carlo log-likelihood approximation}
\label{subsec:mcll}
Maximizing $F(\Xi)$ involves solving a penalized maximum likelihood problem with a separable, concave penalty. 
Unfortunately, $F(\Xi)$ is not available in closed form.
To see this, we have for each $i = 1, \ldots, n$
\[
p(\by_i\mid\bm B,\Omega)
=\int_{\R^{q_c}}\int_{\mathcal H(\by_i^{(B)})}
\delta\!\big(\by_i^{(C)}-\bz_i^{(C)}\big)\,
\phi_q\!\big(\bz_i;\bm B^\top\bx_i,\Omega^{-1}\big)\,
\mathrm d\bz_i^{(B)}\,\mathrm d\bz_i^{(C)},
\]
where $\delta(\cdot)$ is the Dirac density and $\phi_{q}(\cdot; \mu, \Sigma)$ is the density of the $\mvnormaldist{q}{\mu}{\Sigma}$ distribution. The inner integral is a multivariate normal orthant probability, which is analytically intractable for $q_b \ge 2$. Hence $\log p(\by_i\vert \bm B,\Omega)$ is analytically unavailable, and $F(\Xi)$ cannot be evaluated in closed form.

We therefore replace $\sum_i\log p(\by_i\vert\bm B,\Omega)$ in \Cref{eq:ecm_surrogate} by a Monte Carlo average:
\(
H^{-1}\sum_{i=1}^n\sum_{h=1}^H \log p(\by_i,\bz_i^{(h)}\vert\bm B,\Omega),
\)
where $\bz_i^{(h)}$ are i.i.d.\ draws from $p(\bz_i\vert\by_i,\bx_i,\bm B,\Omega)$. Since $\bz_i^{(C)}=\by_i^{(C)}$, only $\bz_i^{(B)}$ is sampled; conditional on $(\by_i,\bx_i,\bm B,\Omega)$ it is a $q_b$-variate Gaussian truncated to $\mathcal H(\by_i^{(B)})$ (see \switchref{\Cref{sec:zposterior}}{Section S1.1 in the Supplementary Materials}). 
We generate $\bz_i^{(B)}$ using the \texttt{LinESS} elliptical slice sampler of \citet{gessner2020integralsgaussianslineardomain}. Thus at the $t^{th}$ iterate, instead of optimizing $F^{(t)}(\Xi)$ with two CM steps, we instead optimize the objective $\tilde{F}^{(t)}$ which replaces the term $\sum_{i = 1}^{n}{\log p(\by_{i} \vert \bm{B}, \Omega)}$ in \Cref{eq:ecm_surrogate} with the Monte Carlo average of the $\bm{z}_i$'s drawn.

\subsection{Our MCECM algorithm}\label{sec:mcecm}
Our MCECM algorithm proceeds by augmenting the model with latent $\bz_i$'s and spike-and-slab indicators $\bdelta^{(\beta)}$ and $\bdelta^{(\omega)}$. Each iteration involves a Monte Carlo E-step to update the adaptive penalties and approximate the log-likelihood, followed by two CM-steps that sequentially update $(\bm{B}, \theta)$ and $(\Omega, \eta)$.

\textbf{CM Step 1: updating $\bm{B}$ and $\theta.$} We first update $(\bm{B},\theta)$ by maximizing the relevant terms of the surrogate objective function:
\begin{align*}
\tilde{F}^{(t)}(\bm{B}, \theta, \Omega^{(t-1)}, \eta^{(t-1)}) &= - \frac{1}{2H}\textrm{tr}\left(\mathbb{M}(\bm{B})^{\top}\mathbb{M}(\bm{B})\Omega^{(t-1)}\right) - \sum_{j=1}^{p}{\sum_{k=1}^{q}{\lambda^{\star}_{j,k}\lvert \beta_{j,k}\rvert}} \\
&+  \left(a_\theta -1 + \sum_{j = 1}^{p}{\sum_{k=1}^{q}{{p^{\star}_{j,k}}}}\right)\log\theta + \left(b_\theta -1 + pq - \sum_{j = 1}^{p}{\sum_{k=1}^{q}{{p^{\star}_{j,k}}}}\right)\log(1-\theta),
\end{align*}
where $\mathbb{M}$ is an $(nh)\times q$ matrix whose rows are $(\bm{z}_i^{(h)} - \bm{B}^\top \bm{x}_i)^\top$. The update for $\theta$ is available in closed form. For $\bm{B}$, we use a cyclic coordinate ascent algorithm where each element $\beta_{jk}$ is updated using a dual-thresholding rule:
\begin{equation}
    \label{eq:hardsoft}
    \beta_{jk}^{\text{new}} \;=\; [|S_{jk}|-\lambda_{jk}^{\star}]_{+} \frac{\text{sign}(S_{jk})}{nH \omega_{kk}^{(t-1)}} \mathbbm{1} \left( \left| \frac{S_{jk}}{nH \omega_{kk} ^{(t-1)}} \right| > \Delta_{jk} \right).
\end{equation}
This update performs both selection and shrinkage. 
If the precision-weighted score $|S_{jk}|$ is below a data-driven threshold $\Delta_{jk}$, the coefficient is set to zero. 
Otherwise, it is shrunk by the adaptive soft-thresholding penalty $\lambda_{jk}^\star$, which is smaller for entries with a high posterior probability of being non-zero.
See \switchref{\Cref{sec:dualthreshold}}{Section S2.2 in the Supplementary Materials for a complete derivation}.

\textbf{CM Step 2: updating $\Omega$ and $\eta$}. After updating $\bm{B}$ and $\theta,$ we update $(\Omega, \eta)$ by maximizing the function
\begin{align*}
\tilde{F}^{(t)}(\bm{B}^{(t)}, \theta^{(t)}, \Omega, \eta) &= \frac{n}{2}\log \lvert \Omega \rvert - \frac{1}{2H}\textrm{tr}\left(\mathbb{M}(\bm{B}^{(t)})^{\top}\mathbb{M}(\bm{B}^{(t)})\Omega \right) - \xi_{1}\sum_{k = 1}^{q}{\omega_{k,k}} - \sum_{1 \leq k < k' \leq q}{\xi^{\star}_{k,k'}\lvert \omega_{k,k'} \rvert} \\
&+ \left(a _{\eta} - 1 + \sum_{1 \leq k < k' \leq q}{q^{\star}_{k,k'}}\right)\log \eta + \left(b_{\eta} - 1 + \frac{q(q-1)}{2} - \sum_{1 \leq k < k' \leq q}{q^{\star}_{k,k'}}\right)\log(1-\eta).
\end{align*}
Just like in the first CM step, this objective is separable and we can update $\eta$ in closed-form.
Updating $\Omega$ essentially amounts to solving a graphical LASSO problem with individual penalties.
In our implementation, we solve this problem using \citet{JMLR:v15:hsieh14a}'s QUIC algorithm. 

\textbf{Implementation.} 
Our MCECM algorithm depends on eight prior hyper-parameters: the spike and slab penalties $\lambda_{0}, \xi_{0}, \lambda_{1},$ and $\xi_{1}$ the Beta hyper-parameters $a_{\theta}, b_{\theta}, a_{\eta},$ and $b_{\eta}.$
Instead of running the algorithm for one specific set of hyperparameters to estimate the mode of one specific posterior, we perform what \citet{RockovaGeorge2018_ssl} term ``dynamic posterior exploration'' by running the algorithm along a grid of hyper-parameter settings with warm-starts.
Specifically, we fix the slab penalties $\lambda_{1}\approx 1/\sqrt{n\log n}$ and $\xi_{1}=n/100$ and $a_\theta=1,\ b_\theta=pq$ and $a_\eta=1,\ b_\eta=q.$
Then, we run the algorithm for every combination of fixed ten equally spaced spike penalties $\lambda_{0}\in[10,100]$ and $\xi_{0}\in[n/10,n].$
In doing so, we maximize several posterior distributions, one for each combination of spike penalties.
Essentially, by using warm-starts, our implementation propagates an initial estimate $(\bm{B}, \Omega)$ through a series of increasingly stringent filters.
See \switchref{\Cref{sec:extraalgorithm}}{Section S2.3 in the Supplementary Materials} for further details.

%% file: short_theory.tex
If our model is well-specified, that is, there are true data generating parameters $\bm{B}_{0}$ and $\Omega_{0}$ such that each $\by_{i}$ is a partially observed version of a latent $\mvnormaldist{q}{\bm{B}_{0}^{\top}\bx_{i}}{\Omega_{0}^{-1}}$ random variable --- then, under mild assumptions and slight prior modification, the \texttt{mixed-mSSL} posterior concentrates increasingly overwhelming amounts of probability into smaller and smaller neighborhoods around $\bm{B}_{0}$ and $\Omega_{0}$ as $n$ diverges.
In this section, we establish posterior concentration rates for both $\bm{B}$ (\Cref{thm:unconditional}) and $\Omega$ (\Cref{thm:theorem3}) in the high-dimensional regime where the number of covariates $p$ and the number of outcomes $q$ grow with $n.$
We further demonstrate that as long as the non-zero entries in $\bm{B}_{0}$ are not too small, combining \texttt{mixed-mSSL} with a simple thresholding rule recovers all non-zero regression coefficients with high probability.

Like \citet{RockovaGeorge2018_ssl} and \citet{shen2022posterior}, for the theory, we do not place hyperpriors on the mixture weights or scale parameters; instead, we treat $(\theta,\eta,\lambda_0,\lambda_1)$ as deterministic sequences of $(n,p,q)$. This modification simplifies the theoretical analysis without substantially altering the essence of the result.
Insofar as our proposed latent variable model for mixed-type outcomes is similar to \citet{wang2023twostep}'s, our theoretical analysis is similar to theirs.
However, that work was limited to the regime in which $p$ grew with $n$ but $q$ remained fixed. 
Our analysis is much more general and allows $q$ to grow polynomially with the sample size.
Henceforth, we use $p_{n}$ and $q_{n}$ to denote the number of covariates and outcomes.
We similarly include the subscript $n$ in our notation for the parameters $\bm{B}$ and $\Omega$ and the matrices $\bX$, $\bY,$ and $\bZ$ to indicate that their dimension depends on $n.$
For each $j = 1, \ldots, p_{n},$ we denote the $j$-th row of $\bm{B}_{n}$ as $\bm{b}_{j}.$
Finally, we use $\P_{0}$ and $\E_{0}$ to denote probabilities and expectations taken with respect to the distribution of $\bY_{n}$ implied by our model with $\bm{B} = \bm{B}_{0}$ and $\Omega = \Omega_{0}.$ 

\subsection{Posterior contraction of $B$}
\label{sec:postcontB}
To establish the posterior concentration for $\bm{B},$ we make the following assumptions:
\begin{itemize}
\item[(A1)]{$(\log p_n)^2=o(n)$, $q_n=n^b,$ and $\log p_n\ge C\,q_n\log n$ for constants $b\in(0,1/2)$ and $C>0$.} 
\item[(A2)]{Let $S_{0} = \{j: \text{row $j$ of $\bm{B}_{0}$ is non-zero}\}$ be the indices of $\bm{B}_{0}$'s non-zero rows. Then $s_{0}^{B} = \lvert S_{0} \rvert \geq 1$ with $s_{0}^{B} = o(n/\log p_{n}).$}
\item[(A3)]{There are constants $0<\underline{\tau}<\bar\tau<\infty$ not depending on $n$ such that $\lVert \bX_{n} \rVert_{\infty} := \max_{ij}{\lvert x_{ij} \rvert } \leq \bar\tau$ and } $$ \inf_{n}\min_{\substack{S\subset\{1,\dots,p_n\}\\1\le|S|<n}}\lambda_{\min}\left(n^{-1}\sum_{i=1}^n \bx_{i}^{S}(\bx_{i}^{S})^{\top} \right) > \underline{\tau}. $$
\item[(A4)]{$\lVert \bm{B}_{0,S_0} \rVert_1 = o(n \epsilon_n ^2)$.} 
\item[(A5)]{There is a constant $k_{1}$ not depending on $n$ such that $\Omega_{0}$'s eigenvalues lie in $[k_{1}^{-1}, k_{1}].$}
\item[(A6)]{Fix $\lambda_{1} > 0$ and set $\lambda_0^2\ge C_0\,p_n q_n^2 n^{1+\rho}$ and
$\theta\le C_1\,n^{-(1+\rho)}/(p_n q_n^2)$ for constants $C_0,C_1>0$ and $0<\rho<b$.}
\end{itemize}
Assumption (A1) allows $p_{n}$ to grow sub-exponentially with $n$ and for $q_{n}$ to grow polynomially with $n$ but slower than $p_{n}.$
Since the quantity $s_{0}^{B}$ in Assumption (A2) lower bounds the total number of non-zero entries in $\bm{B}_{0},$ (A2) can be viewed as a restriction on the overall sparsity of $\bm{B}_{0}.$
Assumptions (A3)--(A5) are routinely invoked in the literature  \citep[e.g.][]{Castillo2015, BAIGhoshJMVA} and ensure parameter identifiability in the high-dimensional regime we study. 
Assumption (A4) keeps the total active signal small relative to the information scale $n\varepsilon_n^2$, ensuring non-negligible prior mass near $\bm B_0$. 
Assumption (A6) fixes the slab penalty $\lambda_1$ to avoid overshrinking true signals, while requiring the spike penalty $\lambda_0$ to increase with the problem's dimensions. This ensures that null coefficients are aggressively shrunk toward zero as the number of parameters grows. Simultaneously, following the same idea as in \citet{RockovaGeorge2018_ssl}, the prior inclusion probability $\theta$ must decay with $n$ to prevent spurious slab assignments. Without this calibration, the prior would allocate excessive mass to non-zero coefficients, causing the procedure to over-select.

\Cref{lemma:lemmawang}, which shows that the \emph{prior} for $\bm{B}$ concentrates, is adapted from Lemma D1 in \citet{wang2023twostep} and we prove it in \switchref{\Cref{sec:prooflemwang}}{Section S1.4 in the Supplementary Materials}.
\begin{lemma}[Prior concentration result for $\bm{B}$]
\label{lemma:lemmawang}
Under Assumptions (A1)--(A6), $\P(\lVert \bm{B}_{n} - \bm{B}_{0}\rVert_{F} < Cn^{-\frac{\rho}{2}}\epsilon_{n}) > \exp\{-Dn\epsilon_{n}^{2}\},$ where $\epsilon_{n}^{2} = n^{-1}(s_0 ^{B} q_n\log p_{n})$ and $C, \rho, D > 0$ are constants not depending on $n.$
\end{lemma}
This result, which is a foundational part of our analysis, shows that as $n \rightarrow \infty,$ the prior concentrates increasing amounts of probability mass in a Frobenius-norm ball around the true $\bm{B}_{0}$ of vanishing radius.
\Cref{thm:unconditional} shows that the rate $\epsilon_{n} = \left(n^{-1}\left\{s_{0}^{B} q_{n}\log p_{n} \right \}\right)^{\frac{1}{2}}$ in \Cref{lemma:lemmawang} is also the posterior contraction rate.

\begin{theorem}[Posterior unconditional contraction result for $\bm{B}$]
\label{thm:unconditional}
Under Assumptions (A1)--(A6), for any constant $M > 0,$ as $n \rightarrow \infty,$ we have
$
\sup_{\bm{B}_{0}} \E_{0}\left[\P\left( \lVert \bm{B}_{n} - \bm{B}_{0} \rVert_{F} > M\epsilon_{n} \vert \bY_{n} \right)\right] \rightarrow 0,
$
where $\epsilon_{n}^{2} =  n^{-1}(s_0 ^{B} q_n\log p_{n}).$
\end{theorem}
We prove \Cref{thm:unconditional} in \switchref{\Cref{sec:proofthmunconditional}}{Section S1.6 of the Supplementary Materials} in two steps. 
First, we show that the posterior concentrates conditionally given the data $\bY_{n}$ \emph{and the latent variable $\bZ_{n}$} (\switchref{\Cref{lem:postconcBgivenZ}}{Lemma S1.1})\footnote{Though weaker than traditional posterior contraction, this result is similar in spirit to Theorem 3 of \citet{ZitoMiller2024}, which also establishes posterior contraction for non-negative matrix factorization conditionally on an auxiliary latent variable.}. 
Then, using a standard ``good-set'' argument \citep[Ch.\ 1][]{Ash2000}, we show that \switchref{\Cref{lem:postconcBgivenZ}}{Lemma S1.1} holds uniformly for all $\bZ_n$ in a set $\mathcal Z_n$ that excludes extreme or pathological latent draws, and that the posterior places asymptotically all its mass on $\mathcal Z_n$ (\switchref{\Cref{lem:goodset}}{Lemma S1.2}).
Integrating out $\bZ_n$ then yields the unconditional result in \Cref{thm:unconditional}; see \switchref{\Cref{sec:proofpostconcBgivenZ,sec:proofthmunconditional}}{Sections S1.5 and 1.6 in the Supplementary Materials for details}.
 
Before proceeding, we pause to compare our $\epsilon^{2}_{n}$ to other rates reported in the literature for similar problems. When $q_n$ is fixed, $\epsilon_n$ coincides with the familiar single–response rate up to a constant factor $q_n$, namely
$\epsilon_n^2 \asymp {s_0^{B}\log p_n}/{n}$, which matches with Theorem~3.1 of \citet{wang2023twostep}.
More generally, our rate is intuitive: it is $q_n$ times the posterior contraction rate for a single column (the case $q_n=1$), reflecting that under Frobenius loss one must estimate $s_0^{B}$ active rows across all $q_n$ responses.

\citet{shen2022posterior} considered a special case with only continuous responses and obtained the squared rate
\(n^{-1}\max\{q_n,s_0^B\}\log\big(\max\{p_n,q_n\}\big)\).
Under our growth regime (A1), $p_n$ may grow much faster than $q_n$, so $\log(\max\{p_n,q_n\})=\log p_n$; in that setting their rate simplifies to
\(n^{-1}\max\{q_n,s_0^B\}\log p_n\).
Our product rate $n^{-1}s_0^{B}q_n\log p_n$ reduces to the same order when either $q_n$ or $s_0^{B}$ is fixed, but is larger when both $q_n$ and $s_0^{B}$ diverge. This difference is natural in our setup: we impose element–wise spike–and–slab priors (rather than row–grouped shrinkage), target Frobenius loss, and allow mixed outcomes. In this configuration, the effective complexity is the number of active coefficients, $s_0^{B}q_n$, and the Frobenius risk aggregates error across all $q_n$ responses, leading to the multiplicative factor $q_n$ in the target rate.


\subsection{Posterior contraction of $\Omega$}
\label{sec:omega_concentration}
Our analysis of the \texttt{mixed-mSSL} posterior for $\Omega$ requires additional assumptions and modifications to the prior.
Specifically, in addition to assuming that the true precision matrix has eigenvalues bounded away from zero and infinity, we additionally assume that it has at most $s_{0}^{\Omega}$ non-zero off-diagonal entries in its upper triangle.
That is, we assume $\Omega_{0}$ lies in the space
$$
\mathcal{U}(k_{1}, s_{0}^{\Omega}) = \left\{\Omega \in \mathcal{M}^{+}_{q_{n}}(k_{1}): \sum_{k<k'}{\ind{\omega_{k,k'} \neq 0}} \leq s_{0}^{\Omega}\right\},
$$
for some $s_{0}^{\Omega} \geq 1$ where $k_{1}$ is as in Assumption (A5) and for any $\lambda > 0$ the space $\mathcal{M}^{+}_{q_{n}}(\lambda)$ is defined as
$$
\mathcal{M}^{+}_{q_{n}}(\lambda) = \{\Omega \succ 0: \lambda^{-1} \leq \lambda_{\min}(\Omega) \leq \lambda_{\max}(\Omega) < \lambda\}.
$$
We modify the prior over $\Omega$ so that the initial element-wise prior is truncated to the space $\mathcal{M}^{+}_{q_{n}}(k_{1})$ instead of the whole positive definite cone.
We additionally assume the following.
\begin{itemize}
\item[(B1)]{$q_{n} = n^{b}$ for some $b \in (0, 1/2)$ and $n^{-1}(q_{n} + s_{0}^{\Omega})\log q_{n} = o(1).$}
\item[(B2)]{$\eta \asymp  {\log q_n}/{n q_n}$.}
\item[(B3)]{$\xi_0^2 \gtrsim (\log q_n)^{-1}nq_n$.}
\end{itemize}
We believe that Assumption (A1) and (B1)'s assumption that $q_{n}$ grows slower than $\sqrt{n}$ is critical.
Indeed, arguments in \citet{sarkar2024posteriorconsistencymultiresponseregression} suggest that relaxing the assumption to allow $q_{n}/n \rightarrow \gamma < \infty$ would lead to posterior inconsistency. 
Assumption (B2) ensures the \texttt{mixed-mSSL} prior puts sufficient mass around the true zero elements in the precision matrix. Assumption (B3) guarantees that the vast majority of coefficients corresponding to structural zeros (i.e., off-diagonal entries that are zero in the true $\Omega_0$) are aggressively shrunk toward zero. It is worthwhile to note that enforcing $\sigma^{2}_{k} = 1$ for $k = q_{c}+1, \ldots, q$ does not invalidate Assumption (A5) since diagonal inflation (or deflation) by a bounded factor only shifts
the spectrum within a bounded interval.
 
\begin{theorem}[Posterior unconditional contraction result for $\Omega$]
\label{thm:theorem3}
Under Assumptions (A5), (B1), (B2), and (B3), for a sufficiently large constant $M > 0,$ $\sup_{\Omega_{0}}\E_{0}\P(\lVert \Omega_{n} - \Omega_{0} \rVert_{F} > M\tilde{\epsilon}_{n} \vert \bY_{n}) \rightarrow 0,$ where $\tilde{\epsilon}_{n}^{2} = n^{-1}(q_{n} + s_{0}^{\Omega})\log q_{n}$ and $s_{0}^{\Omega}$ denote the number of non-zero off-diagonal elements in the upper triangle of $\Omega.$
\end{theorem}

We prove \Cref{thm:theorem3} in \switchref{\Cref{sec:proofthm3}}{Section S1.7 in the Supplementary Materials} using the same strategy with which we proved \Cref{thm:unconditional}: first, we establish conditional contraction given $\bZ_{n}$ (\switchref{\Cref{lem:postconcomega}}{Lemma S1.3}) and then use another ``good-set'' argument to establish unconditional contraction.
Our rate $\tilde{\epsilon}_{n}$ matches the optimal rates for estimating the sparse precision matrices established in \citet[Thm.~3.1]{Banerjee2015} and \citep[Thm.~2]{ZHANG2022154}.

The rates $\epsilon_n$ and $\tilde{\epsilon}_n$ are compatible, reflecting the effective dimension of each parameter block: the rate $\epsilon_n$ for $\bm{B}$ is driven by the total number of non-zero coefficients ($\approx s_0^B q_n$), while the rate $\tilde\epsilon_n$ for $\Omega$ is driven by the size and sparsity of the precision matrix ($q_n+s_0^\Omega$). In ultra-sparse settings, both rates scale similarly with the number of outcomes $q_n$.

\subsection{Variable selection consistency}
\label{sec:surescreening}
Although \Cref{thm:unconditional} implies that we can consistently estimate $\bm{B}_{0}$ with \texttt{mixed-mSSL}, it does not speak to our ability to recover $\bm{B}_{0}$'s support (i.e., the set $S_{0} = \{(j,k): (\bm{B}_{0})_{j,k} \neq 0\}$).
Because we use absolutely continuous priors, the \texttt{mixed-mSSL} will generally place all of its posterior probability on dense matrices $\bm{B}_{n},$ all of whose entries are non-zero.
Similar to \citet{song2022nearly} and \citet{WEI2020215} and inspired by \citet{RockovaGeorge2018_ssl}'s notion of generalized dimension, we can threshold the entries of $\bm{B}_{n}$ to estimate $S_{0}.$
Specifically, given any $\bm{B}_{n}$ and $\Omega_{n},$ let $\widetilde{\bm{B}}_{n}$ be the $p_{n} \times q_{n}$ matrix obtained by thresholding $\bm{B}_{n}$ element-wise as $\widetilde{\beta}_{j,k} = \beta_{j,k} \ind{\lvert \beta_{j,k} \rvert > a_{n}\omega_{k,k}},$ where $a^{2}_{n} = (c^{2}\log p_{n})/(np_{n}^{2})$ for some constant $c > 0.$
Let $\tilde{S}_{n} = \{(j,k): \tilde{\beta}_{j,k} \neq 0\}$ be the support of $\widetilde{\bm{B}}_{n}.$

In order to study the posterior distribution of $\tilde{S}_{n}$, we make the following assumption: 
\begin{itemize}
\item[(C1)]{There are constants $c_{3} > 0$ and $0 < \zeta < 1/4,$ such that $\lvert (\bm{B}_{0})_{j,k} \rvert > c_{3}n^{-\zeta}$ for all $(j,k) \in S_{0}.$}
\end{itemize}

\begin{theorem}[Sure screening via hard thresholding]
\label{thm:surescreening}
Under Assumptions (A1), (A2), (B1), and (C1), as $n \rightarrow \infty,$ $\sup_{\bm{B}_{0}}\E_{0}\P(S_{0} \subset \tilde{S}_{n} \vert \bY_{n}) \rightarrow 1.$
\end{theorem}
Our posterior contraction rate is 
\(
\epsilon_n \asymp n^{(b-1)/2}\sqrt{\log p_n},
\)
and since \(b<1/2\), we have \((b-1)/2<-1/4\). 
Hence, any signal of order
\(n^{-\zeta}\) with \(\zeta<\tfrac14\) dominates the estimation error:
\(n^{-\zeta}\gg \epsilon_n\).
This ``beta–min" dominance is the standard requirement for consistent recovery
\citep{Zhang2008,ZhaoYu2006}, which motivates Assumption (C1). 
We prove \Cref{thm:surescreening} in \switchref{}{Supplementary Section S1.8}.

We emphasize that \Cref{thm:surescreening} does not guarantee perfect model selection, as it allows for the possibility of false positives.
However, the sure screening property provides a strong theoretical guarantee about the number of false negatives that \texttt{mixed-mSSL} can make, provided that the entries of $\bm{B}_{0}$ are large enough.

%% file: short_experiments.tex
Using several synthetic datasets of varying dimensions, we compared \texttt{mixed-mSSL} to three other methods in terms of estimating the supports of matrices $\bm{B}$ and $\Omega.$ 
We compared \texttt{mixed-mSSL} to \citet{wang2023twostep}'s \texttt{mt-MBSP}, and two baseline methods, \texttt{sepSSL} and \texttt{sepglm}. 
Like \texttt{mixed-mSSL}, \texttt{mt-MBSP} is based on a latent, sparse multivariate normal regression model.
However, while \texttt{mixed-mSSL} models binary outcomes with a probit link, \texttt{mt-MBSP} is based on a logistic link (see Equation (2.2) of \citet{wang2023twostep}). 
\texttt{mt-MBSP} also specifies element-wise horseshoe priors \citep{CarvalhoHS} on the matrix of regression coefficients and an inverse Wishart prior on the residual latent precision matrix. 

The first baseline method \texttt{sepSSL} fits a \texttt{SSLASSO} model to each outcome separately, using the standard continuous-outcome version of \citet{RockovaGeorge2018_ssl}'s spike-and-slab LASSO for continuous outcomes and a probit version of spike-and-slab LASSO for binary outcomes. 
The second baseline method, \texttt{sepglm}, separately fits penalized generalized linear models (GLMs) to each outcome, using a Gaussian model with LASSO penalty for continuous outcomes, and a binomial probit model with LASSO penalty for binary outcomes, both implemented using \texttt{glmnet} \citep{glmnet2010}. 
Because the software implementation works only when $n > p,$ we did not include \citet{ekvall}'s \texttt{mmrr} procedure in our simulation study, which focuses on the more challenging $p > n$ regime. 

We considered three problem dimensions $(n,p,q) \in\{(200,500,4),(500,1000,4),(800,1000,6)\}$ with $q_b=q_c=q/2$. 
Following \citet{shencgssl}, we constructed six $\Omega$ structures per $(n,p,q)$: (i) AR(1) for $\Omega^{-1}$ (tri-diagonal $\Omega$); (ii) AR(2) for $\Omega^{-1}$ ($\omega_{kk'}=0$ if $|k-k'|>2$); (iii) block-diagonal with two dense $q/2\times q/2$ blocks; (iv) star graph ($\omega_{k,k'} \neq 0$ only when $k=1$ or $k'=1$); (v) an $\Omega$ whose support corresponds to a small-world network; (vi) an $\Omega$ whose support corresponds to a tree network.
We defer full details about constructing such $\Omega$'s to \switchref{\Cref{sec:covstr}}{Section S3.1 in the Supplementary Materials}. 
We drew each row of $\bm{X}$ independently from a $\mvnormaldist{p}{\mathbf{0}}{\Gamma}$ distribution with $\Gamma_{jj'}=0.5^{|j-j'|}$. 
For each $(n,p,q)$, we generated two different $\bm{B}$'s, each with 70\% sparsity: one with non-zero entries  drawn uniformly from $[-5,5]$ and the other with non-zero entries drawn uniformly from $[-5,-2]\cup[2,5].$
In total, we had 36 combinations of problem dimension $(n,p,q),$ $\bm{B},$ and $\Omega.$
For each combination, we generated 100 synthetic datasets. 

\texttt{mixed-mSSL}, \texttt{sepSSL}, and \texttt{sepglm} return estimates of the matrix $\bm{B}$ from our latent variable model.
Because \texttt{mt-MBSP} models binary outcomes with a logistic link, the scale and interpretation of its corresponding matrix of regression coefficients differ from those for our $\bm{B}.$
Due to these differences, we do not present parameter estimation errors.
Instead, we only report support recovery metrics (e.g., sensitivity, specificity, precision, and accuracy) for recovering $\bm{B}$'s and $\Omega$'s supports, which are scale-invariant.
We also evaluated predictive performance on a held-out test set of size $n/2$: regression-function error (\texttt{RFE}), predictive root mean square error (\texttt{RMSE}; averaged over continuous outcomes), and the average area under the receiver operator characteristic curve (\texttt{AUC}; binary outcomes). 
\texttt{RFE} is the average $\ell_2$ distance between the true conditional means $\E[\bm Y \vert \bm X]$ (computed at test covariates under the true parameters) and each method’s predicted means. 
Predictive \texttt{RMSE} is obtained by simulating test outcomes from the true model, generating predictions from each fitted model, and comparing them to the realized test responses. 
\texttt{AUC} compares predicted probabilities with binary outcomes simulated from the true data-generating process.

We ran \texttt{mixed-mSSL} using the suggested hyperparameters from \Cref{sec:mcecm} and took the MAP estimate corresponding to the largest spike penalties as our point estimates of $\bm{B}$ and $\Omega.$
For \texttt{sepSSL}, we assembled $\hat{\bm B}$ column-wise from the outcome-wise MAPs at the largest spikes using the \citet{RockovaGeorge2018_ssl} defaults.
We built the \texttt{sepglm} point estimate of $\bm{B}$ column-by-column by running cross-validated LASSO models using \textbf{glmnet} package defaults. 
For \texttt{mt-MBSP}, we created a point estimate $\bm{B}$'s support by checking whether marginal 95\% posterior credible intervals of each \texttt{mt-MBSP} regression coefficient excluded zero. 
We ran \texttt{mt-MBSP} for 1100 iterations, which are the recommended package default settings. 

\Cref{table:p200q4unif} summarizes results for $(n,p,q)=(200,500,4)$ with non-zero elements of $\bm{B}$ being drawn uniformly from $[-5,5]$; best values are bolded. 
Across covariance structures, \texttt{mixed-mSSL} delivers the strongest predictive performance (lowest \texttt{RFE} and \texttt{RMSE}, highest \texttt{AUC}). 
\texttt{sepglm} attains the highest \texttt{SEN} but often lower \texttt{SPEC} and \texttt{ACC} than \texttt{mixed-mSSL} and \texttt{mt-MBSP}. 
\texttt{mixed-mSSL} and especially \texttt{mt-MBSP} show high \texttt{SPEC} and \texttt{PREC}, reflecting conservative, reliable detection. 
\texttt{sepSSL} is fastest but has lower \texttt{PREC} and \texttt{SPEC} (more false positives). We note that \texttt{mixed-mSSL}'s low sensitivity in this setting is an artifact of the small sample size, as it has limited power to detect weak signals near zero (\switchref{fig:ARsensplot}{Figure S3.1 in the Supplementary Materials}).

\begin{table}[ht!]
{\small
\centering
\caption{Support recovery and predictive performance across different covariance structures with $(n,p,q)=(200,500,4)$ under the signal setting $\mathcal{U}[-5,5]$.}
\label{table:p200q4unif}
\begin{tabular}{@{}l l cccc c ccc@{}}
\toprule
\multirow{2}{*}{\textbf{Scenario}} & \multirow{2}{*}{\textbf{Method}}
 & \multicolumn{4}{c}{\textbf{Support recovery}} 
 & 
 & \multicolumn{3}{c}{\textbf{Predictive}} \\
\cmidrule(lr){3-7} \cmidrule(lr){8-10}
  &   & \textbf{SEN} & \textbf{SPEC} & \textbf{PREC} & \textbf{ACC} 
  & \textbf{TIME(s)} & \textbf{RFE} & \textbf{RMSE} & \textbf{AUC} \\
\midrule
\multirow{4}{*}{AR1}
  & \texttt{mt-MBSP}     & 0.10 & \textbf{0.97} & \textbf{0.63} & \textbf{0.71} & 205.13 & 43.89 & 1.78 & 0.58 \\
  & \texttt{sepSSL}      & 0.31 & 0.83 & 0.45 & 0.68 &  \textbf{0.37} & 43.69 & 1.14 & 0.53 \\
  & \texttt{sepglm}      & \textbf{0.40} & 0.81 & 0.48 & 0.68 & 1.74 & 43.55 & 0.71 & 0.58 \\
  & \texttt{mixed-mSSL}  & 0.21 & 0.91 & 0.52 & 0.70 & 18.26 & \textbf{43.46} & \textbf{0.69} & \textbf{0.61} \\
\midrule
\multirow{4}{*}{AR2}
  & \texttt{mt-MBSP}     & 0.10 & \textbf{0.97} & \textbf{0.62} & 0.71 & 146.10 & 43.89 & 1.78 & 0.58 \\
  & \texttt{sepSSL}      & 0.31 & 0.84 & 0.44 & 0.67 &  \textbf{0.40} & 43.70 & 1.15 & 0.52 \\
  & \texttt{sepglm}      & \textbf{0.41}  & 0.81 & 0.48 & 0.68 & 1.57 & 43.56 & 0.71 & 0.57 \\
  & \texttt{mixed-mSSL}  & 0.21 & 0.92 & 0.52 & \textbf{0.72} & 24.62 & \textbf{43.47} & \textbf{0.70} & \textbf{0.61} \\
\midrule
\multirow{4}{*}{Block Diagonal}
  & \texttt{mt-MBSP}     & 0.11 & \textbf{0.97} & \textbf{0.63} & \textbf{0.71} & 198.61 & 43.89 & 1.77 & 0.58 \\
  & \texttt{sepSSL}      & 0.31 & 0.83 & 0.44 & 0.67 & \textbf{0.36} & 43.69 & 1.14 & 0.53 \\
  & \texttt{sepglm}      & \textbf{0.40} & 0.81 & 0.48 & 0.69 & 1.71 & 43.56 & 0.72 & 0.58 \\
  & \texttt{mixed-mSSL}  & 0.21 & 0.91 & 0.53 & \textbf{0.71} & 21.53 & \textbf{43.47} & \textbf{0.69} & \textbf{0.60} \\
\midrule
\multirow{4}{*}{Star Graph}
  & \texttt{mt-MBSP}     & 0.10 & \textbf{0.98} & \textbf{0.63} & \textbf{0.71} & 186.66 & 43.89 & 1.78 & 0.58 \\
  & \texttt{sepSSL}      & 0.32 & 0.84 & 0.45 & 0.67 &  \textbf{0.34} & 43.69 & 1.14 & 0.53 \\
  & \texttt{sepglm}      & \textbf{0.41} & 0.81 & 0.48 & 0.68 & 1.63 & 43.55 & 0.71 & 0.58 \\
  & \texttt{mixed-mSSL}  & 0.22 & 0.91 & 0.52 & \textbf{0.71} & 22.40 & \textbf{43.47} & \textbf{0.70} & \textbf{0.61} \\
\midrule
\multirow{4}{*}{Small World}
  & \texttt{mt-MBSP}     & 0.11 & \textbf{0.97} & \textbf{0.63} & \textbf{0.71} & 183.06 & 43.89 & 1.78 & 0.58 \\
  & \texttt{sepSSL}      & 0.31 & 0.83 & 0.44 & 0.67 &  \textbf{0.33} & 43.69 & 1.14 & 0.53 \\
  & \texttt{sepglm}      & \textbf{0.40} & 0.81 & 0.48 & 0.68 & 1.58 & 43.55 & 0.71 & 0.58 \\
  & \texttt{mixed-mSSL}  & 0.21 & 0.91 & 0.54 & 0.70 & 28.97 & \textbf{43.47} & \textbf{0.70} & \textbf{0.61} \\
\midrule
\multirow{4}{*}{Tree Network}
  & \texttt{mt-MBSP}     & 0.10 & \textbf{0.97} & \textbf{0.62} & \textbf{0.71} & 189.30 & 43.89 & 1.78 & 0.58 \\
  & \texttt{sepSSL}      & 0.31 & 0.83 & 0.45 & 0.67 &  \textbf{0.34} & 43.69 & 1.14 & 0.53 \\
  & \texttt{sepglm}      & \textbf{0.41} & 0.81 & 0.48 & 0.69 & 1.63 & 43.55 & 0.71 & 0.58 \\
  & \texttt{mixed-mSSL}  & 0.22 & 0.91 & 0.53 & \textbf{0.71} & 24.40 & \textbf{43.47} & \textbf{0.70} & \textbf{0.61} \\
\bottomrule
\end{tabular}
}
\end{table}

\begin{table}[ht!]
{\small
\centering
\caption{Support recovery and predictive performance across different covariance structures with $(n,p,q)=(200,500,4)$ under the disjoint signal setting $\mathcal{U}[-5,-2] \cup [2,5]$.}
\label{table:p200q4disjt}
\begin{tabular}{@{}l l cccc c ccc@{}}
\toprule
\multirow{2}{*}{\textbf{Scenario}} & \multirow{2}{*}{\textbf{Method}}
 & \multicolumn{4}{c}{\textbf{Support recovery}} 
 & 
 & \multicolumn{3}{c}{\textbf{Predictive}} \\
\cmidrule(lr){3-7} \cmidrule(lr){8-10}
  &   & \textbf{SEN} & \textbf{SPEC} & \textbf{PREC} & \textbf{ACC} 
  & \textbf{TIME(s)} & \textbf{RFE} & \textbf{RMSE} & \textbf{AUC} \\
\midrule
\multirow{4}{*}{AR1}
  & \texttt{mt-MBSP}     & 0.10 & \textbf{0.96} & \textbf{0.57} & \textbf{0.70} & 133.44 & 54.94 & 1.78 & 0.58 \\
  & \texttt{sepSSL}      & 0.31 & 0.83 & 0.44 & 0.67 &  \textbf{0.34} & 54.79 & 1.17 & 0.52 \\
  & \texttt{sepglm}      & \textbf{0.39} & 0.81 & 0.48 & 0.69 & 1.45 & 54.70 & 0.80 & 0.58 \\
  & \texttt{mixed-mSSL}  & 0.35 & 0.86 & 0.51 & \textbf{0.70} & 31.49 & \textbf{54.57} & \textbf{0.77} & \textbf{0.60} \\
\midrule
\multirow{4}{*}{AR2}
  & \texttt{mt-MBSP}     & 0.11 & \textbf{0.96} & \textbf{0.58} & 0.70 & 151.87 & 54.94 & 1.78 & 0.57 \\
  & \texttt{sepSSL}      & 0.32 & 0.83 & 0.45 & 0.67 &  \textbf{0.41} & 54.79 & 1.17 & 0.53 \\
  & \texttt{sepglm}      & \textbf{0.39}  & 0.81 & 0.48 & 0.69 & 1.70 & 54.70 & 0.80 & 0.57 \\
  & \texttt{mixed-mSSL}  & 0.36 & 0.85 & 0.50 & \textbf{0.71} & 19.95 & \textbf{54.57} & \textbf{0.78} & \textbf{0.61} \\
\midrule
\multirow{4}{*}{Block Diagonal}
  & \texttt{mt-MBSP}     & 0.10 & \textbf{0.96} & \textbf{0.57} & \textbf{0.70} & 136.08 & 54.94 & 1.78 & 0.58 \\
  & \texttt{sepSSL}      & 0.31 & 0.83 & 0.45 & 0.68 & \textbf{0.35} & 54.79 & 1.17 & 0.52 \\
  & \texttt{sepglm}      & \textbf{0.39} & 0.82 & 0.48 & 0.69 & 1.47 & 54.70 & 0.80 & 0.57 \\
  & \texttt{mixed-mSSL}  & 0.35 & 0.85 & 0.49 & 0.69 & 29.22 & \textbf{54.58} & \textbf{0.78} & \textbf{0.60} \\
\midrule
\multirow{4}{*}{Star Graph}
  & \texttt{mt-MBSP}     & 0.10 & \textbf{0.96} & \textbf{0.57} & 0.70 & 124.39 & 54.94 & 1.78 & 0.58 \\
  & \texttt{sepSSL}      & 0.31 & 0.83 & 0.45 & 0.68 &  \textbf{0.34} & 54.80 & 1.18 & 0.53 \\
  & \texttt{sepglm}      & \textbf{0.38} & 0.82 & 0.48 & 0.69 & 1.40 & 54.69 & 0.80 & 0.58 \\
  & \texttt{mixed-mSSL}  & 0.35 & 0.86 & 0.52 & \textbf{0.71} & 30.77 & \textbf{54.58} & \textbf{0.78} & \textbf{0.61} \\
\midrule
\multirow{4}{*}{Small World}
  & \texttt{mt-MBSP}     & 0.10 & \textbf{0.96} & \textbf{0.57} & \textbf{0.70} & 204.06 & 54.94 & 1.78 & 0.58 \\
  & \texttt{sepSSL}      & 0.31 & 0.83 & 0.44 & 0.67 &  \textbf{0.36} & 54.79 & 1.17 & 0.52 \\
  & \texttt{sepglm}      & \textbf{0.39} & 0.81 & 0.48 & 0.69 & 1.74 & 54.71 & 0.81 & 0.57 \\
  & \texttt{mixed-mSSL}  & 0.34 & 0.85 & 0.51 & 0.69 & 26.23 & \textbf{54.58} & \textbf{0.78} & \textbf{0.60} \\
\midrule
\multirow{4}{*}{Tree Network}
  & \texttt{mt-MBSP}     & 0.11 & \textbf{0.97} & \textbf{0.58} & \textbf{0.70} & 196.86 & 54.94 & 1.78 & 0.57 \\
  & \texttt{sepSSL}      & 0.32 & 0.83 & 0.44 & 0.67 &  \textbf{0.35} & 54.80 & 1.17 & 0.53 \\
  & \texttt{sepglm}      & \textbf{0.39} & 0.82 & 0.49 & 0.69 & 1.70 & 54.70 & 0.80 & 0.57 \\
  & \texttt{mixed-mSSL}  & 0.35 & 0.84 & 0.53 & \textbf{0.70} & 23.64 & \textbf{54.57} & \textbf{0.77} & \textbf{0.61} \\
\bottomrule
\end{tabular}
}
\end{table}

\Cref{table:p200q4disjt} summarizes estimation and predictive performance results for signals drawn from the disjoint interval $\mathcal{U}[-5,-2] \cup [2,5]$. This setting is relatively less challenging compared to the previous scenario, where signals were typically concentrated around zero. Consequently, we observe that \texttt{mixed-mSSL}'s sensitivity (\texttt{SEN}) notably improves. 
\texttt{mt-MBSP} consistently achieves high specificity (\texttt{SPEC}) and precision (\texttt{PREC}), albeit with relatively lower sensitivity.
On the other hand, \texttt{sepSSL} and \texttt{sepglm} exhibit higher sensitivity but at the cost of lower specificity and precision, indicating more frequent false positives.
\texttt{mixed-mSSL} demonstrates an optimal balance, substantially improving sensitivity without a large sacrifice in specificity, precision, or accuracy (\texttt{ACC}). 
Specifically, it attains competitive predictive performance metrics across all covariance structures, evidenced by consistently low regression-function error (\texttt{RFE}), predictive RMSE, and the highest mean area under the ROC curve (\texttt{AUC}).

\switchref{\Cref{table:omega_support}}{Table S3.1 in the Supplementary Materials} reports \texttt{mixed-mSSL}'s  support recovery performance for $\Omega$ across the various covariance structures. 
Because \texttt{mt-MBSP} typically returns a dense $\Omega$ (thanks to its use of an inverse Wishart prior), rendering it uninformative for evaluating support recovery or variable selection performance.
Qualitatively similar results for other settings are presented in \switchref{\Cref{sec:additional_experiments}}{Section S3 of the Supplementary Materials}.

%% file: ckddata.tex
Our first application reanalyzes the CKD dataset of \citet{CKDdata} ($n=400$). Outcomes are (i) CKD status (binary) and (ii) urine specific gravity (SG, continuous). We consider $p=24$ covariates (clinical labs such as albumin, hemoglobin, random blood glucose; comorbidities such as diabetes, hypertension, anemia, CAD). \citet{CKDdata} reported 10 biomarkers predictive of CKD. We compare \texttt{mixed-mSSL}, \texttt{mt-MBSP}, \texttt{sepSSL}, and \texttt{sepglm}; although $n>p$, \texttt{mmrr} failed to converge and is omitted.

\Cref{tab:ckd_metrics} summarizes selection against the 10 gold standards. All methods recover albumin, hemoglobin, and diabetes (full lists in Table~S4.1). \texttt{mixed-mSSL} identifies 6 out of 10 with the highest precision (0.86), accuracy (0.79), and specificity (0.93), indicating the best FP/FN balance. \texttt{mt-MBSP} and \texttt{sepSSL} are more conservative (4 true biomarkers each), with \texttt{sepSSL} incurring lower specificity/precision due to extra false positives. \texttt{sepglm} selects many variables (16), achieving the highest sensitivity (0.90) but low specificity (0.50) and accuracy (0.56).

Notably, only \texttt{mixed-mSSL} and \texttt{sepglm} select random blood glucose (BGR), consistent with prior evidence linking hyperglycemia to CKD progression \citep{KumarCKD,HassaneinCKD}.

\begin{table}[ht]
\centering
\caption{Classification performance of different methods for detecting the 10 gold-standard CKD biomarkers.}
\label{tab:ckd_metrics}
\begin{tabular}{@{}lccccc@{}}
\toprule
\textbf{Method}         & \textbf{\# Selected} & \textbf{SEN} & \textbf{SPEC} & \textbf{ACC} & \textbf{PREC} \\ 
\midrule
\texttt{mixed-mSSL}     & 7           & 0.60        & 0.93        & 0.79     & 0.86      \\ 
\texttt{mt-MBSP}        & 5           & 0.40        & 0.93        & 0.71     & 0.80      \\ 
\texttt{sepSSL}         & 6           & 0.40        & 0.86        & 0.67     & 0.67      \\ 
\texttt{sepglm}         & 16          & 0.90        & 0.50        & 0.56     & 0.67      \\
\bottomrule
\end{tabular}
\end{table}

We additionally compared each method's predictive performance using five-fold cross-validation. We assessed predictive performance for CKD using the area under the receiver operating characteristic curve (AUC) and for SG using RMSE. Across the five folds, \texttt{mixed-mSSL} had the highest AUC for classifying CKD status (0.993 versus 0.972 for \texttt{mt-MBSP}, 0.967 for \texttt{sepSSL} and 0.988 for \texttt{sepglm}) and the lowest RMSE for predicting SG ($4.33 \times 10^{-3}$ versus $4.44 \times 10^{-3}$ for \texttt{mt-MBSP}, $4.63 \times 10^{-3}$ for \texttt{sepSSL}, and $4.52 \times 10^{-3}$ for \texttt{sepglm}).

%% file: crcdata.tex
To showcase performance with $p\gg n$, we reanalyzed a pooled gut–metagenomic compendium of five CRC cohorts studied in \citet{Wirbel2019meta}. 
After harmonization, the data comprise $n=286$ subjects and $p=401$ species present in at least $20\%$ of samples. 
There are $q = 3$ outcomes: two binary indicators for CRC status and Type-2 diabetes, and the continuous body mass index (BMI).
Joint analysis of these three outcomes is biologically motivated as obesity and diabetes are well-established risk modifiers for CRC and share overlapping microbiome signatures \citep{Qin2012T2D,Mandic24}.   

Using \texttt{mixed-mSSL}, we identified five taxa as protective (i.e., the corresponding $\beta_{j,k}$'s were negative) and five as risk factors (i.e., the corresponding $\beta_{j,k}$'s were positive) for CRC.
\emph{Blautia wexlerae} and \emph{Anaerotipes hadrus} --- both prominent butyrate producers --- were identified as protective; reduced levels of these taxa have been consistently observed in CRC patients relative to healthy controls \citep{Louis2014Butyrate, Kostic2013CRC}. 
By contrast, \emph{Ruminococcus torques}, \emph{Clostridium citroniae}, \emph{Alistipes indistinctus}, and \emph{Oscillibacter} spp. were all positively associated with CRC risk by \texttt{mixed-mSSL}; prior studies have reported over‑representation of these anaerobes in tumor tissue or stool from CRC cohorts, suggesting pro‑inflammatory or genotoxic roles \citep{zeller2014potential, Liu2023}.

\texttt{mixed-mSSL} identified a single Clostridial cluster member (\emph{Clostridium} sp.) as having a nominally protective effect against Type-2 diabetes.
This finding aligns with previous murine studies showing that this species is involved in short-chain fatty acid synthesis and improved insulin sensitivity \citep{Razavi2024}. 
\texttt{mixed-mSSL} also identified \emph{Bacteroides vulgatus} and certain \emph{Clostridiales spp.} as risk factors for diabetes, consistent with metagenomic surveys that report enrichment of these taxa in hyperglycemic individuals \citep{Qin2012T2D, Karlsson2013T2D}. 
Turning to $\Omega,$ \texttt{mixed-mSSL} recovered very small negative partial correlations (-0.008) between diabetes and CRC status, suggesting that having a latent risk for one disease barely decreases the risk of the other, after accounting for the covariates and BMI.
It also identified a very small positive partial correlation (0.011) between BMI and diabetes and a small negative partial correlation between BMI and CRC risk. 

Compared to \texttt{mixed-mSSL}, \texttt{sepSSL} identified far fewer taxa (1) as predictive of CRC and diabetes status, while \texttt{mt-MBSP} and \texttt{sepglm} identified many more: (72) and (90) respectively. 
Because we lacked a gold standard set of significant taxa against which to assess these identifications, we performed another predictive comparison using five-fold cross-validation. 
Similar to our analysis of the CKD data, we assessed each method's ability to predict binary outcomes using AUC and continuous outcomes using RMSE.
For discriminating between CRC cases and controls, \texttt{mixed-mSSL} had the highest AUC (0.774 compared to 0.722 for \texttt{mt-MBSP},  0.601 for \texttt{sepSSL}, and 0.693 for \texttt{sepglm}).
It similarly performed the best at predicting diabetes status, achieving an AUC of 0.614 compared to \texttt{mt-MBSP}'s 0.585, \texttt{sepSSL}'s 0.529, and \texttt{sepglm}'s 0.566 and BMI, displaying a much lower RMSE (4.133) than \texttt{mt-MBSP} (4.813), \texttt{sepSSL} (19.511), and \texttt{sepglm} (6.577).
Taken together, these results suggest that \texttt{mixed-mSSL} was better able to share information across related outcomes than \texttt{mt-MBSP}, yielding better predictions with fewer covariates.

%% file: hmscdata.tex
We finally demonstrate \texttt{mixed-mSSL}'s utility in settings with only binary outcomes by re-analyzing a dataset containing yearly counts of the $q = 50$ most common Finnish breeding birds present at $n = 137$ locations that were originally compiled by \citet{Lindstrom15} and later analyzed by \citet{Ovaskainen2020}. 
We focus on a single survey year (2014) and let $Y_{i,k} = 1$ if the number of birds of species $k$ observed at site $i$ exceeds that species' median count across all sites in 2013, and $Y_{i,k} = 0$ otherwise.
The dataset contains measurements of $p = 4$ covariates at each location: the average temperatures in April \& May; the effort (in hours) expended by the observer to survey the location; the average temperature in June \& July; and the average winter temperature. 
\texttt{mixed-mSSL} selected the temperature variables as the most relevant predictors of whether most species were present in greater abundance at a site in 2014 than in 2013.

More interesting, ecologically, is the estimated precision matrix $\Omega.$
Recall that latent residual dependencies captured by $\Omega$ induce dependence between the observed outcomes.
\Cref{fig:networkhmsc} shows a network whose vertices correspond to species and whose edges correspond to non-zero off-diagonal elements in the upper triangle of $\Omega.$
Edges corresponding to positive $\omega_{k,k'}$ values are colored black and suggest a positive partial correlation --- that is, the two species tend to co-occur at the same location more often than expected based on the covariates and presence of other species. 
Edges corresponding to negative $\omega_{k,k'}$ values, which indicate negative partial correlations, are colored red.
The widths of the edges in \Cref{fig:networkhmsc} are proportional to $\lvert \omega_{k,k'} \rvert$, with thicker lines indicating larger partial correlations. 
\begin{figure}[ht!]
    \centering
    \includegraphics[width=0.6\textwidth]{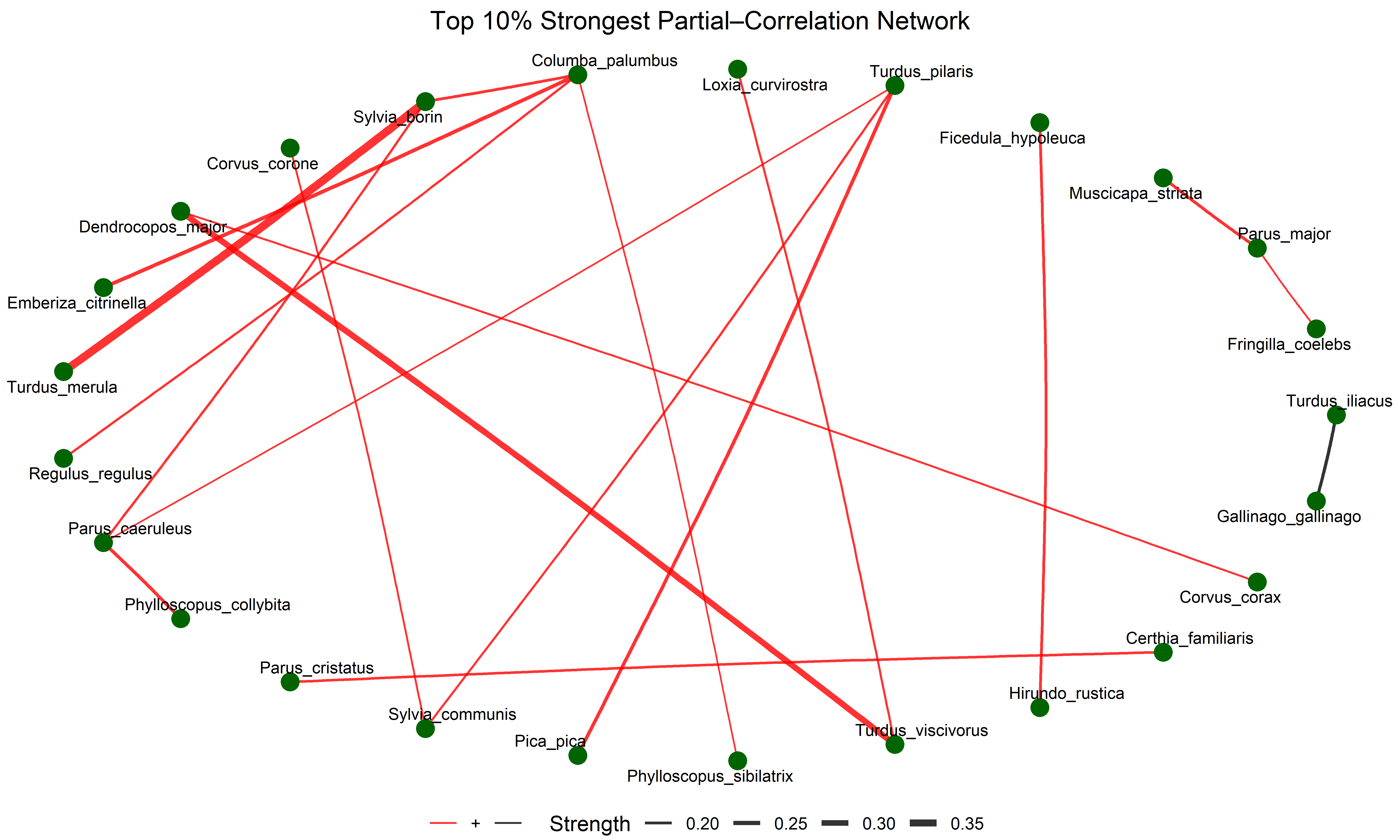}
    \caption{Covariate-adjusted partial-correlation network of bird species}
    \label{fig:networkhmsc}
\end{figure}
In \Cref{fig:networkhmsc}, we observe a single positive partial correlation between  \emph{Turdus-iliacus} (redwing) and \emph{Gallinago-gallinago} (common snipe).
These two species often share wet, rush‑dominated meadows and bog margins in boreal Finland. 
The positive edge recovered by \texttt{mixed-mSSL} plausibly reflects a real tendency for these species to co-occur after climatic effects are removed \citep{Virkkala2014}.
In contrast, the abundance of negative edges among woodland versus farmland species (e.g.,\emph{Columba-palumbus} and \emph{Sylvia-borin} opposing \emph{Emberiza-citrinella} and \emph{Regulus-regulus}) matches long‑term point‑count data showing spatial segregation driven by habitat preferences rather than by temperature alone \citep{Jarvinen2006,Newton2003}.

%% file: discussion.tex
We introduced \texttt{mixed-mSSL} to fit sparse high-dimensional regression models with multiple, possibly dependent, binary and continuous outcomes.
We represented the vector of continuous and binary outcomes as a partially observed latent realization from a multivariate linear regression model.
We specified element-wise spike-and-slab LASSO priors on the latent regression coefficient and residual precision matrices and derived an MCECM algorithm for computing the MAP.
We also showed that \texttt{mixed-mSSL} posterior has favorable theoretical properties. 


On synthetic data \texttt{mixed-mSSL} outperformed another recently proposed Bayesian shrinkage method \cite[\texttt{mt-MBSP}]{wang2023twostep}, which utilizes global-local shrinkage priors for the regression coefficients, in terms of support recovery and prediction.
When applied to the CKD dataset, \texttt{mixed-mSSL} identified more well-established biomarkers for CKD status and achieved a smaller prediction error than \texttt{mt-MBSP} and separately modeling each outcome.
\texttt{mixed-mSSL} also produced more accurate predictions with fewer selected covariates than \texttt{mt-MBSP} on the CRC microbiome data and uncovered a sparse residual dependence network among Finnish birds that is consistent with known co-occurrences.

Beyond demonstrating \texttt{mixed-mSSL}'s excellent empirical performance, we derived the posterior contraction rates for the matrix of latent regression coefficients $\bm{B}$ (\Cref{thm:unconditional}) and the latent residual precision matrix $\Omega$ (\Cref{thm:theorem3}).
While similar in spirit, our results substantially extend those proved in \citet{wang2023twostep} by (i) allowing the number of outcomes to diverge with $n$ and (ii) showing that the posterior over $\Omega$ contracts.
We further showed that with a natural truncation strategy, \texttt{mixed-mSSL} has the desirable sure screening property (\Cref{thm:surescreening}), which guarantees that the truly non-zero entries in $\bm{B}$ can be identified with high probability as $n$ increases.

For uncertainty quantification, rather than full MCMC, we can follow the Bayesian–bootstrap SSL (\texttt{BB-SSL}) approach of \citet{NieRockova2021}: draw observation weights and optimize a re-weighted log posterior using our MCECM algorithm to obtain pseudo–posterior draws of $(\bm B,\Omega)$. 
At least for high-dimensional regression with a single outcome, \texttt{BB-SSL} is a scalable approach to uncertainty quantification with near-optimal contraction in high-dimensional regression. 
As alluded to in \Cref{sec:method}, we can easily extend \texttt{mixed-mSSL} to accommodate integer- or count-valued outcomes by using a different deterministic transformation.
Practically, this requires changing the Monte Carlo E-step so that the missing latent components are drawn from multivariate normals whose entries are truncated to intervals like $[a, a+1]$ for some integer $a,$ as opposed to half-spaces like $(-\infty, 0)$ or $(0, \infty).$
For such draws, we can still use \texttt{LinESS}.
Analyzing such extensions theoretically introduces some new challenges, which can be overcome by developing bespoke tail bounds for interval-truncated normals and adapting our entropy calculations to handle shifting truncation sets; we leave such exploration for future work.
A more ambitious extension of \texttt{mixed-mSSL} would allow the transformation $g$ to be learned from the data instead of being fixed \textit{a priori}, similar to the STAR framework of \citet{KowalCanale2019_star}.

%% file: proofs.tex
We present the proofs of some of the theoretical results in this section. Before delving into the technicalities of the proofs, we state a few preliminary results (\Cref{sec:zposterior} -- \ref{sec:mixedcorr}) which aid us to better understand our theory.
\subsection{Conditional posterior of the binary latents}\label{sec:zposterior}
We partition the precision matrix and latent mean as \( \Omega = \begin{pmatrix}
    \Omega_{CC} & \Omega_{CB} \\
    \Omega_{CB} & \Omega_{BB} \\
\end{pmatrix},\) and \(\bm{m}_i = \bm{B}^{\top} \bx_i = \begin{pmatrix}
    \bm{m}_i ^{(C)} \\
    \bm{m}_i ^{(B)}
\end{pmatrix}
\) respectively, where $\Omega_{BB} \in \mathbb{R}^{q_b \times q_b}$ and $\text{diag}(\Omega_{BB} ^{-1})=\bm{1}_{q_b}$. Following the multivariate Normal conditional distribution formula in \citet[Appendix A2]{Rasmussen2005}, we have $\bm{z}_i^{(B)}|\bm{z}_i ^{(C)} = \bm{y}_i ^{(C)}, \bm{B}, \Omega \sim \mathcal{N}(\bm{m}_{B|C},\Omega_{BB} ^{-1})$ with $\bm{m}_{B|C} =\bm{m}_i ^{(B)}-\Omega_{BB}^{-1}\Omega_{CB}(\by_i ^{(C)}-\bm{m}_i ^{(C)})$.
\subsection{Interpretation of model coefficients}\label{sec:interp}
For every covariate-outcome pair $(j,k)$, the entry $\beta_{jk}$ measures the effect of increasing $x_j$ by one unit while holding the other covariates fixed. 
For continuous outcomes, $\beta_{j,k}$ is the familiar partial regression coefficient $\beta_{j,k} = \E[Y_{k} \vert X_{j} = x+1, \bm{X}_{-j}] - \E[Y_{k} \vert X_{j} = x, \bm{X}_{-j}],$ where $\bm{X}_{-j}$ is the covariate vector without its $j$-th entry.
For binary outcomes, a first-order Taylor expansion shows that an upper bound of the associated risk difference depends on the magnitude of $\beta_{j,k}$:
\begin{align*}
\P(Y_{k} = 1 \vert X_{j} = x+1, \bX_{-j} = \bx_{-j}) - \P(Y_{k} = 1 \vert X_{j} = x, \bX_{-j} = \bx_{-j}) &= \\
\Phi\left(\beta_{j,k}(x+1) + \sum_{j' \neq j}{\beta_{j',k}x_{j'}}\right) - \Phi\left(\beta_{j,k}x + \sum_{j' \neq j}{\beta_{j',k}x_{j'}}\right) &\approx \\
\beta_{j,k}\phi\left(\left[\beta_{j,k}x + \sum_{j' \neq j}{\beta_{j',k}x_{j'}}\right]\right) &\leq (2\pi)^{-\frac{1}{2}}\lvert \beta_{j,k}\rvert,
\end{align*}
where $\phi$ and $\Phi$ are the standard normal probability density and cumulative distribution functions. In a nutshell: $\beta_{jk}$ is the per–unit, ceteris paribus effect. For binary outcomes, a one-unit increase in $x_j$ changes $\mathbb{P}(Y_k{=}1)$ by approximately $\beta_{jk}\,\phi(\eta)$ (with $\eta$ the current linear predictor), and this effect is universally bounded by $\approx 0.40\,|\beta_{jk}|$; the sign of $\beta_{jk}$ sets the direction.

\subsection{Proof of Lemma 1}
\label{sec:mixedcorr}
\switchref{\Cref{lem:correlation_bounds}}{Lemma 1} asserts that the maximal correlation allowed between a binary and continuous outcome under our model is $\sqrt{2/\pi}.$

\begin{proof}

Suppose $y_{i1}$ is continuous and $y_{i2}$ is binary, with latent Gaussian variables
$z_{i1}\sim \mathcal{N}(0,1),$ $z_{i2}\sim \mathcal{N}(0,1),$ and $\mathrm{Corr}(z_{i1},z_{i2})=\rho.$
We have $y_{i1} = z_{i1}, y_{i2} = \mathbbm{1}\{z_{i2} > 0\}.$

We want to bound $\mathrm{Corr}(y_{i1},y_{i2})$. 
Here,
$\mathrm{Var}(y_{i1}) = \mathrm{Var}(z_{i1})=1,
  \quad
  \mathrm{Var}(y_{i2}) 
   = \mathrm{Var}\bigl(\mathbbm{1}\{z_{i2}>0\}\bigr)
   = \tfrac12\,\bigl(1-\tfrac12\bigr)
   = \tfrac{1}{4}.$
Following the ideas from Lemma 2.1 in \citet{ekvall}, for the numerator, $\mathrm{Cov}(z_{i1},\mathbbm{1}\{z_{i2}>0\})$, we write
\[
   \mathbb{E}\bigl[z_{i1}\,\mathbbm{1}\{z_{i2}>0\}\bigr]
   \;=\;\int_{-\infty}^{\infty}
        z_1\,\mathbb{P}(z_{i2}>0 \mid z_{i1}=z_1)
        \,\phi(z_1)\,\mathrm{d}z_1,
\]
where $\phi(\cdot)$ is the standard normal pdf, and
\(
  \mathbb{P}\{z_{i2}>0\mid z_{i1}=z_1\}
  = \Phi\!\Bigl(\rho\,z_1/\sqrt{1-\rho^2}\Bigr).
\)
Therefore, $\mathrm{Cov}(z_{i1},\mathbbm{1}\{z_{i2}>0\})
   \;=\;\int_{-\infty}^\infty
        z_1\,\phi(z_1)\,\Phi\!\Bigl(\tfrac{\rho\,z_1}{\sqrt{1-\rho^2}}\Bigr)
        \,\mathrm{d}z_1.$
Hence,
\[
   \mathrm{Corr}\bigl(y_{i1},y_{i2}\bigr)
   = \frac{\mathbb{E}[z_{i1}\,\mathbbm{1}\{z_{i2}>0\}]}
          {\tfrac12}
   = 2\int_{-\infty}^\infty
        z_1\,\phi(z_1)\,\Phi\!\Bigl(\tfrac{\rho\,z_1}{\sqrt{1-\rho^2}}\Bigr)
     \,\mathrm{d}z_1.
\]
As $\rho$ increases toward $1$, this integral grows but cannot exceed $\sqrt{\tfrac{2}{\pi}}\approx 0.8$.
A rigorous argument uses the Dominated Convergence Theorem to show $\lim_{\rho\to1^-} \mathrm{Corr}(y_{i1},y_{i2})
   \;=\;\sqrt{\tfrac{2}{\pi}}$. 
   To see this formally, we write  
\(
   I_\rho
   \;=\;
   2\!\int_{-\infty}^{\infty}
        z\,\phi(z)\,
        \Phi\Bigl(\tfrac{\rho\,z}{\sqrt{1-\rho^{2}}}\Bigr)\,dz,
   \ 0\le \rho<1,
\)
and define  
\(
   h_\rho(z)=2\,z\,\phi(z)\,
             \Phi\Bigl(\tfrac{\rho\,z}{\sqrt{1-\rho^{2}}}\Bigr),
   \
   g(z)=2\,|z|\,\phi(z).
\)
For any fixed \(z\neq 0\),
\(
  \dfrac{\rho z}{\sqrt{1-\rho^{2}}}\to\pm\infty
\)
as \(\rho\uparrow 1\); hence  
\(
   h_\rho(z)\to h(z):=2\,z\,\phi(z)\,\mathbbm {1}_{\{z>0\}}.
\)
Because \(0\le\Phi(\cdot)\le 1\), \(0\le h_\rho(z)\le g(z)\).  
Moreover,
\(
 \displaystyle\int_{-\infty}^{\infty} g(z)\,dz
 =\frac{4}{\sqrt{2\pi}}<\infty,
\)
so \(g\) is integrable. 
Now, by the Dominated Convergence Theorem, we have
\[
   \lim_{\rho\uparrow 1} I_\rho
   \;=\;
   \int_{-\infty}^{\infty} h(z)\,dz
   \;=\;
   2\!\int_{0}^{\infty} z\,\phi(z)\,dz
   \;=\;
   2\,\frac{1}{\sqrt{2\pi}}
   \;=\;
   \sqrt{\tfrac{2}{\pi}}.
\]
The same argument with \(\rho\to -1^+\) gives the lower bound
\(-\sqrt{2/\pi}\). 
\end{proof}

Intuitively, $y_{i2} = \ind{z_{i2}>0}$ only tracks the \emph{sign} of $z_{i2}$, so it cannot fully capture the magnitude alignment between $z_{i1}$ and $z_{i2}.$ 
This partial alignment caps the observed correlation below 1. 
Hence, for a mixed continuous--binary pair, the observed correlation cannot exceed $\sqrt{2/\pi}$, reflecting a fundamental limit from thresholding one of the two latent variables. 

\subsection{Proof of Lemma 2}
\label{sec:prooflemwang}

\begin{proof}
We write the row–wise decomposition
\begin{equation}\label{eq:frobsplit}
\|\bm B-\bm B_0\|_F^2
=\sum_{j=1}^{p_n}\|\bm b_j-\bm b^0_j\|_2^2
=\sum_{j\in S_0}\|\bm b_j-\bm b^0_j\|_2^2
+\sum_{j\notin S_0}\|\bm b_j\|_2^2.
\end{equation}
Fix $C>0$, $\rho\in(0,b)$, and set
\(
\Delta_n^2:=\frac{C^2\,\epsilon_n^2}{2 s_0^B\,n^\rho},
\quad
\Delta:=\frac{\Delta_n}{\sqrt{q_n}}.
\)
Then
\[
\mathbb{P}\left(\|\bm B-\bm B_0\|_F < C\,n^{-\rho/2}\epsilon_n\right)
\;\ge\;
\underbrace{\prod_{j\in S_0}\mathbb{P}\left(\|\bm b_j-\bm b^0_j\|_2^2<\Delta_n^2\right)}_{\mathcal A_1}
\cdot
\underbrace{\mathbb{P}\left(\sum_{j\notin S_0}\|\bm b_j\|_2^2<\tfrac{C^2\epsilon_n^2}{2 n^\rho}\right)}_{\mathcal A_2}.
\]
This holds by observing the events 
\[
E_1 \coloneqq \bigcap_{j \in S_0} \left\{ \|\bm{b}_j-\bm{b}_j ^{0} \| _2 ^2 < \Delta_n ^2 \right\}, \quad E_2 \coloneqq \left \{\sum_{j \notin S_0} \| \bm{b}_j \|_2 ^2 < \frac{C^2 \epsilon_n ^2}{2 n^\rho} \right \}, 
\]
then on $E_1 \cap E_2,$ using \Cref{eq:frobsplit}, we have $\| \bm{B}-\bm{B}_0 \|_F ^2 < s_0 ^B \Delta_n ^2 + C^2 \epsilon_n ^2 /(2 n^{\rho})=C^2 \epsilon_n^2/n^\rho,$ so $E_1 \cap E_2 \subset \{ \|\bm{B}-\bm{B}_0\|_F < C n^{-\rho/2}\epsilon_n\}.$ Hence, $\P(\| \bm{B}-\bm{B}_0\|_F < C n^{-\rho/2} \epsilon_n) \ge \P(E_1 \cap E_2).$ Now, under the element–wise spike–and–slab prior, the coefficients $\{\beta_{jk}\}$ are independent across $(j,k),$ so the rows $\{\bm{b}_j\}$ are independent. It is trivial to see that $E_1$ and $E_2$ are independent and thus the inequality follows:
\[
\P(E_1 \cap E_2) = \P(E_1)\P(E_2) = \left( \prod_{j \in S_0} \P(\|\bm{b}_j - \bm{b}_j ^{0} \|_2 ^2 < \Delta_n ^2) \right) \cdot \P \left( \sum_{j \notin S_0} \|\bm{b}_j \|_2 ^2 < \frac{C^2 \epsilon_n ^2}{2 n^{\rho}} \right).
\]
\paragraph{Bounding $\mathcal A_2$.}
For $j\notin S_0$, the element–wise spike–slab prior implies
\(
\mathbb E[\beta_{jk}^2]=(1-\theta) \cdot (2/\lambda_0^2)+\theta \cdot (2/\lambda_1^2)
\).
By {\rm(A6)},
\begin{equation*}
\begin{split}
\mathbb E\|\bm b_j\|_2^2
& = q_n\Big((1-\theta)\frac{2}{\lambda_0^2}+\theta\frac{2}{\lambda_1^2}\Big) \\
& \le q_n \cdot \frac{2}{\lambda_0 ^2} + q_n \theta \cdot \frac{2}{\lambda_1 ^2} \\
& \le \frac{2}{C_0}\cdot \frac{n^{(-1+\rho)}}{p_n q_n} + \frac{2C_1}{\lambda_1 ^2}\cdot \frac{n^{-(1+\rho)}}{p_n q_n} \\
& \le c_0 \cdot \frac{n^{-(1+\rho)}}{p_n q_n}
\end{split}
\end{equation*}
for some constant $c_0 \coloneqq 2/C_0 + 2C_1/\lambda_1 ^2$. Hence, by Markov's inequality,
\[
\mathcal A_2
\ge 1-
\frac{(p_n-s_0^B)\,c_0\,\frac{n^{-(1+\rho)}}{p_n q_n}}{\frac{C^2\epsilon_n^2}{2 n^\rho}}
=
1-\frac{2c_0}{C^2}\cdot\frac{1}{n q_n \epsilon_n^2}\cdot\frac{p_n-s_0^B}{p_n}.
\]
With $\epsilon_n^2=(q_n s_0^B \log p_n)/n$,
\(
n q_n\epsilon_n^2=q_n^2 s_0^B\log p_n\to\infty
\)
by {\rm(A1)} and {\rm(A2)}; hence $\mathcal A_2\to 1$. Therefore,
$\mathcal A_2\ge e^{-D_0 n\epsilon_n^2}$ for some $D_0>0$ and all large $n$.

\paragraph{Bounding $\mathcal A_1$.}
For any active $(j,k)$, using only the slab part,
\[
\mathbb{P}\left(|\beta_{jk}-\beta^0_{jk}|\le \Delta\right)
\ge
\int_{\beta^0_{jk}-\Delta}^{\beta^0_{jk}+\Delta}
\theta\frac{\lambda_1}{2}e^{-\lambda_1|t|}\,dt
\ge \theta\lambda_1\Delta\,e^{-\lambda_1(|\beta^0_{jk}|+\Delta)}
\ge C_1'\,\Delta\,e^{-\lambda_1|\beta^0_{jk}|},
\]
for large $n$ (since $e^{-\lambda_1\Delta}\ge 1/2$), with $C_1'=\theta\lambda_1/2$.
Multiplying over the $s_0^B q_n$ active entries yields
\[
\mathcal A_1
\ge (C_1'\Delta)^{s_0^B q_n}\exp\!\big(-\lambda_1\|\bm B_{0,S_0}\|_1\big).
\]
Taking logs and recalling $\Delta=\Delta_n/\sqrt{q_n}$ and
\(
\Delta_n^2=\frac{C^2}{2}\cdot\frac{q_n\log p_n}{n^{1+\rho}}
\),
\[
-\log\mathcal A_1
\le s_0^B q_n\Big\{-\log C_1'-\tfrac12\log \Delta_n^2+\tfrac12\log q_n\Big\}
+\lambda_1\|\bm B_{0,S_0}\|_1.
\]
Now
\(
-\tfrac12\log \Delta_n^2+\tfrac12\log q_n
=\tfrac12\{-(\log(C^2/2))-\log\log p_n+(1+\rho)\log n\}.
\)
Using {\rm(A1)} ($\log p_n\gtrsim q_n\log n$) gives
\(
s_0^B q_n\{-\tfrac12\log \Delta_n^2+\tfrac12\log q_n\}
\lesssim s_0^B\log p_n\lesssim n\epsilon_n^2.
\)
Finally, {\rm(A4)} ensures
$\lambda_1\|\bm B_{0,S_0}\|_1=o(n\epsilon_n^2)$, so
$\mathcal A_1\ge \exp\{-D_1 n\epsilon_n^2\}$ for some $D_1>0$.
Combining the steps completes the proof with $D=D_0+D_1$.
\end{proof}

\subsection{Conditional Posterior Contraction of $\textrm{B}$}
\label{sec:proofpostconcBgivenZ}
To prove unconditional contraction for $\bm B$, we first work conditionally on the augmented latents $\bm Z_n$. 
Conditioning on $\bm Z_n$ converts the mixed-type likelihood into a standard multivariate Gaussian regression with mean $\bm B^\top \bx_i$ and precision $\Omega$, thereby side-stepping orthant probabilities. 

\begin{lemma}[Posterior conditional contraction result for $\bm{B}$]
\label{lem:postconcBgivenZ}
Under the Assumptions (A1) -- (A6),  for any arbitrary constant $M > 0$ and $\epsilon_n ^2$ as in Lemma 2, $$\sup_{\bm{B}_0} \mathbb{E}_{0}\left[ \mathbb{P} \left( \lVert \bm{B}_n-\bm{B}_0 \rVert_F > M \epsilon_n \mid \bm{Y}_n,\bm{Z}_n \right)\right] \rightarrow 0,$$ as $n \rightarrow \infty.$  
\end{lemma}
We prove \Cref{lem:postconcBgivenZ} using the standard recipe of \citet{Ghosal_2007}: we need to (1) construct a sequence of sets (sieves) of bounded complexity that essentially ``capture" the parameter space as $n \rightarrow \infty$ (see \Cref{prop:tests-prod}). (2) Then, we need to verify that the prior places enough mass (or concentration) around the true parameter within the sieves (see \Cref{prop:wang}). The sieve construction for this problem is outlined as follows.

\subsubsection*{(1) The sieve construction}
We define $S \subset \mathcal{I}=\{1,\ldots,p_n\}$ to be a set of row indices for the regression coefficients matrix. Let $S_0 \subset \{1,2,\ldots,p_n\}$ denote the set of indices of the rows in $\bm{B}_0$ with at least one non-zero entry. Observe that $|S_0|=s_0 ^{B} \ge 1$ and $s_0 ^{B} = o(n/\log p_n)$. Additionally, we denote $G_S \coloneqq n^{-1}\bm{X}_S^\top \bm{X}_S$. The sieve construction is similar to \citet{wang2023twostep}. 
Let $\epsilon_n^2=(q_n s_0^B \log p_n)/n$ and fix $\rho\in(0,b)$. We define the model index class
\[
\mathcal M
=\Bigl\{S:\ S\supset S_0,\ S\neq S_0,\ |S|\le m_n\Bigr\},
\quad
m_n:=\big\lceil K\,s_0^B\,n^\rho\big\rceil,
\]
for a constant $K>1$, and the sieve
\begin{equation}\label{eq:sieve}
\mathcal C_n
=\bigcup_{S\in\mathcal M}
\Bigl\{\ \|\widehat{\bm B}^{\,S}-\bm B_0^S\|_F>\tfrac12\,\delta\,\epsilon_n\ \Bigr\},
\end{equation}
where $\widehat{\bm B}^{\,S}$ is the MLE in the Gaussian regression of $\bm Z_n$ on $\bm{X}_S$
with precision $\Omega$.

\begin{proposition}[Exponentially consistent tests on $\mathcal{C}_n$]\label{prop:tests-prod}
Following the arguments in \citet{wang2023twostep}, there exist constants $C_1,C_2>0$ such that, with a test function
$\Phi_n=\mathbbm{1}\{\bm Z_n\in\mathcal C_n\}$,
\[
\mathbb E_{\bm B_0}[\Phi_n]\ \le\ e^{-C_1 n\epsilon_n^2},
\quad
\sup_{\|\bm B-\bm B_0\|_F>\delta\,\epsilon_n}\ \mathbb E_{\bm B}(1-\Phi_n)\ \le\ e^{-C_2 n\epsilon_n^2}.
\]
for an arbitrary $\delta>0.$
\end{proposition}

\begin{proof}
Conditionally on $(\bm{X}_n,\Omega)$,
\[
\mathrm{vec}\big(\widehat{\bm B}^{\,S}-\bm B_0^S\big)
\sim \mathcal N_{p^S q_n}\!\Big(0,\ \frac1n\,G_S^{-1}\otimes \Sigma\Big),
\]
so with $d=p^S q_n$ and $U\sim\mathcal N_d(0,I_d)$,
\[
\|\widehat{\bm B}^{\,S}-\bm B_0^S\|_F^2
=\frac{1}{n}\,U^\top(G_S^{-1}\otimes\Sigma)\,U.
\]
From the eigenvalue bounds,
\(
\underline c\,\chi^2_d/n\ \le\ \|\widehat{\bm B}^{\,S}-\bm B_0^S\|_F^2
\ \le\ \overline c\,\chi^2_d/n
\),
with $\underline c=1/(\bar\tau k_1)$ and $\overline c=k_1/\underline\tau$.

\emph{Type I error.} We fix $S\in\mathcal M$ and apply the inequality in \citet[Lemma 1]{LaurentMassart2000}:
$\mathbb P(\chi^2_d>d+2\sqrt{dx}+2x)\le e^{-x}$ for any $x>0$.
Now, we set $x=c_2 n\epsilon_n^2$. Because $d\le m_n q_n\le K s_0^B n^\rho q_n$ while
$n\epsilon_n^2=s_0^B q_n\log p_n$, and $\rho<b$ with {\rm(A1)}, it follows that
\[
\frac{\overline c}{n}\big(d+2\sqrt{dx}+2x\big)
=o(\epsilon_n^2)+2\overline c c_2\epsilon_n^2
\ \le\ \frac{\delta^2}{4}\epsilon_n^2
\]
for large $n$ with $c_2\le \delta^2/(16\overline c)$, hence
\(
\mathbb P_{\bm B_0}\big(\|\widehat{\bm B}^{\,S}-\bm B_0^S\|_F^2>\tfrac14\delta^2\epsilon_n^2\big)
\le e^{-c_2 n\epsilon_n^2}.
\)
A union bound over $S\in\mathcal M$ yields
\[
\mathbb E_{\bm B_0}[\Phi_n]\le |\mathcal M|\,e^{-c_2 n\epsilon_n^2}.
\]
The combinatorial bound
$|\mathcal M|\le \sum_{k=s_0^B+1}^{m_n}\binom{p_n}{k}
\le \exp\{m_n\log(p_n e/m_n)\}$ with $m_n=K s_0^B n^\rho$
gives $\log|\mathcal M|=o(n\epsilon_n^2)$ since $q_n=n^b\gg n^\rho$ and
$n\epsilon_n^2=s_0^B q_n\log p_n$.
Thus $\mathbb E_{\bm B_0}[\Phi_n]\le e^{-C_1 n\epsilon_n^2}$.

\emph{Type II error.}
Fix $\delta>0$ and take any parameter $\bm B$ with
$\|\bm B-\bm B_0\|_F>\delta\,\epsilon_n$.
We write the rowwise squared deviations
$r_j:=\|\bm b_j-\bm b_j^0\|_2^2$,
and split the indices outside $S_0$ into
\[
S_\Delta:=\{j\notin S_0:\ r_j\ge \Delta_n^2\},\qquad
U:=\{j\notin S_0:\ r_j<\Delta_n^2\}.
\]
We set
\(
T:=\sum_{j\in S_\Delta} r_j,\quad
M_{\mathrm{small}}:=\sum_{j\in U} r_j,
\quad
s:=m_n-s_0^B=\lfloor K s_0^B n^\rho\rfloor,
\)
with $K>0$ to be fixed below.
Our idea is to produce a set $S\in\mathcal M$ for which
\begin{equation}\label{eq:signal-lb}
\|\bm B^{S}-\bm B_0^{S}\|_F \ \ge\ c_\star\,\delta\,\epsilon_n
\end{equation}
with a constant $c_\star>1/2$, uniformly over all such $\bm B$.
This will yield an exponentially small upper bound for the Type II error.

\medskip
\noindent\textit{Case 1: $T\ge \tfrac12\,\delta^2\epsilon_n^2$.}
Let $S$ be $S_0$ together with the $s$ indices from $S_\Delta$ having the largest $r_j$ values
(or all of $S_\Delta$ if $|S_\Delta|<s$).
Because every $j\in S_\Delta$ satisfies $r_j\ge \Delta_n^2$,
\[
\|\bm B^{S}-\bm B_0^{S}\|_F^2
\ \ge\
\sum_{j\in S\setminus S_0} r_j
\ \ge\
\min\big\{T,\ s\,\Delta_n^2\big\}.
\]
Choose $K\ge \delta^2/C^2$. Then
\[
s\,\Delta_n^2
= (m_n-s_0^B)\,\Delta_n^2
\ \ge\ K s_0^B n^\rho\cdot \frac{C^2\epsilon_n^2}{2 s_0^B n^\rho}
\ =\ \frac{K C^2}{2}\,\epsilon_n^2
\ \ge\ \frac12\,\delta^2\epsilon_n^2.
\]
Hence $\|\bm B^{S}-\bm B_0^{S}\|_F^2\ge \tfrac12\,\delta^2\epsilon_n^2$, i.e.
\begin{equation}\label{eq:case1}
\|\bm B^{S}-\bm B_0^{S}\|_F\ \ge\ \frac{\delta}{\sqrt{2}}\,\epsilon_n.
\end{equation}

\noindent\textit{Case 2: $T< \tfrac12\,\delta^2\epsilon_n^2$.}
We split this case depending on the size of $M_{\mathrm{small}}$.

If $M_{\mathrm{small}}<\Delta_n^2$,
then
\[
\sum_{j\in S_0} r_j
\ =\
\|\bm B-\bm B_0\|_F^2 - \sum_{j\notin S_0} r_j
\ \ge\
\delta^2\epsilon_n^2 - (T+M_{\mathrm{small}})
\ >\
\frac12\,\delta^2\epsilon_n^2 - \Delta_n^2.
\]
Fix any $\eta\in(0,1/4)$.
Since $\epsilon_n^2/\Delta_n^2\to\infty$ (because $\Delta_n^2\asymp \epsilon_n^2/(s_0^B n^\rho)$ and $s_0^B n^\rho\to\infty$), for all large $n$ we have
$\Delta_n^2\le \eta\,\delta^2\epsilon_n^2$. Therefore,
\[
\sum_{j\in S_0} r_j\ \ge\ \big(\tfrac12-\eta\big)\,\delta^2\epsilon_n^2.
\]
Taking $S=S_0$ (which is in $\mathcal M$ since $|S_0|=s_0^B\le m_n$),
\begin{equation}\label{eq:case2a}
\|\bm B^{S}-\bm B_0^{S}\|_F
\ \ge\ \sqrt{\tfrac12-\eta}\,\delta\,\epsilon_n.
\end{equation}
Again, if $M_{\mathrm{small}}\ge \Delta_n^2$, let $r_{(1)}\ge r_{(2)}\ge\cdots$ be the nonincreasing rearrangement of $\{r_j:j\in U\}$,
and set $\ell:=\big\lceil M_{\mathrm{small}}/\Delta_n^2\big\rceil$.
Since each $r_j<\Delta_n^2$, at least $\ell$ terms are needed to reach $M_{\mathrm{small}}$,
so $\sum_{j=1}^{\ell} r_{(j)}\ge M_{\mathrm{small}}$ and $\ell\le 2M_{\mathrm{small}}/\Delta_n^2$.
The average of the top $s$ is at least the average of the top $\ell$; hence
\[
\sum_{j=1}^{s} r_{(j)}
\ \ge\ \frac{s}{\ell}\sum_{j=1}^{\ell} r_{(j)}
\ \ge\ \frac{s}{\ell}\,M_{\mathrm{small}}
\ \ge\ \frac{s}{2}\,\Delta_n^2.
\]
Take $S$ to be $S_0$ together with these $s$ indices from $U$.
As in Case~1, with $K\ge \delta^2/C^2$ we get
\[
\|\bm B^{S}-\bm B_0^{S}\|_F^2
\ \ge\ \frac{s}{2}\,\Delta_n^2
\ \ge\ \frac{K C^2}{4}\,\epsilon_n^2
\ \ge\ \frac12\,\delta^2\epsilon_n^2,
\]
so
\begin{equation}\label{eq:case2b}
\|\bm B^{S}-\bm B_0^{S}\|_F
\ \ge\ \frac{\delta}{\sqrt{2}}\,\epsilon_n.
\end{equation}

Combining \eqref{eq:case1}, \eqref{eq:case2a}, \eqref{eq:case2b}, we have produced an $S\in\mathcal M$ for which
\[
\|\bm B^{S}-\bm B_0^{S}\|_F
 \ge c_\star\delta\epsilon_n,
\quad
c_\star:=\min\Big\{\tfrac{1}{\sqrt2},\sqrt{\tfrac12-\eta}\Big\}
>\tfrac12,
\]
for all sufficiently large $n$.

By definition of the test $\Phi_n$, the event $\{1-\Phi_n=1\}$ implies
$\|\widehat{\bm B}^{\,S}-\bm B_0^{S}\|_F\le \tfrac12\,\delta\,\epsilon_n$
for \emph{every} $S\in\mathcal M$, and in particular for the $S$ just constructed.
Hence, by the triangle inequality,
\[
\|\widehat{\bm B}^{\,S}-\bm B^{S}\|_F
 \ge
\|\bm B^{S}-\bm B_0^{S}\|_F-\|\widehat{\bm B}^{\,S}-\bm B_0^{S}\|_F
 \ge (c_\star-\tfrac12)\,\delta\,\epsilon_n
 \ge \frac{\delta}{8}\,\epsilon_n,
\]
for all large $n$.

Finally, conditionally on $(X_n,\Omega)$,
\[
\mathrm{vec}\big(\widehat{\bm B}^{\,S}-\bm B^{S}\big)
\sim \mathcal N_{d}\Big(0,\ \frac1n\,G_S^{-1}\otimes \Sigma\Big),
\quad d=p^{S}q_n,
\]
so by the eigenvalue bounds in {\rm(A3)}–{\rm(A5)},
\[
\frac{\underline c}{n}\,\chi^2_{d}
\ \le\ \|\widehat{\bm B}^{\,S}-\bm B^{S}\|_F^2
\ \le\ \frac{\overline c}{n}\,\chi^2_{d},
\qquad
\underline c=\frac{1}{\bar\tau k_1},\ \overline c=\frac{k_1}{\underline\tau}.
\]
Using Laurent–Massart with $x=\tilde{c}n\epsilon_n^2$
for $\tilde{c}:=\delta^2/(128\,\overline c)$ and
$d\le m_n q_n\lesssim s_0^B n^\rho q_n=o(n\epsilon_n^2)$
(by {\rm(A1)} and $\rho<b$), we obtain
\[
\mathbb P_{\bm B}\!\left(\|\widehat{\bm B}^{\,S}-\bm B^{S}\|_F\ge \frac{\delta}{8}\,\epsilon_n\right)
 \le \exp \big\{-\tilde{c}n\epsilon_n^2\big\}.
\]
Because $\{1-\Phi_n=1\}\subseteq\{\|\widehat{\bm B}^{\,S}-\bm B^{S}\|_F\ge \delta\epsilon_n/8\}$,
\[
\mathbb P_{\bm B}(1-\Phi_n)\ \le\ \exp\!\big(-C_2\,n\epsilon_n^2\big),
\qquad C_2:=c_\sharp>0.
\]
Since $\Phi_n$ is an indicator, $\mathbb E_{\bm B}(1-\Phi_n)=\mathbb P_{\bm B}(1-\Phi_n)$,
and the same bound holds for the expectation. This completes the Type--II error bound.

\end{proof}
\subsubsection*{(2) Prior mass and denominator control}
It remains to control the posterior denominator via a localized prior mass lower bound. Thus, we only need to prove the following proposition to complete the proof of \Cref{lem:postconcBgivenZ}. 
\begin{proposition}\label{prop:wang}
Let $\epsilon_n^2=(q_n s_0^B \log p_n)/n$.
Under {\rm(A1)}–{\rm(A6)}, for any fixed $M>0$,
\[
\sup_{\bm B_0}\;
\mathbb E_{\bm B_0}\Big[
\mathbb{P}\big(\|\bm B-\bm B_0\|_F>M\epsilon_n\ \big|\ \bm Y_n,\bm Z_n\big)
\Big]\ \longrightarrow\ 0.
\]
\end{proposition}

\begin{proof}
Conditionally on $(\bm X_n,\bm Z_n)$ the model is a $q_n$–variate Gaussian regression
with mean $\bm B^\top \bm x_i$ and precision $\Omega$.
For $\Delta=\bm B-\bm B_0$ and $G:=n^{-1}\bm X^\top \bm X$ we set
\[
Q_n(\Delta)
=\sum_{i=1}^n \Delta_i^\top \Omega\,\Delta_i
=\mathrm{tr}\big(\Omega\,\Delta^\top \bm X^\top \bm X \Delta\big)
=n\,\mathrm{tr}\big(\Omega\,\Delta^\top G\,\Delta\big).
\]
By {\rm(A3)}–{\rm(A5)}, there are constants
\(
\underline c:=\underline\tau/k_1
\)
and
\(
\overline c:=k_1\bar\tau
\)
such that
\begin{equation}\label{eq:Q-two-sided}
n\,\underline c\,\|\Delta\|_F^2
\;\le\; Q_n(\Delta)\;\le\; n\,\overline c\,\|\Delta\|_F^2.
\end{equation}
The log–likelihood ratio between $p(\bm Z_n\mid \bm B)$ and $p(\bm Z_n\mid \bm B_0)$ is
\[
L(\bm B)
:=\log\frac{p(\bm Z_n\mid \bm B)}{p(\bm Z_n\mid \bm B_0)}
=\sum_{i=1}^n \bm\varepsilon_i^\top \Omega\,\Delta_i-\frac12\,Q_n(\Delta)
=:S_n(\Delta)-\tfrac12 Q_n(\Delta),
\]
where $\bm\varepsilon_i=\bm z_i-\bm B_0^\top \bm x_i\sim\mathcal N_{q_n}(\bm 0,\Omega^{-1})$, conditionally on $\bm X_n$. Hence $S_n(\Delta)\mid \bm X_n\sim\mathcal N\big(0,Q_n(\Delta)\big)$.

Letting $\mathcal B_M=\{\|\bm B-\bm B_0\|_F>M\epsilon_n\}$ and using the same sieve $\mathcal C_n$ and the test $\Phi_n=\mathbf 1\{\bm Z_n\in\mathcal C_n\}$ as in \Cref{prop:tests-prod}.
That gives us constants $C_1,C_2>0$ (independent of $n$ and $\bm B_0$) such that
\begin{equation}\label{eq:test-bounds}
\mathbb E_{\bm B_0}[\Phi_n]\;\le\;e^{-C_1 n\epsilon_n^2},
\quad
\int_{\mathcal B_M}\mathbb P_{\bm B}(1-\Phi_n)\,d\Pi(\bm B)
\;\le\; e^{-C_2 n\epsilon_n^2}.
\end{equation}

For Gaussian regression with common precision $\Omega$,
the per–sample KL divergence and its variance proxy satisfy
\begin{equation}\label{eq:KL-bounds}
K\bigl(p_{\bm B_0},p_{\bm B}\bigr)
=\frac{1}{2n}\,Q_n(\Delta),
\quad
V\bigl(p_{\bm B_0},p_{\bm B}\bigr)
\lesssim \frac{1}{n}\,Q_n(\Delta),
\end{equation}
hence, by \eqref{eq:Q-two-sided},
\(
K \le \frac{\overline c}{2}\|\Delta\|_F^2,
\quad
V \le C_V \|\Delta\|_F^2
\quad\text{for some }C_V>0.
\)
Fixing a small constant $\kappa'>0$, we  define the KL–ball 
\[
\mathcal K_n
=\Bigl\{\bm B:\ \|\bm B-\bm B_0\|_F \le r_n\Bigr\},
\quad
r_n^2:=\kappa'\,\frac{\epsilon_n^2}{n}.
\]
Then by \eqref{eq:KL-bounds},
\[
\sup_{\bm B\in\mathcal K_n}
\Bigl\{\,K\bigl(p_{\bm B_0},p_{\bm B}\bigr)
\ \vee\ V\bigl(p_{\bm B_0},p_{\bm B}\bigr)\Bigr\}
 \le C_{\mathrm{KL}}\epsilon_n^2
\quad\text{with }C_{\mathrm{KL}}:=\max\{\tfrac{\overline c}{2},\,C_V\}\kappa'.
\]

Repeating the prior mass argument used in the proof of Lemma 2,
but now with $\Delta\asymp r_n/\sqrt{q_n}$ instead of $n^{-\rho/2}\epsilon_n/\sqrt{q_n}$,
gives the same multiplicative bound
\(
\Pi(\mathcal K_n) \ge \exp\{-D n\epsilon_n^2\}
\)
for some $D>0$.
The only change is an extra $\tfrac12\log n$ term inside
$s_0^Bq_n\,[-\log\Delta]$, which is absorbed by (A1) since
$s_0^Bq_n\log n \lesssim s_0^B\log p_n = n\epsilon_n^2$.

By the \emph{denominator lemma} of \citet[]{Ghosal2000},
we have that there exists $C_0>0$ such that
\begin{equation}\label{eq:denom-lower}
\mathbb P_{\bm B_0}\!\left(
J_n:=\int \frac{p(\bm Z_n\mid \bm B)}{p(\bm Z_n\mid \bm B_0)}\,d\Pi(\bm B)
\ \ge\ \exp\{-(D+C_0)\,n\,\epsilon_n^2\}
\right) \ge 1-\exp(-cn\epsilon_n ^2),
\end{equation}
uniformly in $\bm B_0$.
Equivalently, for any fixed $\psi\in(0,1)$ and all large $n$,
\begin{equation}\label{eq:denom-prob}
\mathbb P_{\bm B_0}\!\left(
J_n \le \psi  e^{-(D+C_0)n \epsilon_n^2}
\right) \le \mathbb P_{\bm B_0}\!\left(
J_n \le e^{-(D+C_0)n \epsilon_n^2}
\right) \le e^{-c\,n\,\epsilon_n^2}
\end{equation}
for some $c>0$.

By definition,
\[
\Pi(\mathcal B_M\mid \bm Z_n)
=\frac{J_{\mathcal B_M}}{J_n}
\le \Phi_n + (1-\Phi_n)\,\frac{J_{\mathcal B_M}}{J_n},
\qquad
J_{\mathcal B_M}
=\int_{\mathcal B_M}\frac{p(\bm Z_n\mid \bm B)}{p(\bm Z_n\mid \bm B_0)}\,d\Pi(\bm B).
\]
Taking $\mathbb E_{\bm B_0}$ and splitting on the event $\{J_n\ge \psi e^{-(D+C_0)n\epsilon_n^2}\}$,
\begin{align*}
\mathbb E_{\bm B_0}\big[\Pi(\mathcal B_M\mid \bm Z_n)\big]
&\le \mathbb E_{\bm B_0}[\Phi_n]
+
\mathbb E_{\bm B_0}\!\left[
(1-\Phi_n)J_{\mathcal B_M}\,
\mathbbm{1}\!\left\{J_n\ge \psi e^{-(D+C_0)n\epsilon_n^2}\right\}
\cdot \frac{e^{(D+C_0)n\epsilon_n^2}}{\psi}
\right]\\
&\hspace{1.2cm}
+ \mathbb P_{\bm B_0}\!\left(J_n< \psi e^{-(D+C_0)n\epsilon_n^2}\right).
\end{align*}
By \eqref{eq:test-bounds} and \eqref{eq:denom-prob},
\[
\mathbb E_{\bm B_0}[\Phi_n]\le e^{-C_1 n\epsilon_n^2},
\quad
\mathbb P_{\bm B_0}\!\left(J_n< \psi e^{-(D+C_0)n\epsilon_n^2}\right)\le e^{-c n\epsilon_n^2}.
\]
Moreover, using Fubini-Tonelli,
\(
\mathbb E_{\bm B_0}\big[(1-\Phi_n)J_{\mathcal B_M}\big]
=\int_{\mathcal B_M}\mathbb P_{\bm B}(1-\Phi_n)\,d\Pi(\bm B)
\)
and the Type–II bound in \eqref{eq:test-bounds},
\[
\mathbb E_{\bm B_0}\!\left[
(1-\Phi_n)J_{\mathcal B_M}\,
\mathbbm{1}\!\left\{J_n\ge \psi e^{-(D+C_0)n\epsilon_n^2}\right\}
\right]
\;\le\;
\mathbb E_{\bm B_0}\big[(1-\Phi_n)J_{\mathcal B_M}\big]
\;\le\; e^{-C_2 n\epsilon_n^2}.
\]
Putting the pieces together,
\[
\mathbb E_{\bm B_0}\big[\Pi(\mathcal B_M\mid \bm Z_n)\big]
\;\le\;
e^{-C_1 n\epsilon_n^2}
\;+\;
\psi^{-1}\,e^{(D+C_0)n\epsilon_n^2}\,e^{-C_2 n\epsilon_n^2}
\;+\;
e^{-c n\epsilon_n^2}.
\]
Choosing $\psi\in(0,1)$ fixed and noting that all constants are independent of $n$, we may pick the constant in the KL–ball radius $\kappa'$ so that $D+C_0<C_2$.
Then the middle term decays exponentially, and the RHS $\to 0$ as $n\to\infty$, \emph{uniformly} in $\bm B_0$. This proves the proposition.
\end{proof}

\medskip

\subsubsection*{Removing conditioning via ``good--sets"}
The final ingredient of the proof recipe for \switchref{\Cref{thm:unconditional}}{Theorem 1} is by proving \Cref{lem:goodset}, showing that the effective support of $\bm{Z}_n$ is $\mathcal{Z}_n$. We remove the conditioning on $\bZ_{n}$ using a conventional ``good--set" argument \citep[Chapter 1]{Ash2000}.
To elaborate, we let $\bmu_{i} = \bm{B}^\top \bx_i \in \R^{q_n}$ denote the mean of the latent variables for each $i=1,\dots,n$ and define the ``good--set''
\begin{equation}
\label{eq:goodset}
  \mathcal{Z}_n
  =
  \Bigl\{
    (z_{ik}) 
    :
      \max_{\substack{i,k:\, y_{ik}\text{ binary}}}\,
        \bigl|z_{ik}-\mu_{ik}\bigr| \le C_0\,\sqrt{\log n}
      \ \ \text{and}\
      \max_{\substack{i,k:\, y_{ik}\text{ continuous}}}\,
        |z_{ik}| \le M_0
  \Bigr\},
\end{equation}
where $C_0>0$ is suitably large, and $M_0>0$ is a uniform bound for $y_{ik}$ if the response is continuous. 
Informally, $\mathcal{Z}_n$ excludes extreme or pathological
latent draws that deviate too far from their means for binary responses (i.e., $|z_{ik}-\mu_{ik}| > 
C_0\sqrt{\log n}$) or exceed a suitable bound for continuous $z_{ik} = y_{ik}$. 

\Cref{lem:goodset} crucially shows that almost all posterior probability resides in $\mathcal{Z}_{n}$.
\begin{lemma}[High‐probability lemma for good sets]
\label{lem:goodset}
Under Assumptions (A1)–(A5), there exists a sequence $\delta_{n} \rightarrow 0$ such that $\P(\bm{Z}_{n} \notin \mathcal{Z}_{n} \mid \bY_{n}) \leq \delta_{n}$ with high $\P_{0}$-probability.
\end{lemma}

\begin{proof}[Proof of \Cref{lem:goodset}]

We treat continuous and binary responses separately.
For every continuous coordinate, the model enforces
\(z_{ik}=y_{ik}\) deterministically.  
Because each $y_{ik}$ is assumed to be sub–Gaussian (Assumptions (A3)–(A4)), the diagonal sub–Gaussian tail bound gives
\(
   \mathbb{P}\bigl(|y_{ik}|>M_0\bigr)
   \le
   2\exp\bigl\{-{M_0^{\,2}}/{2\sigma_{\max}^{2}}\bigr\},
\)
where \(\sigma_{\max}^{2}<\infty\) is an upper bound on the conditional variance.  
Choosing \(M_0\ge 4\sigma_{\max}\sqrt{\log n}\) makes this probability at most \(2n^{-8}\).  
A union bound over at most \(nq_c\le n^{3/2}\) indices (Assumption (A1): \(q_c\le q_n\le n^{b}\) with \(b<1/2\)) shows that
\[
\P  \bigl(\exists (i,k\le q_c):|z_{ik}|>M_0\bigr)
\le nq_c \cdot 2n^{-8} \le 2n^{-5},
\]
so the continuous part already satisfies the desired bound.
Conditional on \(\bm Y_n,\bm X_n\) and the parameters, each latent
\(z_{ik}\;(k>q_c)\) is a \emph{univariate truncated normal}, standard Mills–ratio bounds for a truncated normal imply that for any \(t>0\), 
\(
   \P \ \bigl(|z_{ik}-\theta_{ik}|>t\mid\bm Y_n\bigr)
   \;\le\; 2e^{-t^{2}/2}.
\)
Set \(t=C_0\sqrt{\log n}\).  Then
\[
   \P  \bigl(|z_{ik}-\theta_{ik}|>C_0\sqrt{\log n}\mid\bm Y_n\bigr)
   \;\le\; 2n^{-\,C_0^{2}/2}.
\]
Choose \(C_0\) so that
\(c:=C_0^{2}/2>b+2\).  With this choice
\(2n^{-c}\le 2n^{-2}\).
There are \(nq_b\le n q_n \le n^{1+b}\) binary latents.  
By the union bound,
\[
  \P \ \!\bigl(
        \exists(i,k>q_c):|z_{ik}-\theta_{ik}|>C_0\sqrt{\log n}
        \,\big|\,\bm Y_n
       \bigr)
  \;\le\;
  2n^{-c}\,\bigl(nq_n\bigr)
  \;\le\;
  2n^{-c+b+1}.
\]
The exponent \(-c+b+1<-1\); hence the right‐hand side is \(o(n^{-1})\).
Define
\(
  \delta_n := 2n^{-5} + 2n^{-\,c+b+1},
\)
which satisfies \(\delta_n\to0\).  
Combining the two parts,
\(
  \P \bigl(\bm Z_n\notin\mathcal Z_n\mid\bm Y_n\bigr)\le\delta_n
\)
for all sufficiently large \(n\).  This proves the lemma.

\end{proof}

\subsection{Proof of Theorem 1}
\label{sec:proofthmunconditional}
Combining all the previous lemmas, we now present the proof of \switchref{\Cref{thm:unconditional}}{Theorem 1}.

\begin{proof}
By \Cref{lem:postconcBgivenZ}, we know that for any $(\bm{Y}_n,\bm{Z}_n)$ in the \emph{good set} $\mathcal{Z}_n$, 
the posterior mass outside $\{\|\bm{B}_n-\bm{B}_0\|_F \le M\,\epsilon_n\}$ 
is at most some small $\delta_n'$.  
On the complement $\mathcal{Z}_n^c$, we do not control that posterior mass, 
but from Lemma~\ref{lem:goodset} we know 
$\mathbb{P}(\bm{Z}_n\notin \mathcal{Z}_n \mid \bm{Y}_n)\le \delta_n.$

Hence, unconditionally,
\[
\begin{aligned}
  \mathbb{P}\bigl(\|\bm{B}_n-\bm{B}_0\|_F > M\,\epsilon_n 
    \,\big|\; \bm{Y}_n\bigr)
  &= \int 
    \mathbb{P}\bigl(\|\bm{B}_n-\bm{B}_0\|_F > M\,\epsilon_n 
            \,\big|\;\bm{Y}_n,\bm{Z}_n\bigr)\;
    \mathbb{P}(\bm{Z}_n\mid \bm{Y}_n)\;
    d\bm{Z}_n
  \\[6pt]
  &\le \int_{\mathcal{Z}_n} \delta_n'\,\mathbb{P}(\bm{Z}_n\mid \bm{Y}_n)\,d\bm{Z}_n
       \;+\;
       \int_{\mathcal{Z}_n^c} 1 \,\times\,\mathbb{P}(\bm{Z}_n\mid \bm{Y}_n)\,d\bm{Z}_n
  \\[6pt]
  &\le \delta_n' + \delta_n.
\end{aligned}
\]
Because $\delta_n'\to 0$ and $\delta_n\to 0$ as $n\to\infty$, it follows that
\[
  \mathbb{P}\bigl(\|\bm{B}_n-\bm{B}_0\|_F > M\,\epsilon_n 
    \,\big|\; \bm{Y}_n\bigr)
  \;\longrightarrow\;0.
\]
Finally, taking the supremum over $\bm{B}_0$ and the expectation $\mathbb{E}_{\bm{B}_0}[\cdots]$
does not alter this limiting behavior, giving
\[
  \sup_{\bm{B}_0}\,
  \mathbb{E}_{\bm{B}_0}\Bigl[
   \mathbb{P}\Bigl(\|\bm{B}_n-\bm{B}_0\|_F > M\,\epsilon_n 
     \,\Big|\;\bm{Y}_n\Bigr)
  \Bigr]
  \;\to\;0.
\]
This completes the proof.    
\end{proof}

\subsection{Proof of Theorem 2}\label{sec:proofthm3}
The key to this proof lies in proving the following lemma.
\begin{lemma}[Conditional posterior contraction for $\Omega$]\label{lem:postconcomega}
Under the SSL prior configuration $p(\Omega)$ the posterior distribution of $\Omega$ satisfies for a sufficiently large $M>0$,
$$\sup_{\Omega_0}\mathbb{E}_{\Omega_0} \left[ \mathbb{P} \left( \lVert \Omega - \Omega_0 \rVert_{F} > M \tilde{\epsilon}_n \mid \bm{Y}_n,\bm{Z}_n \right) \right] \longrightarrow 0 $$    
\end{lemma}
\begin{proof}
Let $\bm Z_n=(\bz_1,\dots,\bz_n)^\top$ denote the latent Gaussian draws, and define residuals $\br_i=\bz_i-\bm B_0^\top\bx_i$. Conditionally on $\bm Z_n$ (equivalently on $\{\br_i\}$), inference for $\Omega$ reduces to i.i.d.\ Gaussian precision learning with $\br_i\stackrel{\text{iid}}{\sim}\mathcal N_{q_n}(\bm 0,\Sigma_0)$ and $\Omega_0=\Sigma_0^{-1}$. The proof thereafter follows the standard trajectory of showing posterior contraction using the \citet{Ghosal2000} conditions. 
Verifying the KL condition and the sieve construction closely follows the proof technique in \citet[Section 7.1]{KNHS}. Note that since $\eta$ in assumption (B2) depends on $n,$ we denote it as $\eta_n$ in our proof.

\medskip
\noindent\emph{KL control.}
Let $A=\Omega_0^{-1/2}\Omega\,\Omega_0^{-1/2}$ and write its eigenvalues as $d_1,\dots,d_{q_n}$. For Gaussian likelihoods,
\[
K\bigl(p_{\Omega_0},p_\Omega\bigr)
=\frac12\sum_{k=1}^{q_n}\bigl(-\log d_k-1+d_k\bigr),
\qquad
V\bigl(p_{\Omega_0},p_\Omega\bigr)\ \lesssim\ \sum_{k=1}^{q_n}(d_k-1)^2.
\]
Since
\[
\sum_{k=1}^{q_n}(d_k-1)^2
=\|A-I\|_F^2
=\bigl\|\Omega_0^{-1/2}(\Omega-\Omega_0)\Omega_0^{-1/2}\bigr\|_F^2
\le \|\Omega_0^{-1/2}\|_{\mathrm{op}}^4\,\|\Omega-\Omega_0\|_F^2
\le k_1^{2}\,\|\Omega-\Omega_0\|_F^2,
\]
we obtain the upper bound
\begin{equation}\label{eq:K-upper}
K\bigl(p_{\Omega_0},p_\Omega\bigr)\ \vee\ V\bigl(p_{\Omega_0},p_\Omega\bigr)
\ \le\ C_K\,\|\Omega-\Omega_0\|_F^2,
\qquad C_K:=c\,k_1^{2}.
\end{equation}
For a local \emph{lower} bound, set $E=A-I$. If $\|E\|_{\mathrm{op}}\le 1/2$, then the scalar inequality $-\log(1+t)-(-t)\ge t^2/3$ for $|t|\le 1/2$ gives
\[
2K=\sum_{k}[-\log(1+\lambda_k(E))+\lambda_k(E)]
\ \ge\ \frac13\sum_k\lambda_k(E)^2
=\frac13\|E\|_F^2.
\]
In order to ensure $\|E\|_{\mathrm{op}}\le 1/2$, see that
\[
\|E\|_{\mathrm{op}}
=\bigl\|\Omega_0^{-1/2}(\Omega-\Omega_0)\Omega_0^{-1/2}\bigr\|_{\mathrm{op}}
\le \|\Omega_0^{-1/2}\|_{\mathrm{op}}^{2}\,\|\Omega-\Omega_0\|_{\mathrm{op}}
\le k_1\,\|\Omega-\Omega_0\|_{F}.
\]
Thus, on the Frobenius ball $\{\|\Omega-\Omega_0\|_{F}\le 1/(2k_1)\}$ we have $\|E\|_{\mathrm{op}}\le 1/2$.
Seeing that the smallest singular value of $\Omega_0 ^{-1/2}$ is $1/\| \Omega_0 ^{1/2}\|_{\rm{op}} \ge 1/\sqrt{k_1}.$ Hence, 
\[
\|E\|_{F} = \|\Omega_0^{-1/2}(\Omega-\Omega_0)\Omega_0^{-1/2}\|_F \ge k_1 ^{-1} \| \Omega-\Omega_0\|_{F}.
\]
On $\{\|E\|_{\rm{op}} \le 1/2\},$ we have 
\[
2K \ge \frac{1}{3} \|E\|_F ^2 \ge \frac{1}{3k_1 ^2} \| \Omega-\Omega_0 \|_F ^2,
\]
so $K \ge (6 k_1 ^2)^{-1} \| \Omega-\Omega_0 \|_F ^2$ for $\|\Omega-\Omega_0 \|_F \le \delta_0$ (depending only on $k_1$).
Consequently, in a small Frobenius ball, $h^2(p_{\Omega_0},p_\Omega)\asymp \|\Omega-\Omega_0\|_F^2$, since the squared Hellinger distance satisfies $h^2 \le K$, and for small $K$ we also have $K \le Ch^2$ using Pinsker's inequality.

\medskip
\noindent\emph{Prior mass on a KL ball.}
Let
\[
\tilde\epsilon_n^2=\frac{(q_n+s_0^\Omega)\log q_n}{n},\qquad
\delta_n=\frac{\tilde\epsilon_n}{2k_1\sqrt{q_n+s_0^\Omega}},
\qquad
\mathcal G=\{(i,i):1\le i\le q_n\}\cup S_0^\Omega.
\]
Define the events
\[
\begin{aligned}
E_{1,n}&=\Bigl\{|\omega_{ij}-\omega_{0,ij}|\le \delta_n\ \ \forall (i,j)\in \mathcal G\Bigr\},\\
E_{2,n}&=\Bigl\{|\omega_{ij}|\le \delta_n\ \ \forall (i,j)\in T_n\Bigr\},\quad T_n\subset (S_0^\Omega)^c,\ |T_n|=\lceil\sqrt{Q}\rceil,\ Q=\tbinom{q_n}{2},\\
E_{3,n}&=\Bigl\{S_n:=\sum_{(i,j)\in R_n}\omega_{ij}^2\le \tfrac12\tilde\epsilon_n^2\Bigr\},\quad R_n=(S_0^\Omega)^c\setminus T_n.
\end{aligned}
\]
On $A_n:=E_{1,n}\cap E_{2,n}\cap E_{3,n}$ we bound $\Delta:=\Omega-\Omega_0$ as follows. Using symmetry of $\Delta$,
\[
\|\Delta\|_F^2
=\sum_{i=1}^{q_n}\Delta_{ii}^2\ +\ 2\!\!\sum_{1\le i<j\le q_n}\Delta_{ij}^{2}
=\bigl[\text{diag}\bigr]\ +\ 2\bigl[\mathcal G_{\mathrm{off}}\bigr]\ +\ 2\bigl[T_n\bigr]\ +\ 2\bigl[R_n\bigr].
\]
where we denote $[\text{diag}] \coloneqq \sum_{i=1}^{q_n} \Delta_{ii} ^2$, $[\mathcal{G}_{\text{off}}] \coloneqq \sum_{(i,j) \in \mathcal{G}_{\text{off}}} \Delta_{ij}^2,$ $[T_n] \coloneqq \sum_{(i,j) \in T_n} \Delta_{ij}^2$, and $[R_n] \coloneqq \sum_{(i,j) \in R_n} \Delta_{ij} ^2$.

On $E_{1,n}$, each diagonal and each $(i,j)\in S_0^\Omega$ has $|\Delta_{ij}|\le \delta_n$; on $E_{2,n}$ each $(i,j)\in T_n$ satisfies $|\Delta_{ij}|\le \delta_n$; and on $E_{3,n}$, $\sum_{(i,j)\in R_n}\Delta_{ij}^{2}\le \tilde\epsilon_n^2/2$. Therefore
\[
\|\Delta\|_F^2
\le q_n\delta_n^2 + 2s_0^\Omega\delta_n^2 + 2|T_n|\delta_n^2 + \tilde\epsilon_n^2
\ \le\ 4(q_n+s_0^\Omega)\delta_n^2 + \tilde\epsilon_n^2
\ =\ \frac{\tilde\epsilon_n^2}{k_1^{2}} + \tilde\epsilon_n^2
\ =\ C_*\tilde\epsilon_n^2,
\]
with $C_*:=1+1/k_1^2$. Combining with \eqref{eq:K-upper}, $K\vee V\lesssim \tilde\epsilon_n^2$ on $A_n$, and (by Weyl's inequality) $\Omega\in\mathcal M_{q_n}^{+}(k_1)$ for large $n$.

We now lower bound $\Pi(A_n)$. For diagonals, $\omega_{ii}\sim\mathrm{Exp}(\xi_1)$ and $\omega_{0,ii}\in[1/k_1,k_1]$ imply, for large $n$,
\[
\Pi\big(|\omega_{ii}-\omega_{0,ii}|\le \delta_n\big)
=\int_{\omega_{0,ii}-\delta_n}^{\omega_{0,ii}+\delta_n}\xi_1 e^{-\xi_1 x}\,dx
\ \ge\ C_1\,\delta_n.
\]
For $(i,j)\in S_0^\Omega$,
$$\begin{aligned}
\Pi\bigl(|\omega_{ij}-\omega_{0,ij}|\le \delta_n\bigr)
&\ge \eta_n\int_{\omega_{0,ij}-\delta_n}^{\omega_{0,ij}+\delta_n}\frac{\xi_1}{2}\,e^{-\xi_1|x|}\,dx\\
&\ge \eta_n\,(2\delta_n)\,\min_{x\in[\omega_{0,ij}-\delta_n,\ \omega_{0,ij}+\delta_n]}\frac{\xi_1}{2}e^{-\xi_1|x|}\\
&\ge \eta_n\,(2\delta_n)\,\frac{\xi_1}{2}\,e^{-\xi_1(\,|\omega_{0,ij}|+\delta_n\,)}\\
&= \eta_n\,\xi_1\,\delta_n\,e^{-\xi_1(|\omega_{0,ij}|+\delta_n)}.
\end{aligned}
$$
Because $\Omega_0\in\mathcal M_{q_n}^{+}(k_1)$ we have $\|\Omega_0\|_{\text{op}}\le k_1$, and thus $|\omega_{0,ij}|\le \|\Omega_0\|_{\text{op}}\le k_1$. Hence
$$
  \Pi\bigl(|\omega_{ij}-\omega_{0,ij}|\le \delta_n\bigr)
\ \ge\ \eta_n\,\xi_1\,\delta_n\,e^{-\xi_1(k_1+\delta_n)}.
$$
Since $\delta_n\to 0$, for large $n$ we can enforce $\delta_n\le 1$, and then
$e^{-\xi_1(k_1+\delta_n)}\ge e^{-\xi_1(k_1+1)}=:C_2/\xi_1$.
This gives the bound 
\[
\Pi\big(|\omega_{ij}-\omega_{0,ij}|\le \delta_n\big)
\ \ge\ \eta_n\,\xi_1\,\delta_n\,e^{-\xi_1(k_1+\delta_n)}
\ \ge\ C_2\,\eta_n\,\delta_n.
\]
with a constant $C_2:=\xi_1\,e^{-\xi_1(k_1+1)}\in(0,\infty),$ which is independent of $n$.

Now, independence gives
\[
\Pi(E_{1,n})
\ \ge\ (C_1\delta_n)^{q_n}\,(C_2\eta_n\delta_n)^{s_0^\Omega}
\ =\ \exp\bigl\{-C(q_n+s_0^\Omega)\log q_n\bigr\}
\ =\ \exp\{-C\,n\,\tilde\epsilon_n^2\},
\]
using $\delta_n=\tilde\epsilon_n/(2k_1\sqrt{q_n+s_0^\Omega})$ and assumption (B2).

For $(i,j)\in T_n$, (B3) implies $\xi_0\delta_n\asymp \sqrt{q_n}\to\infty$, so for large $n$,
\[
\Pi\big(|\omega_{ij}|\le \delta_n\big)
=(1-\eta_n)\bigl(1-e^{-\xi_0\delta_n}\bigr)+\eta_n\bigl(1-e^{-\xi_1\delta_n}\bigr)
\ \ge\ \tfrac12\cdot \tfrac34\ \ge\ \tfrac{3}{8},
\]
and hence
\[
\Pi(E_{2,n})\ \ge\ (3/8)^{|T_n|}\ \ge\ \exp\{-C\,q_n\}.
\]
The bound $3/8$ stems from the fact that $1-\exp(-\xi_0 \delta_n) \uparrow 1$, so for all sufficiently large $n$, we have $1-\exp(-\xi_0 \delta_n) \ge 3/4$ and again by (B2), we have for all large $n$, $1-\eta_n \ge 1/2$. Thus, it follows by observing that $\Pi(|\omega_{ij}| \le \delta_n) \ge (1-\eta_n)(1-\exp(-\xi_0 \delta_n))$.

For $E_{3,n}$, let $\mathcal E_{\mathrm{spike}}$ be the event that all $(i,j)\in R_n$ are drawn from the spike. Then
\[
\Pi(\mathcal E_{\mathrm{spike}}^c)
=1-(1-\eta_n)^{|R_n|}
\ \le\ |R_n|\eta_n
\ \lesssim\ q_n^2\cdot \frac{\log q_n}{n q_n}
\ =\ \frac{q_n\log q_n}{n}
\ =\ o(1),
\]
so $\Pi(\mathcal E_{\mathrm{spike}})\ge 1/2$ for large $n$. Conditional on $\mathcal E_{\mathrm{spike}}$, $\omega_{ij}\overset{\text{iid}}{\sim}\mathrm{Laplace}(\xi_0)$, whence $\omega_{ij}^2$ are sub–exponential with $\mathbb E[\omega_{ij}^2]=2/\xi_0^2$ and $\mathrm{Var}(\omega_{ij}^2)=20/\xi_0^4$. Thus
\[
\mathbb E[S_n\mid\mathcal E_{\mathrm{spike}}]
=|R_n|\cdot\frac{2}{\xi_0^2}
\ \lesssim\ \frac{q_n^2}{\xi_0^2}
\ \lesssim\ \frac{q_n\log q_n}{n}
\ \lesssim\ \tilde\epsilon_n^2.
\]
Bernstein’s inequality yields $\Pi\bigl(S_n\le \tfrac12\tilde\epsilon_n^2\mid\mathcal E_{\mathrm{spike}}\bigr)\ge 1/2$ (for large $n$), hence $\Pi(E_{3,n})\ge 1/4$. Altogether,
\begin{equation}\label{eq:KL-mass}
\Pi(A_n)\ \ge\ \exp\{-C\,(q_n+s_0^\Omega)\log q_n\}
\ =\ \exp\{-C\,n\,\tilde\epsilon_n^2\},
\end{equation}
and since $A_n\subset\{K\vee V\lesssim \tilde\epsilon_n^2\}$, \eqref{eq:KL-mass} is a KL small–ball lower bound.

\medskip
\noindent\emph{Sieve and metric entropy.}
Define
\[
\mathcal P_n
=\Bigl\{\Omega\in\mathcal M_{q_n}^{+}(k_1):\
\#\{(i,j):1\le i<j\le q_n,\ |\omega_{ij}|>\nu_n\}\le r_n,\
\|\Omega\|_\infty\le B_n\Bigr\},
\]
with thresholds
\[
\nu_n=\frac{\tilde\epsilon_n}{q_n^{b'}},\quad b'>0,\qquad
r_n=c_1\,\frac{n\tilde\epsilon_n^2}{\log n},\qquad
B_n=c_2\,n\,\tilde\epsilon_n^2,
\]
and $Q=\binom{q_n}{2}$. We cover $\mathcal P_n$ in the parameter sup–norm. Each $\Omega\in\mathcal P_n$ admits a support $S\subset\{(i,j):i<j\}$ with $|S|\le r_n$ containing the off–diagonals larger than $\nu_n$. The number of supports is $\sum_{j=0}^{r_n}\binom{Q}{j}$.

Fix $S$ with $|S|=j\le r_n$. For the $q_n$ diagonals and the $j$ selected off–diagonals, each coordinate lies in an interval of length at most $2B_n$. Placing an equispaced grid of mesh $\nu_n$ yields at most $(C\,B_n/\nu_n)^{q_n+j}$ grid points covering those coordinates within sup–norm $\nu_n$. Off–diagonals outside $S$ are already bounded by $\nu_n$, so setting them to $0$ incurs at most $\nu_n$ sup–norm error.

Summing over $S$ and using $\sum_{j=0}^{r_n}\binom{Q}{j}\le \bigl(\tfrac{eQ}{r_n}\bigr)^{r_n}$,
\[
N\bigl(\nu_n,\mathcal P_n,\|\cdot\|_\infty\bigr)
\ \le\ 
(r_n+1)\,\Bigl(C\,\frac{B_n}{\nu_n}\Bigr)^{q_n+r_n}
\Bigl(\frac{eQ}{r_n}\Bigr)^{r_n}.
\]
Therefore
\begin{equation}\label{eq:logN-raw}
\log N\bigl(\nu_n,\mathcal P_n,\|\cdot\|_\infty\bigr)
\ \lesssim\ (q_n+r_n)\,\log\frac{B_n}{\nu_n}
\ +\ r_n\,\log\frac{Q}{r_n}.
\end{equation}
Now
\[
\frac{B_n}{\nu_n}
=\frac{c_2\,n\,\tilde\epsilon_n^2}{\tilde\epsilon_n/q_n^{b'}}
=c_2\,n\,\tilde\epsilon_n\,q_n^{b'},
\qquad
\log\frac{B_n}{\nu_n}
=O(\log n+\log q_n),
\]
since $\tilde\epsilon_n^2=\frac{(q_n+s_0^\Omega)\log q_n}{n}$ implies $\log(n\tilde\epsilon_n)=O(\log n+\log q_n)$. Also, $Q\asymp q_n^2$ and $r_n=c_1\frac{n\tilde\epsilon_n^2}{\log n}$ yield
\[
\log\frac{Q}{r_n}
=O(\log q_n+\log n).
\]
Plugging into \eqref{eq:logN-raw},
\[
\log N\bigl(\nu_n,\mathcal P_n,\|\cdot\|_\infty\bigr)
\ \lesssim\ (q_n+r_n)\,(\log n+\log q_n)
\ \lesssim\ n\,\tilde\epsilon_n^2,
\]
because $r_n\log n\asymp n\tilde\epsilon_n^2$ and $q_n\log q_n\lesssim n\tilde\epsilon_n^2$.

On $\mathcal M_{q_n}^+(k_1)$, $h(p_\Omega,p_{\Omega'})\le C(k_1)\,\|\Omega-\Omega'\|_F\le C(k_1)\,q_n\,\|\Omega-\Omega'\|_\infty$. Hence
\[
\log N\bigl(\tilde\epsilon_n,\mathcal P_n,h\bigr)
\ \le\ \log N\Bigl(\frac{\tilde\epsilon_n}{C(k_1)\,q_n},\mathcal P_n,\|\cdot\|_\infty\Bigr)
\ \lesssim\ \log N\bigl(\nu_n,\mathcal P_n,\|\cdot\|_\infty\bigr)
\ \lesssim\ n\,\tilde\epsilon_n^2,
\]
since $\nu_n=\tilde\epsilon_n/q_n^{b'}$ with $b'>0$ is no larger than a constant multiple of $\tilde\epsilon_n/q_n$ for large $n$.

\medskip
\noindent\emph{Prior outside the sieve.}
Under the untruncated prior $\Pi^*$, let $N=\sum_{i<j}\mathbbm{1}(|\omega_{ij}|>\nu_n)$ with $Q=\binom{q_n}{2}$. Then
\[
p_{ij}:=\Pi^*\bigl(|\omega_{ij}|>\nu_n\bigr)
=(1-\eta_n)e^{-\xi_0\nu_n}+\eta_n e^{-\xi_1\nu_n}\ \le\ e^{-\xi_0\nu_n}+\eta_n.
\]
By (B3), $\xi_0\nu_n\asymp \sqrt{q_n}\to\infty$, so $e^{-\xi_0\nu_n}$ is negligible; and (B2) gives $p_{ij}\le 2\eta_n\lesssim \frac{\log q_n}{n q_n}$. Thus $\mu:=\mathbb E[N]\le Q\cdot \frac{2\log q_n}{n q_n}\lesssim \frac{q_n\log q_n}{n}=O(\tilde\epsilon_n^2)$. A Chernoff bound yields
\[
\Pi^*(N\ge r_n+1)
\ \le\ \Bigl(\frac{e\mu}{r_n}\Bigr)^{r_n}
\ \le\ \exp\{-c\,n\,\tilde\epsilon_n^2\}.
\]
For the sup–norm,
\[
\Pi^*(\|\Omega\|_\infty>B_n)
\ \le\ (Q+q_n)\Bigl((1-\eta_n)e^{-\xi_0 B_n}+\eta_n e^{-\xi_1 B_n}\Bigr)
\ \le\ \exp\{-c'\,n\,\tilde\epsilon_n^2\},
\]
for $c_2$ large. Hence $\Pi^*(\mathcal P_n^c)\le \exp\{-c_4 n\tilde\epsilon_n^2\}$. Following \citet{KNHS}, truncation to $\mathcal M_{q_n}^{+}(k_1)$ only multiplies by a constant factor (the denominator $\Pi^*(\mathcal M_{q_n}^{+}(k_1))\ge c>0$), so $\Pi(\mathcal P_n^c)\le \exp\{-c_5 n\tilde\epsilon_n^2\}$.

\end{proof}
We complete the proof of Theorem 2 by removing the conditioning with respect to $\bm{Z}$ in \Cref{sec:proofthm3} using the same good--set argument mimicking the proof of Theorem 1.

\subsection{Proof of Theorem 3}
The proof mirrors the arguments of the proof of Theorem 3 in \citet{wang2023twostep}. 
\begin{proof}
By definition,
\[
\widetilde{\beta}_{j,k}\neq 0
\quad\Longleftrightarrow\quad
\frac{\lvert \beta_{j,k}\rvert}{\omega_{k,k}} \;>\; a_n,
\]
where $a_n = c\,\frac{\sqrt{\log p_n}}{\sqrt{n}\,p_n}$. Hence, to show $S_0\subseteq \widetilde{S}$ (sure screening), we need to ensure that for each $(j,k)\in S_0$, the above ratio exceeds $a_n$ with high probability.

From standard contraction arguments in Theorem 1, under assumptions (A1)--(A6), we obtain that 
\[
    \max_{j,k}\,\lvert \beta_{j,k} - (\beta_{j,k})_0\rvert \le M \epsilon_n
\]
since, we know that $\max_{j,k} |\beta_{jk}| \le \lVert\bm{B}\rVert_{F}$.
Meanwhile, Theorem 2 shows that whp, $\Omega$ is close to $\Omega_0$ in Frobenius norm, and crucially $\omega_{k,k}$ remains bounded away from zero because $\Omega$ stays positive‐definite.  By assumption (A5), $\Omega_0$ has all eigenvalues bounded and away from $0$, so for sufficiently large $n$,
\[
    \omega_{k,k}
    \;\in\;
    \bigl[\underline{\omega},\,\overline{\omega}\bigr]
    \;\subset\;(0,\infty)
    \quad
    \text{with high probability}.
\]

By assumption (C1), each true nonzero coefficient satisfies 
\[
    \lvert (\beta_{j,k})_0\rvert 
    \;\ge\;
    \frac{c_3}{n^\zeta}
    \quad
    \text{for }\,(j,k)\in S_0.
\]
Given the posterior contraction result in \switchref{\Cref{thm:unconditional}}{Theorem 1}, if we pick any small $\delta>0$, then  for large $n$, with probability going to 1,
\[
    \lvert \beta_{j,k} - (\beta_{j,k})_0\rvert 
    \;\le\; 
    \delta\,\frac{c_3}{n^\zeta}.
\]
Hence
\[
    \lvert \beta_{j,k}\rvert
    \;\ge\;
    \lvert (\beta_{j,k})_0\rvert 
    - 
    \delta\,\frac{c_3}{n^\zeta}
    \;\ge\;
    (1-\delta)\,\frac{c_3}{n^\zeta}
    \quad
    \text{for large }n.
\]
On the other hand, w.h.p. we also have
\[
    \omega_{k,k}
    \;\le\;
    \overline{\omega}
    \quad
    \text{for some constant }\overline{\omega}>0.
\]
Thus
\[
    \frac{\lvert \beta_{j,k}\rvert}{\omega_{k,k}}
    \;\;\ge\;
    \frac{(1-\delta)\,c_3}{\,\overline{\omega}\,n^\zeta\,}.
\]
Comparing this with the threshold $a_n = c\,\nicefrac{\sqrt{\log p_n}}{\sqrt{n}\,p_n}$, note that for $\zeta<1/4$ and $p_n$ growing sub‐exponentially in $n$, we have
\[
    n^{-\zeta}
    \;\gg\;
    \frac{\sqrt{\log p_n}}{\sqrt{n}\,p_n}.
\]
Therefore, for sufficiently large $n$,
\[
    \frac{\lvert \beta_{j,k}\rvert}{\omega_{k,k}}
    \;>\;
    a_n
    \quad
    \text{with high probability}.
\]
Consequently, $\beta_{j,k}$ is \emph{not} thresholded to zero, i.e.\ $\widetilde{\beta}_{j,k}\neq 0$ for every $(j,k)\in S_0$. This shows that all truly nonzero entries survive the screening step.  In other words,
\[
    \mathbb{P}\bigl(S_0\subseteq\widetilde{S}\bigr)
    \;\to\;
    1
    \quad
    \text{as }n\to\infty.
\]
Taking supremum over $\bm{B}_0$ just ensures uniformity across all possible true parameter choices.  Hence
\[
    \sup_{\bm{B}_0}\,
    \mathbb{E}_{\bm{B}_0}\!
    \Bigl[\,
      \mathbb{P} \ \!\bigl(S_0\subseteq\widetilde{S}\mid \bm{Y}_n\bigr)
    \Bigr]
    \;\longrightarrow\; 1,
\]
which completes the proof of the sure‐screening property.
\end{proof}

%% file: extraalgorithm.tex
\subsection{Tackling non-concavity with an EM-algorithm}
\label{sec:slabprobs}

To sidestep the challenges presented by the non-concave penalty term $\log p(\bm{B}, \Omega \vert \theta, \eta)$ we follow the general strategy of \citet{Deshpande2019_mSSL} and augment our model with two collections of binary indicators $\bdelta^{(\beta)} = \{\delta^{(\beta)}_{j,k}: 1 \leq j \leq p, 1 \leq k \leq q\}$ and $\bdelta^{(\omega)} = \{\delta^{(\omega)}_{k,k'}: 1 \leq k < k' \leq q\}.$
Each $\delta^{(\beta)}_{j,k}$ and $\delta^{(\omega)}_{k,k'}$ encodes whether or not the elements $\beta_{j,k}$ and $\omega_{k,k'}$ are drawn from their respective slab ($\delta = 1$) or spike ($\delta = 0$) distributions.

To motivate this approach, first observe that the conditional density of $\bm{B}$ and $\Omega$ can be expressed as
$$
p(\bm{B}, \Omega \vert \theta, \eta) = \int{p(\bm{B} \vert \bdelta^{(\beta)})p(\bdelta^{(\beta)}|\theta)d\bdelta^{(\beta)}} \times \int{p(\Omega \vert \bdelta^{(\omega)})p(\bdelta^{(\omega)} \vert \eta)d\bdelta^{(\omega)}}
$$
where each $\delta^{(\beta)}_{j,k} \vert \theta \sim \berndist{\theta}$ and $\delta^{(\omega)}_{k,k'} \vert \eta \sim \berndist{\eta},$ independently and
\begin{align*}
p(\bm{B} \vert \bdelta^{(\beta)}) &\propto \exp\left\{-\sum_{j = 1}^{p}{\sum_{k = 1}^{q}{\left(\lambda_{1}\delta_{j,k}^{(\beta)} + \lambda_{0}(1 - \delta_{j,k}^{(\beta)})\right)\lvert \beta_{j,k}\rvert}} \right\} \\
p(\Omega \vert \bdelta^{(\omega)}) &\propto \ind{\Omega \succ 0} \times \exp\left\{-\xi_{1}\sum_{k = 1}^{q}{\omega_{k,k}} - \sum_{1 \leq k < k' \leq q}{\left(\xi_{1}\delta_{k,k'}^{(\omega)} + \xi_{0}(1 - \delta_{k,k'}^{(\omega)})\right)\lvert \omega_{k,k'}\rvert}\right\}
\end{align*}

Based on this representation, an ECM algorithm for approximating the MAP of $\Xi := (\bm{B}, \theta, \Omega, \eta)$ proceeds by iteratively optimizing a surrogate objective function defined by integrating $\log p(\Xi \vert \bY)$ against the conditional posterior distributions of the indicators.
Specifically, starting from some initial guess $\Xi^{(0)},$ for $t > 1,$ the $t$-th iteration of the algorithm consists of an E-step, in which we compute the surrogate objective
$$
F^{(t)}(\bm{B}, \theta, \Omega, \eta) = \E_{\bdelta^{(\beta)}, \bdelta^{(\omega)} \vert \cdot}[\log p(\Xi, \bdelta^{(b)}, \bdelta^{(\omega)} \vert \bY) \vert \Xi^{(t-1)}]
$$
and two conditional maximization (CM) steps.
In the first CM step, we maximize $F^{(t)}$ with respect to $(\bm{B}, \theta)$ while fixing $(\Omega, \eta) = (\Omega^{(t-1)}, \eta^{(t-1)}).$
Then, in the second CM step, we maximize $F^{(t)}$ with respect to $(\Omega, \eta)$ while fixing $(\bm{B}, \theta) = (\bm{B}^{(t)}, \theta^{(t)}).$

It is straightforward to verify that the indicators $\delta^{(\beta)}_{j,k}$'s (resp. $\delta^{(\omega)}_{k,k'})$  are conditionally independent given $\bm{B}$ and $\theta$ (resp. $\Omega$ and $\eta$).
Moreover, their conditional expectations are available in closed form:
\begin{align}
\begin{split}
\label{eq:pstar_qstar}
\E[\delta^{(\beta)}_{j,k} \vert \bm{B}, \theta] := p^{\star}_{j,k} &= \left[1 + \frac{1-\theta}{\theta} \times \frac{\lambda_{0}}{\lambda_{1}} \times \exp\left\{-(\lambda_{0} - \lambda_{1})\lvert \beta_{j,k} \rvert \right\}\right]^{-1} \\
\E[\delta^{(\omega)}_{k,k'} \vert \Omega, \eta] := q^{\star}_{k,k'} &= \left[1 + \frac{1-\eta}{\eta} \times \frac{\xi_{0}}{\xi_{1}} \times \exp\left\{-(\xi_{0} - \xi_{1})\lvert \omega_{k,k'} \rvert \right\}\right]^{-1}
\end{split}
\end{align}

With this notation, the surrogate objective optimized in the $t$-th iteration of the ECM algorithm which is also \switchref{\Cref{eq:ecm_surrogate}}{Equation (3) in the main text}.
\begin{align}
\begin{split}
\label{eq:ecm_surrogate}
F^{(t)} &= \sum_{i = 1}^{n}{\log p(\by_{i} \vert \bm{B}, \Omega)} -\sum_{j=1}^{p}{\sum_{k=1}^{q}{\lambda^{\star}_{j,k}\lvert \beta_{j,k} \rvert }} - \sum_{1 \leq k < k' \leq q}{\xi^{\star}_{k,k'}\lvert \omega_{k,k'} \rvert} - \xi_{1}\sum_{k = 1}^{q}{\omega_{k,k}} \\
&+ \left(a_\theta -1 + \sum_{j = 1}^{p}{\sum_{k = 1}^{q}{{p^{\star}_{j,k}}}}\right)\log\theta + \left(b_\theta -1 + pq - \sum_{j = 1}^{p}{\sum_{k = 1}^{q}{{p^{\star}_{j,k}}}}\right)\log(1-\theta) \\
& + \left(a _{\eta} - 1 + \sum_{1 \leq k < k' \leq q}{q^{\star}_{k,k'}}\right) + \left(b_{\eta} - 1 + \frac{q(q-1)}{2} - \sum_{1 \leq k < k' \leq q}{q^{\star}_{k,k'}}\right)\log(1-\eta),
\end{split}
\end{align}
where $\lambda^{\star}_{j,k} = \lambda_{1}p^{\star}_{j,k} + \lambda_{0}(1 - p^{\star}_{j,k})$ and $\xi^{\star}_{k,k'} = \xi_{1}q^{\star}_{k,k'} + \xi_{0}(1 - q^{\star}_{k,k'}).$ 
From the first line of \Cref{eq:ecm_surrogate}, optimizing $F^{(t)}$ with respect to $\bm{B}$ or $\Omega$ involves solving a maximum likelihood problem with individual penalties $\ell_{1}$ in each entry $\beta_{j,k}$ and $\omega_{k,k'}.$
A useful way to view the coefficients $\lambda_{j,k}^{\star}$ and $\xi_{k,k'}^{\star}$ is as entry-specific data-driven penalty levels. 
Because $p_{j,k}^{\star}$ and $q_{k,k'}^{\star}$ are updated at every E–step according to the current magnitudes of $\beta_{j,k}$ and $\omega_{k,k'}$  (\Cref{eq:pstar_qstar}), the effective penalties $\lambda_{j,k}^{\star}$ and $\xi_{k,k'}^{\star}$ \emph{adaptively mix} between a strong spike $(\lambda_0,\xi_0)$ and a weak slab $(\lambda_1,\xi_1)$. 
Entries that appear important in the current iterate (i.e., those with large magnitude) receive small penalties in the next CM–step and are allowed to persist, whereas entries close to zero are pushed more aggressively toward exact zero in the next iteration.
This adaptive penalty mixing is the key mechanism by which the spike-and-slab LASSO can aggressively shrink negligible parameter values without severely biasing the estimates of truly significant values \citep[see][for an overview]{GeorgeRockova2020_penalty_mixing}.

\subsection{Dual-thresholding rule derivation}
\label{sec:dualthreshold}
Fix $\Omega=\Omega^{(t-1)}$ and all coefficients except $\beta_{jk}$.
We write the partial residuals (excluding predictor $j$) as
\[
r_{ik'}^{(-j,h)}
\;=\;
z_{ik'}^{(h)}-\sum_{\ell\neq j}x_{i\ell}\beta_{\ell k'}^{\text{old}},
\quad k'=1,\dots,q.
\]
Let
\[
S_{jk}
\;:=\;
\sum_{i=1}^n\sum_{h=1}^H x_{ij}\,\Big(\sum_{k'=1}^q\omega_{kk'}^{(t-1)}\,r_{ik'}^{(-j,h)}\Big),
\quad
a_{jk}\;:=\;nH\,\omega_{kk}^{(t-1)}.
\]
With these definitions, the part of the surrogate objective
$\tilde F^{(t)}(\bm B,\theta,\Omega^{(t-1)},\eta^{(t-1)})$
that depends on $\beta_{jk}$ alone is (up to an additive constant)
\[
Q_{jk}(\beta)
\;=\;
-\frac{a_{jk}}{2}\,\beta^2 \;+\; S_{jk}\,\beta \;-\; \lambda_{jk}^\star\,|\beta|,
\]
where $\lambda_{jk}^\star=\lambda_1 p^\star_{jk}+\lambda_0(1-p^\star_{jk})$ is fixed in the CM–step.

\paragraph{KKT conditions.}
Let $\partial|\beta|=\{\operatorname{sign}(\beta)\}$ for $\beta\neq0$ and
$\partial|0|=[-1,1]$. Similar to \citet{Deshpande2019_mSSL}, the first–order (concave) KKT conditions for a maximizer $\hat\beta=\beta_{jk}^{\text{new}}$ are
\[
\begin{cases}
-a_{jk}\,\hat\beta + S_{jk} - \lambda_{jk}^\star\,\operatorname{sign}(\hat\beta)=0,
&\text{if }\hat\beta\neq 0,\\[2mm]
S_{jk}\in[-\lambda_{jk}^\star,\ \lambda_{jk}^\star],
&\text{if }\hat\beta=0.
\end{cases}
\]
From the second line, $\hat\beta=0$ is optimal when $|S_{jk}|\le \lambda_{jk}^\star$.

If $|S_{jk}|>\lambda_{jk}^\star$, the first line yields
\[
\hat\beta
\;=\;
\frac{S_{jk}-\lambda_{jk}^\star\,\operatorname{sign}(S_{jk})}{a_{jk}}
\;=\;
\frac{\big(|S_{jk}|-\lambda_{jk}^\star\big)_{+}}{a_{jk}}\,
\operatorname{sign}(S_{jk}).
\]
Defining the (per–coordinate) threshold
\(
\Delta_{jk}\;:=\;\frac{\lambda_{jk}^\star}{a_{jk}}
=\frac{\lambda_{jk}^\star}{nH\,\omega_{kk}^{(t-1)}},
\)
then the update can be written as a \emph{dual-threshold} rule:
\[
\beta_{jk}^{\text{new}}
\;=\;
\Big[|S_{jk}|-\lambda_{jk}^\star\Big]_{+}\,
\frac{\operatorname{sign}(S_{jk})}{nH\,\omega_{kk}^{(t-1)}}\;
\mathbbm{1}\!\left(\frac{|S_{jk}|}{nH\,\omega_{kk}^{(t-1)}}>\Delta_{jk}\right)
\]
i.e., a hard-threshold at $\Delta_{jk}$ (which is equivalent to
$|S_{jk}|>\lambda_{jk}^\star$) followed by a soft-threshold of size
$\lambda_{jk}^\star/(nH\,\omega_{kk}^{(t-1)})$.

As a sanity check, when $|S_{jk}|>\lambda_{jk}^\star$, plugging $\hat\beta$ back into $Q_{jk}$ gives
\[
Q_{jk}(\hat\beta)-Q_{jk}(0)
=\frac{\big(|S_{jk}|-\lambda_{jk}^\star\big)^2}{2\,a_{jk}}
>0,
\]
and otherwise the maximizer is $\hat\beta=0$. Thus the update is the unique coordinatewise maximizer of the concave surrogate.

\subsection{Implementation details of the MCECM algorithm}\label{sec:implementationdetails}
Our MCECM algorithm depends on three sets of hyperparameters.
The first set are the spike and slab penalties $\xi_{0}, \lambda_{0}, \xi_{1},$ and $\lambda_{1}.$
We recommend setting $\lambda_{1} \approx (\sqrt{n\log{n}})^{-1},$ which induces an amount of shrinkage similar to that induced by the global shrinkage parameter $\tau$ in \citet{wang2023twostep}'s \texttt{mt-MBSP} procedure.
We further recommend setting $\xi_{1} = n/100,$ similar to \citet{Deshpande2019_mSSL}.
Instead of fixing single values for the spike penalties $\lambda_{0}$ and $\xi_{0},$ we run our MCECM algorithm along a grid of $(\lambda_{0}, \xi_{0})$ combinations with warm starts.
We vary $\lambda_{0}$ and $\xi_{0}$ along grid of ten evenly-spaced values respectively ranging from $10$ to $100$ and from $n/10$ and  $n$.
In this way, we optimize $100$ increasingly spiky posterior distributions, one for each combination of $\lambda_{0}$ and $\xi_{0}.$
Such dynamic posterior exploration has proven extremely effective in gradually filtering out negligible parameter values \citep{RockovaGeorge2018_ssl, Deshpande2019_mSSL}. 

The hyperparameters $a_{\theta}, b_{\theta}, a_{\eta}$ and $b_{\eta}$ influence the overall sparsity of the outputted solution.
Following the suggestions of \citet{shencgssl} and \citet{Deshpande2019_mSSL}, we recommend setting $a_\theta = 1, \ b_\theta = pq$, and $a_\eta=1, \ b_\eta=q$.  
Finally, we draw $H = 2000$ realizations of $\bz_{i}$ in each MC E-step of our algorithm.

%% file: additional_experiments.tex
\subsection{Simulating the different covariance structures}\label{sec:covstr}
Following the settings discussed in \citet{shencgssl}, for the AR$(1)$ model, we specifically set $(\Omega^{-1})_{k,k'} = 0.7^{|k-k'|}$ so that $\omega_{k,k'}=0$ whenever $|k-k'| > 1$. In the AR$(2)$ model, we set $\omega_{k,k}=1,\ \omega_{k-1,k}=\omega_{k,k-1}=0.5$, and $\omega_{k-2,k}=\omega_{k,k-2}=0.25$. The block model is obtained by partitioning $\Omega^{-1}$ into four $q/2 \times q/2$ blocks and setting all entries in the off-diagonal blocks to zero. Then, we take $\sigma_{k,k} = 1$ and $\sigma_{k,k'}=0.5$ for $1 \le k \neq k' \le q/2$ and for $q/2 +1 \le k \neq k' \le q$. For the star graph, we choose $\omega_{k,k}=1, \ \omega_{1,k}=\omega_{k,1}=0.1$ for each $k>1$, and force the remaining off-diagonal elements of $\Omega$ equal to zero. The small-world and tree networks necessitate the initial generation of an appropriate random graph and then drawing $\Omega$ from a G-Wishart distribution with three degrees of freedom and an identity scale matrix \citep{Roverato2002}. The Watts-Strogatz model \citep{Watts1998} with a single community and rewiring probability of $0.1$ was used to generate the small-world graph. Finally, the tree graph was generated by running a loop-erased random walk on a complete graph.  

\subsection{Additional simulation results}

\begin{figure}[ht!]
    \centering
    \includegraphics[width=\linewidth]{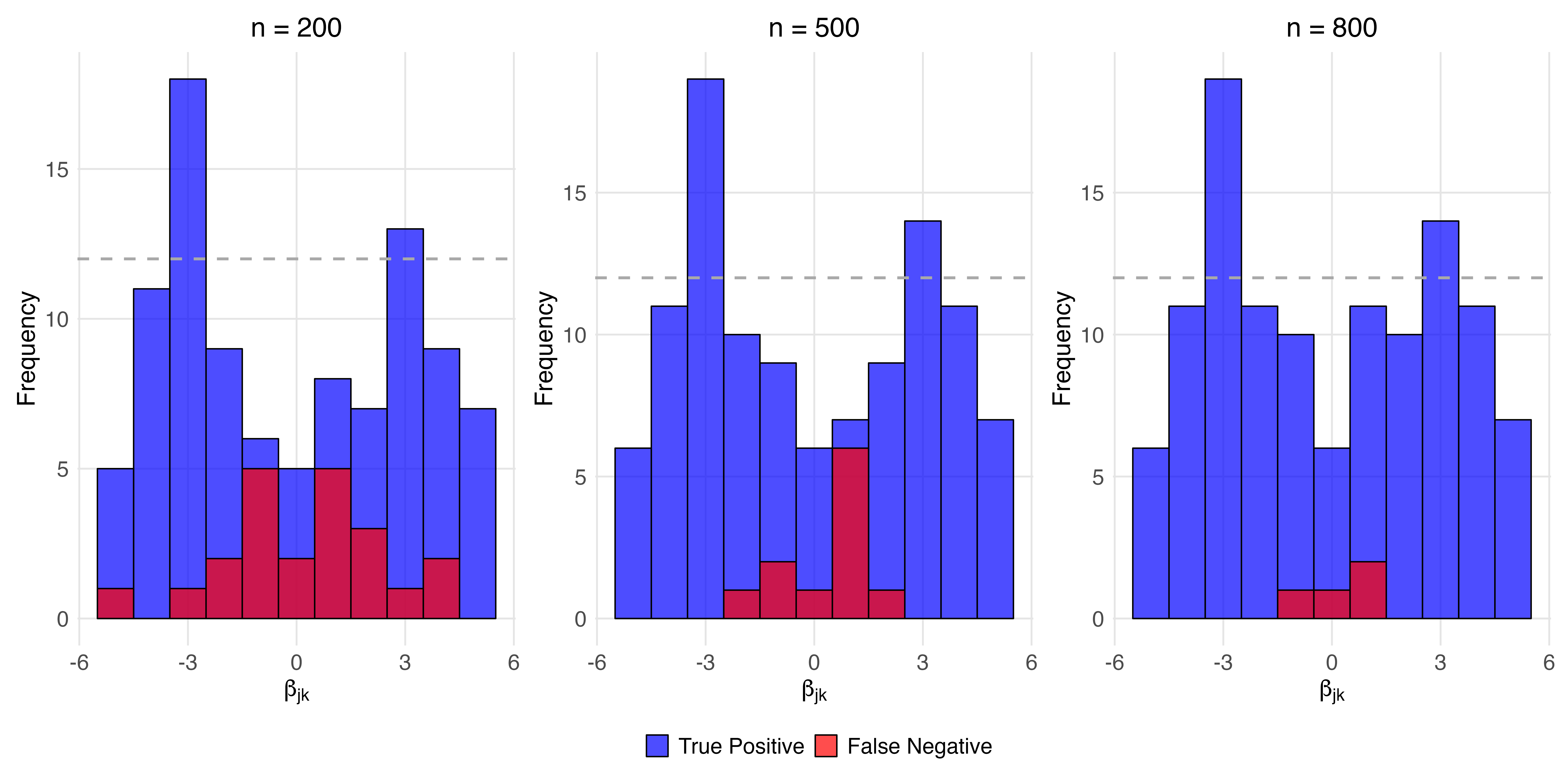}
    \caption{Variation of sensitivity for \texttt{mixed-mSSL} as the sample size $n$ increases for fixed $p,q$.}
    \label{fig:ARsensplot}
\end{figure}
The left panel of Figure \ref{fig:ARsensplot} shows a histogram of \texttt{mixed-mSSL} true positives (blue) and false negatives (red) for a single AR(1) replication. It pinpoints the source of lower \texttt{SEN} under $\mathcal U[-5,5]$: limited power for very small signals at $n=200$. Holding $(\bm B,\Omega)$ fixed and increasing $n$ from $200$ to $500$ and $800$ concentrates the false–negative $\beta_{jk}$ values closer to $0$, i.e., detection of small effects improves with $n$, as anticipated by Theorem 5.
\begin{table}[htbp]
\centering
\caption{Support‐recovery for $\Omega$ under \texttt{mixed-mSSL} across six graph structures and three $(n,p,q)$ settings.  For each block, we report sensitivity (SEN), specificity (SPEC), precision (PREC) and accuracy (ACC) under two signal regimes: $\mathcal{U}[-5,5]$ vs. $\mathcal{U}[-5,-2]\cup[2,5]$.}
\label{table:omega_support}
\small
\begin{tabular}{@{}lcccccccc@{}}
\toprule
\multicolumn{1}{c}{\textbf{Scenario}}
 & \multicolumn{4}{c}{$\mathcal{U}[-5,5]$}
 & \multicolumn{4}{c}{$\mathcal{U}[-5,-2]\cup[2,5]$} \\
\cmidrule(lr){2-5} \cmidrule(lr){6-9}
 & SEN & SPEC & PREC & ACC
 & SEN & SPEC & PREC & ACC \\
\midrule
\multicolumn{9}{c}{$(n,p,q) = (200,500,4)$}\\
\midrule
AR1           & 0.23 & 0.81 & 0.51 & 0.52 & 0.26 & 0.76 & 0.52 & 0.51 \\
AR2           & 0.22 & 0.77 & 0.83 & 0.31 & 0.25 & 0.77 & 0.84 & 0.33 \\
Block Diagonal& 0.25 & 0.80 & 0.35 & 0.62 & 0.25 & 0.76 & 0.33 & 0.59 \\
Star Graph    & 0.19 & 0.75 & 0.43 & 0.47 & 0.20 & 0.72 & 0.42 & 0.46 \\
Small World   & 0.23 & 0.80 & 0.70 & 0.42 & 0.25 & 0.75 & 0.69 & 0.42 \\
Tree Network  & 0.94 & NaN & 1.00 & 0.94 & 0.94 & NaN & 1.00 & 0.94 \\
\midrule
\multicolumn{9}{c}{$(n,p,q) =(500,1000,4)$}\\
\midrule
AR1           & 0.52 & 0.58 & 0.58 & 0.55 & 0.47 & 0.54 & 0.50 & 0.50 \\
AR2           & 0.48 & 0.63 & 0.89 & 0.51 & 0.46 & 0.53 & 0.83 & 0.48 \\
Block Diagonal& 0.49 & 0.56 & 0.40 & 0.54 & 0.49 & 0.55 & 0.34 & 0.53 \\
Star Graph    & 0.43 & 0.50 & 0.45 & 0.46 & 0.45 & 0.50 & 0.46 & 0.46 \\
Small World   & 0.49 & 0.52 & 0.68 & 0.50 & 0.50 & 0.53 & 0.68 & 0.51 \\
Tree Network  & 0.47 & NaN & 1.00 & 0.47 & 0.47 & NaN & 1.00 & 0.47 \\
\midrule
\multicolumn{9}{c}{$(n,p,q) = (800,1000,6)$}\\
\midrule
AR1           & 0.18 & 0.81 & 0.34 & 0.60 & 0.31 & 0.66 & 0.32 & 0.55 \\
AR2           & 0.20 & 0.82 & 0.66 & 0.45 & 0.30 & 0.66 & 0.56 & 0.45 \\
Block Diagonal& 0.17 & 0.81 & 0.37 & 0.56 & 0.30 & 0.66 & 0.36 & 0.52 \\
Star Graph    & 0.22 & 0.83 & 0.40 & 0.63 & 0.29 & 0.67 & 0.30 & 0.55 \\
Small World   & 0.18 & 0.80 & 0.40 & 0.56 & 0.32 & 0.68 & 0.39 & 0.53 \\
Tree Network  & 0.36 & NaN & 1.00 & 0.36 & 0.55 & NaN & 1.00 & 0.55 \\
\bottomrule
\end{tabular}
\end{table}

The support recovery performance for $\Omega$ in \Cref{table:omega_support} under \texttt{mixed-mSSL} shows consistently high specificity across most graph structures, indicating a strong ability to avoid false positives when detecting conditional independencies among responses. In particular, specificity values are generally high across all settings and graph types, except in the Tree Network scenario where it is reported as \texttt{NaN}. This \texttt{NaN} arises because the true precision matrix in the Tree Network is fully dense, i.e., there are no true zero entries in the off-diagonal positions. 

\begin{table}[htbp]
\centering
\caption{Support recovery and predictive performance across different covariance structures with $(n,p,q)=(500,1000,4)$ under the signal setting $\mathcal{U}[-5,5]$.}
\label{table:p1000q4unif}
\begin{tabular}{@{}l l cccc c ccc@{}}
\toprule
\multirow{2}{*}{\textbf{Scenario}} & \multirow{2}{*}{\textbf{Method}}
 & \multicolumn{4}{c}{\textbf{Support recovery}} 
 & 
 & \multicolumn{3}{c}{\textbf{Predictive}} \\
\cmidrule(lr){3-7} \cmidrule(lr){8-10}
  &   & \textbf{SEN} & \textbf{SPEC} & \textbf{PREC} & \textbf{ACC} 
  & \textbf{TIME(s)} & \textbf{RFE} & \textbf{RMSE} & \textbf{AUC} \\
\midrule
\multirow{4}{*}{AR1}
  & \texttt{mt-MBSP}     & 0.20 & 0.94 & 0.62 & \textbf{0.72} & 1239.60 & 62.05 & 1.64 & 0.59 \\
  & \texttt{sepSSL}      & 0.42 & 0.78 & 0.46 & 0.67 &  \textbf{2.60} & 61.66 & 1.05 & 0.54 \\
  & \texttt{sepglm}      & \textbf{0.55} & 0.75 & 0.48 & 0.69 & 10.22 & 61.48 & 0.49 & 0.61 \\
  & \texttt{mixed-mSSL}  & 0.16 & \textbf{0.95} & \textbf{0.67} & 0.71 & 277.20 & \textbf{61.42} & \textbf{0.48} & \textbf{0.63} \\
\midrule
\multirow{4}{*}{AR2}
  & \texttt{mt-MBSP}     & 0.19 & 0.95 & 0.64 & \textbf{0.72} & 915.60 & 62.04 & 1.64 & 0.59 \\
  & \texttt{sepSSL}      & 0.42 & 0.78 & 0.46 & 0.68 &  \textbf{1.60} & 61.66 & 1.05 & 0.55 \\
  & \texttt{sepglm}      & \textbf{0.55}  & 0.75 & 0.48 & 0.69 & 5.76 & 61.48 & \textbf{0.49} & 0.60 \\
  & \texttt{mixed-mSSL}  & 0.15 & \textbf{0.95} & \textbf{0.67} & 0.71 & 375.60 & \textbf{61.42} & \textbf{0.49} & \textbf{0.63} \\
\midrule
\multirow{4}{*}{Block Diagonal}
  & \texttt{mt-MBSP}     & 0.20 & 0.94 & 0.61 & \textbf{0.72} & 871.74 & 62.04 & 1.64 & 0.59 \\
  & \texttt{sepSSL}      & 0.43 & 0.78 & 0.46 & 0.67 & \textbf{1.47} & 61.67 & 1.04 & 0.54 \\
  & \texttt{sepglm}      & \textbf{0.55} & 0.75 & 0.49 & 0.69 & 5.30 & 61.48 & 0.49 & 0.60 \\
  & \texttt{mixed-mSSL}  & 0.16 & \textbf{0.95} & \textbf{0.65} & 0.71 & 14.52 & \textbf{61.42} & \textbf{0.48} & \textbf{0.63} \\
\midrule
\multirow{4}{*}{Star Graph}
  & \texttt{mt-MBSP}     & 0.19 & 0.94 & 0.63 & \textbf{0.71} & 449.40 & 62.04 & 1.64 & 0.59 \\
  & \texttt{sepSSL}      & 0.43 & 0.79 & 0.46 & 0.68 &  \textbf{1.09} & 61.66 & 1.05 & 0.54 \\
  & \texttt{sepglm}      & \textbf{0.54} & 0.75 & 0.48 & 0.69 & 3.51 & 61.49 & 0.49 & 0.59 \\
  & \texttt{mixed-mSSL}  & 0.16 & \textbf{0.95} & \textbf{0.67} & \textbf{0.71} & 19.97 & \textbf{61.43} & \textbf{0.48} & \textbf{0.63} \\
\midrule
\multirow{4}{*}{Small World}
  & \texttt{mt-MBSP}     & 0.19 & 0.95 & 0.62 & \textbf{0.72} & 459.00 & 62.05 & 1.64 & 0.60 \\
  & \texttt{sepSSL}      & 0.42 & 0.78 & 0.46 & 0.68 &  \textbf{1.15} & 61.66 & 1.03 & 0.54 \\
  & \texttt{sepglm}      & \textbf{0.54} & 0.75 & 0.49 & 0.69 & 3.63 & 61.48 & \textbf{0.50} & 0.60 \\
  & \texttt{mixed-mSSL}  & 0.16 & \textbf{0.95} & \textbf{0.65} & 0.71 & 24.58 & \textbf{61.43} & \textbf{0.50} & \textbf{0.62} \\
\midrule
\multirow{4}{*}{Tree Network}
  & \texttt{mt-MBSP}     & 0.19 & 0.95 & 0.63 & \textbf{0.72} & 491.40 & 62.04 & 1.64 & 0.59 \\
  & \texttt{sepSSL}      & 0.43 & 0.78 & 0.46 & 0.68 &  \textbf{1.23} & 61.66 & 1.05 & 0.54 \\
  & \texttt{sepglm}      & \textbf{0.55} & 0.75 & 0.48 & 0.69 & 3.93 & 61.48 & 0.49 & 0.59 \\
  & \texttt{mixed-mSSL}  & 0.15 & \textbf{0.95} & \textbf{0.68} & \textbf{0.72} & 26.87 & \textbf{61.42} & \textbf{0.48} & \textbf{0.63} \\
\bottomrule
\end{tabular}
\end{table}
\Cref{table:p1000q4unif} shows that \texttt{sepglm} finds the largest share of true signals (highest \texttt{SEN}) but at the expense of many false positives, whereas \texttt{mixed-mSSL} keeps the highest \texttt{SPEC} or \texttt{PREC} and nearly the best \texttt{ACC} sacrificing \texttt{SEN}. \texttt{mt-MBSP} displays a similar conservative behavior to \texttt{mixed-mSSL}, and \texttt{sepSSL} sits in between with moderate \texttt{SEN} but low \texttt{SPEC} and \texttt{PREC}. \texttt{mixed-mSSL} still consistently achieves the best predictive accuracy metrics.

\begin{table}[htbp]
\centering
\caption{Support recovery and predictive performance across different covariance structures with $(n,p,q)=(500,1000,4)$ under the signal setting $\mathcal{U}[-5,-2] \cup [2,5]$.}
\label{table:p1000q4disjt}
\begin{tabular}{@{}l l cccc c ccc@{}}
\toprule
\multirow{2}{*}{\textbf{Scenario}} & \multirow{2}{*}{\textbf{Method}}
 & \multicolumn{4}{c}{\textbf{Support recovery}} 
 & 
 & \multicolumn{3}{c}{\textbf{Predictive}} \\
\cmidrule(lr){3-7} \cmidrule(lr){8-10}
  &   & \textbf{SEN} & \textbf{SPEC} & \textbf{PREC} & \textbf{ACC} 
  & \textbf{TIME(s)} & \textbf{RFE} & \textbf{RMSE} & \textbf{AUC} \\
\midrule
\multirow{4}{*}{AR1}
  & \texttt{mt-MBSP}     & 0.19 & \textbf{0.93} & \textbf{0.55} & \textbf{0.71} & 460.80 & 77.50 & 1.64 & 0.59 \\
  & \texttt{sepSSL}      & 0.43 & 0.78 & 0.46 & 0.67 &  \textbf{1.09} & 77.20 & 1.08 & 0.55 \\
  & \texttt{sepglm}      & \textbf{0.55} & 0.75 & 0.48 & 0.69 & 3.64 & 77.05 & 0.64 & 0.60 \\
  & \texttt{mixed-mSSL}  & 0.32 & 0.86 & 0.52 & 0.70 & 5.84 & \textbf{76.99} & \textbf{0.65} & \textbf{0.63} \\
\midrule
\multirow{4}{*}{AR2}
  & \texttt{mt-MBSP}     & 0.19 & \textbf{0.93} & \textbf{0.56} & \textbf{0.71} & 442.80 & 77.49 & 1.64 & 0.59 \\
  & \texttt{sepSSL}      & 0.43 & 0.78 & 0.46 & 0.67 &  \textbf{1.08} & 77.20 & 1.14 & 0.53 \\
  & \texttt{sepglm}      & \textbf{0.55}  & 0.74 & 0.48 & 0.68 & 3.57 & 77.07 & 0.66 & 0.60 \\
  & \texttt{mixed-mSSL}  & 0.33 & 0.86 & 0.51 & \textbf{0.71} & 5.74 & \textbf{76.99} & \textbf{0.65} & \textbf{0.62} \\
\midrule
\multirow{4}{*}{Block Diagonal}
  & \texttt{mt-MBSP}     & 0.19 & \textbf{0.93} & \textbf{0.55} & \textbf{0.71} & 477.60 & 77.50 & 1.64 & 0.59 \\
  & \texttt{sepSSL}      & 0.43 & 0.78 & 0.46 & 0.67 & \textbf{1.18} & 77.20 & 1.06 & 0.53 \\
  & \texttt{sepglm}      & \textbf{0.55} & 0.75 & 0.48 & 0.69 & 4.00 & 77.07 & \textbf{0.65} & 0.60 \\
  & \texttt{mixed-mSSL}  & 0.33 & 0.86 & 0.51 & 0.70 & 6.12 & \textbf{76.99} & \textbf{0.65} & \textbf{0.62} \\
\midrule
\multirow{4}{*}{Star Graph}
  & \texttt{mt-MBSP}     & 0.19 & \textbf{0.93} & \textbf{0.56} & \textbf{0.71} & 510.24 & 77.50 & 1.64 & 0.58 \\
  & \texttt{sepSSL}      & 0.43 & 0.78 & 0.46 & 0.67 &  \textbf{1.29} & 77.20 & 1.08 & 0.53 \\
  & \texttt{sepglm}      & \textbf{0.54} & 0.75 & 0.48 & 0.69 & 4.35 & 77.06 & \textbf{0.64} & 0.60 \\
  & \texttt{mixed-mSSL}  & 0.33 & 0.86 & 0.52 & \textbf{0.71} & 7.27 & \textbf{76.99} & \textbf{0.64} & \textbf{0.62} \\
\midrule
\multirow{4}{*}{Small World}
  & \texttt{mt-MBSP}     & 0.19 & \textbf{0.93} & \textbf{0.54} & \textbf{0.70} & 475.32 & 77.50 & 1.64 & 0.59 \\
  & \texttt{sepSSL}      & 0.43 & 0.78 & 0.46 & 0.67 &  \textbf{1.25} & 77.20 & 1.12 & 0.53 \\
  & \texttt{sepglm}      & \textbf{0.54} & 0.75 & 0.48 & 0.69 & 4.12 & 77.07 & \textbf{0.65} & 0.60 \\
  & \texttt{mixed-mSSL}  & 0.32 & 0.86 & 0.51 & \textbf{0.70} & 6.28 & \textbf{76.99} & \textbf{0.65} & \textbf{0.62} \\
\midrule
\multirow{4}{*}{Tree Network}
  & \texttt{mt-MBSP}     & 0.18 & \textbf{0.93} & \textbf{0.55} & \textbf{0.71} & 486.30 & 77.49 & 1.64 & 0.59 \\
  & \texttt{sepSSL}      & 0.44 & 0.78 & 0.46 & 0.67 &  \textbf{1.34} & 77.20 & 1.14 & 0.53 \\
  & \texttt{sepglm}      & \textbf{0.55} & 0.74 & 0.48 & 0.68 & 4.20 & 77.06 & \textbf{0.64} & 0.60 \\
  & \texttt{mixed-mSSL}  & 0.33 & 0.86 & 0.52 & 0.70 & 7.42 & \textbf{76.99} & 0.65 & \textbf{0.63} \\
\bottomrule
\end{tabular}
\end{table}
Overall, in \Cref{table:p1000q4disjt}, the disjoint-uniform signals make all the methods slightly bolder (higher \texttt{SEN}) compared to \Cref{table:p1000q4unif}. \texttt{mixed-mSSL} offers a decent balance between finding real effects and controlling false discoveries, while also giving the best out-of-sample predictions.

\begin{table}[htbp]
\centering
\caption{Support recovery and predictive performance across different covariance structures with $(n,p,q)=(800,1000,6)$ under the signal setting $\mathcal{U}[-5,5]$.}
\label{table:p1000q6unif}
\begin{tabular}{@{}l l cccc c ccc@{}}
\toprule
\multirow{2}{*}{\textbf{Scenario}} & \multirow{2}{*}{\textbf{Method}}
 & \multicolumn{4}{c}{\textbf{Support recovery}} 
 & 
 & \multicolumn{3}{c}{\textbf{Predictive}} \\
\cmidrule(lr){3-7} \cmidrule(lr){8-10}
  &   & \textbf{SEN} & \textbf{SPEC} & \textbf{PREC} & \textbf{ACC} 
  & \textbf{TIME(s)} & \textbf{RFE} & \textbf{RMSE} & \textbf{AUC} \\
\midrule
\multirow{4}{*}{AR1}
  & \texttt{mt-MBSP}     & 0.49 & \textbf{0.98} & \textbf{0.92} & \textbf{0.83} & 4269.60 & 79.04 & 1.44 & 0.60 \\
  & \texttt{sepSSL}      & 0.58 & 0.72 & 0.47 & 0.68 &  \textbf{6.76} & 78.27 & 0.92 & 0.59 \\
  & \texttt{sepglm}      & \textbf{0.73} & 0.75 & 0.56 & 0.75 & 33.12 & 78.06 & 0.15 & \textbf{0.65} \\
  & \texttt{mixed-mSSL}  & 0.51 & 0.86 & 0.60 & 0.75 & 442.56 & \textbf{77.97} & \textbf{0.13} & \textbf{0.65} \\
\midrule
\multirow{4}{*}{AR2}
  & \texttt{mt-MBSP}     & 0.49 & \textbf{0.97} & \textbf{0.90} & \textbf{0.83} & 2127.60 & 79.04 & 1.44 & 0.61 \\
  & \texttt{sepSSL}      & 0.58 & 0.70 & 0.45 & 0.66 &  \textbf{4.02} & 78.27 & 0.92 & 0.59 \\
  & \texttt{sepglm}      & \textbf{0.73}  & 0.75 & 0.56 & 0.74 & 13.96 & 78.06 & 0.15 & \textbf{0.66} \\
  & \texttt{mixed-mSSL}  & 0.50 & 0.86 & 0.60 & 0.75 & 269.40 & \textbf{77.98} & \textbf{0.14} & 0.65 \\
\midrule
\multirow{4}{*}{Block Diagonal}
  & \texttt{mt-MBSP}     & 0.49 & \textbf{0.98} & \textbf{0.92} & \textbf{0.83} & 4302.00 & 79.04 & 1.44 & 0.60 \\
  & \texttt{sepSSL}      & 0.58 & 0.72 & 0.47 & 0.68 & \textbf{6.78} & 78.27 & 0.92 & 0.59 \\
  & \texttt{sepglm}      & \textbf{0.73} & 0.75 & 0.56 & 0.74 & 32.99 & 78.06 & 0.15 & 0.65 \\
  & \texttt{mixed-mSSL}  & 0.50 & 0.86 & 0.61 & 0.75 & 431.28 & \textbf{77.97} & \textbf{0.13} & \textbf{0.66} \\
\midrule
\multirow{4}{*}{Star Graph}
  & \texttt{mt-MBSP}     & 0.49 & \textbf{0.98} & \textbf{0.92} & \textbf{0.83} & 2274.60 & 79.04 & 1.44 & 0.61 \\
  & \texttt{sepSSL}      & 0.59 & 0.72 & 0.47 & 0.68 &  \textbf{4.19} & 78.27 & 0.91 & 0.59 \\
  & \texttt{sepglm}      & \textbf{0.73} & 0.75 & 0.56 & 0.75 & 15.40 & 78.06 & 0.15 & \textbf{0.65} \\
  & \texttt{mixed-mSSL}  & 0.50 & 0.86 & 0.61 & 0.76 & 271.80 & \textbf{77.97} & \textbf{0.13} & \textbf{0.65} \\
\midrule
\multirow{4}{*}{Small World}
  & \texttt{mt-MBSP}     & 0.47 & \textbf{0.94} & \textbf{0.77} & \textbf{0.80} & 1419.00 & 79.04 & 1.45 & 0.61 \\
  & \texttt{sepSSL}      & 0.58 & 0.69 & 0.45 & 0.66 &  \textbf{2.58} & 78.27 & 0.94 & 0.58 \\
  & \texttt{sepglm}      & \textbf{0.73} & 0.73 & 0.54 & 0.73 & 8.78 & 78.07 & \textbf{0.17} & \textbf{0.65} \\
  & \texttt{mixed-mSSL}  & 0.49 & 0.83 & 0.55 & 0.73 & 159.00 & \textbf{77.99} & 0.18 & \textbf{0.65} \\
\midrule
\multirow{4}{*}{Tree Network}
  & \texttt{mt-MBSP}     & 0.49 & \textbf{0.98} & \textbf{0.92} & \textbf{0.83} & 2110.20 & 79.04 & 1.44 & 0.60 \\
  & \texttt{sepSSL}      & 0.58 & 0.71 & 0.47 & 0.67 &  \textbf{3.90} & 78.27 & 0.92 & 0.59 \\
  & \texttt{sepglm}      & \textbf{0.73} & 0.75 & 0.56 & 0.74 & 13.74 & 78.06 & 0.15 & \textbf{0.66} \\
  & \texttt{mixed-mSSL}  & 0.50 & 0.86 & 0.61 & 0.75 & 227.88 & \textbf{77.97} & \textbf{0.13} & 0.65 \\
\bottomrule
\end{tabular}
\end{table}

\begin{table}[H]
\centering
\caption{Support recovery and predictive performance across different covariance structures with $(n,p,q)=(800,1000,6)$ under the signal setting $\mathcal{U}[-5,-2] \cup [2,5]$.}
\label{table:p1000q6disjt}
\begin{tabular}{@{}l l cccc c ccc@{}}
\toprule
\multirow{2}{*}{\textbf{Scenario}} & \multirow{2}{*}{\textbf{Method}}
 & \multicolumn{4}{c}{\textbf{Support recovery}} 
 & 
 & \multicolumn{3}{c}{\textbf{Predictive}} \\
\cmidrule(lr){3-7} \cmidrule(lr){8-10}
  &   & \textbf{SEN} & \textbf{SPEC} & \textbf{PREC} & \textbf{ACC} 
  & \textbf{TIME(s)} & \textbf{RFE} & \textbf{RMSE} & \textbf{AUC} \\
\midrule
\multirow{4}{*}{AR1}
  & \texttt{mt-MBSP}     & 0.53 & \textbf{0.98} & \textbf{0.94} & \textbf{0.85} & 2127.00 & 98.91 & 1.44 & 0.60 \\
  & \texttt{sepSSL}      & 0.60 & 0.70 & 0.46 & 0.67 &  \textbf{4.06} & 98.23 & 0.96 & 0.57 \\
  & \texttt{sepglm}      & \textbf{0.78} & 0.73 & 0.55 & 0.74 & 15.51 & 97.94 & 0.17 & 0.64 \\
  & \texttt{mixed-mSSL}  & 0.51 & 0.86 & 0.61 & 0.76 & 354.30 & \textbf{97.82} & \textbf{0.13} & \textbf{0.65} \\
\midrule
\multirow{4}{*}{AR2}
  & \texttt{mt-MBSP}     & 0.53 & \textbf{0.98} & \textbf{0.93} & \textbf{0.84} & 2064.60 & 98.92 & 1.44 & 0.60 \\
  & \texttt{sepSSL}      & 0.60 & 0.69 & 0.45 & 0.66 &  \textbf{3.83} & 98.23 & 0.96 & 0.57 \\
  & \texttt{sepglm}      & \textbf{0.78}  & 0.72 & 0.55 & 0.74 & 14.25 & 97.94 & 0.18 & \textbf{0.65} \\
  & \texttt{mixed-mSSL}  & 0.51 & 0.86 & 0.60 & 0.75 & 342.66 & \textbf{97.82} & \textbf{0.13} & \textbf{0.65} \\
\midrule
\multirow{4}{*}{Block Diagonal}
  & \texttt{mt-MBSP}     & 0.53 & \textbf{0.98} & \textbf{0.94} & \textbf{0.85} & 2046.60 & 98.91 & 1.44 & 0.60 \\
  & \texttt{sepSSL}      & 0.60 & 0.69 & 0.45 & 0.66 & \textbf{3.69} & 98.23 & 0.96 & 0.58 \\
  & \texttt{sepglm}      & \textbf{0.78} & 0.73 & 0.55 & 0.74 & 14.47 & 97.94 & 0.17 & \textbf{0.65} \\
  & \texttt{mixed-mSSL}  & 0.50 & 0.86 & 0.60 & 0.76 & 319.20 & \textbf{97.82} & \textbf{0.13} & 0.64 \\
\midrule
\multirow{4}{*}{Star Graph}
  & \texttt{mt-MBSP}     & 0.53 & \textbf{0.98} & \textbf{0.94} & \textbf{0.85} & 1956.60 & 98.91 & 1.45 & 0.59 \\
  & \texttt{sepSSL}      & 0.60 & 0.71 & 0.47 & 0.68 &  \textbf{3.67} & 98.23 & 0.95 & 0.58 \\
  & \texttt{sepglm}      & \textbf{0.78} & 0.73 & 0.55 & 0.74 & 13.76 & 97.95 & 0.19 & \textbf{0.65} \\
  & \texttt{mixed-mSSL}  & 0.51 & 0.86 & 0.61 & 0.76 & 294.96 & \textbf{97.84} & \textbf{0.17} & \textbf{0.65} \\
\midrule
\multirow{4}{*}{Small World}
  & \texttt{mt-MBSP}     & 0.52 & \textbf{0.93} & \textbf{0.78} & \textbf{0.81} & 1759.68 & 98.92 & 1.44 & 0.60 \\
  & \texttt{sepSSL}      & 0.60 & 0.68 & 0.45 & 0.66 &  \textbf{3.51} & 98.23 & 0.95 & 0.59 \\
  & \texttt{sepglm}      & \textbf{0.78} & 0.71 & 0.54 & 0.73 & 12.45 & 97.95 & 0.19 & 0.64 \\
  & \texttt{mixed-mSSL}  & 0.51 & 0.83 & 0.57 & 0.74 & 261.60 & \textbf{97.84} & \textbf{0.17} & \textbf{0.65} \\
\midrule
\multirow{4}{*}{Tree Network}
  & \texttt{mt-MBSP}     & 0.53 & \textbf{0.98} & \textbf{0.94} & \textbf{0.85} & 1807.74 & 98.91 & 1.44 & 0.60 \\
  & \texttt{sepSSL}      & 0.60 & 0.69 & 0.46 & 0.67 &  \textbf{3.44} & 98.23 & 0.98 & 0.58 \\
  & \texttt{sepglm}      & \textbf{0.78} & 0.72 & 0.55 & 0.74 & 13.89 & 97.94 & 0.17 & 0.64 \\
  & \texttt{mixed-mSSL}  & 0.51 & 0.87 & 0.61 & 0.76 & 260.10 & \textbf{97.82} & 0.12 & \textbf{0.65} \\
\bottomrule
\end{tabular}
\end{table}
The common story for \Cref{table:p1000q6unif} and \Cref{table:p1000q6disjt} is the superior predictive performance for \texttt{mixed-mSSL} and a balanced trade-off between \texttt{SEN}, \texttt{SPEC} and \texttt{PREC}. It is worth noting that in this large $n$ and large $p$ setting, \texttt{mt-MBSP} trades enormous computation for near-perfect \texttt{PREC} and \texttt{SPEC} but middling prediction.

\subsection{Hyperparameter Sensitivity Analysis}
Below is a sensitivity analysis for the two spike–slab penalty hyperparameters in the AR1 setting with $(n,p,q) = (200,500,4)$ under the signal setting $\mathcal{U}[-5,-2]\cup[2,5]$. We compare the default penalties $(\lambda_0,\lambda_1)$ elaborated in Section 3 to halved penalties $(\lambda_0/2,\lambda_1/2)$ and doubled penalties $(2 \lambda_0, 2\lambda_1)$. Table \ref{table:hp_sensitivity} shows that halving both spike‐penalty parameters $(\lambda_0/2,\lambda_1/2)$ yields higher sensitivity  at the expense of specificity and precision, while doubling them $(2 \lambda_0, 2 \lambda_1)$ produces the reverse effect—a conservative fit with very high specificity but low sensitivity. Notably, the predictive metrics remain largely stable across this two‐fold range of penalties, demonstrating that \texttt{mixed-mSSL}’s out‐of‐sample performance is robust to moderate hyperparameter changes. 
\begin{table}[H]
\centering
\caption{Hyperparameter sensitivity for \texttt{mixed-mSSL} under AR(1), \((n,p,q)=(200,500,4)\).}
\label{table:hp_sensitivity}
\small
\begin{tabular}{@{}lccccccc@{}}
\toprule
\textbf{Penalties} 
  & \multicolumn{4}{c}{\textbf{Support}} 
  & \multicolumn{3}{c}{\textbf{Predictive}} \\
\cmidrule(lr){2-5} \cmidrule(lr){6-8}
  & SEN & SPEC & PREC & ACC 
  & RFE & RMSE & AUC  \\
\midrule
\(\bigl(\tfrac{\lambda_0}{2},\tfrac{\lambda_1}{2}\bigr)\)
  & 0.42 & 0.80 & 0.47 & 0.69 
  & 54.57 & 0.78 & 0.60 \\
\((\lambda_0,\lambda_1)\) (default)
  & 0.35 & 0.86 & 0.51 & 0.70 
  & 54.57 & 0.77 & 0.60 \\
\(\bigl(2\lambda_0,2\lambda_1\bigr)\)
  & 0.10 & 0.99 & 0.78 & 0.70 
  & 54.60 & 0.80 & 0.58 \\
\bottomrule
\end{tabular}
\end{table}

%% file: additional_real.tex
\subsection{Results for the chronic kidney disease (CKD) application} 
The top ten important biomarkers identified by \citet{CKDdata}, which are treated as the ``gold-standard" for our classification metrics were: Albumin (AL), Hemoglobin (HEMO), Packed Cell Volume (PCV), Red Blood Cell Count (RBCC), Serum Creatinine (SC), Blood Glucose Random (BGR), Blood Urea (BU), Sodium (SOD), White Blood Cell Count (WBCC), and Hypertension (HTN).
Table \ref{tab:ckd_selected_list} summarizes the set of selected biomarkers for our CKD analysis.
\begin{table}[H]
\centering
\caption{CKD biomarkers selected by each method}
\label{tab:ckd_selected_list}
\begin{tabular}{@{}lp{0.75\textwidth}@{}}
\toprule
\textbf{Method}   & \textbf{Selected biomarkers} \\
\midrule
\texttt{mixed-mSSL} & AL, RBCC, HEMO, DM, BGR, PCV, HTN \\[4pt]
\texttt{mt-MBSP}    & AL, SC, HEMO, DM, SOD \\[4pt]
\texttt{sepSSL}     & AL, HEMO, DM, PCV, HTN, BP \\[4pt]
\texttt{sepglm}     & BP, AL, SU, RBC, PCC, BA, BGR, SC, SOD, HEMO, \\[4pt]
                    & PCV, WBCC, RBCC, HTN, DM, APPET \\[4pt]
\bottomrule
\end{tabular}
\end{table}
The key to these biomarker abbreviations are listed below: 
\begin{table}[H]
\centering
\caption{CKD biomarker abbreviations}
\label{tab:ckd_abbrev_key}
\small
\begin{tabular}{@{}ll@{}}
\toprule
\textbf{Code} & \textbf{Description} \\ \midrule
BP    & Blood pressure (mm/Hg) \\
AL    & Albumin (qualitative) \\
SU    & Sugar (qualitative) \\
RBC   & Red blood cells (qualitative) \\
RBCC  & Red blood cell count ($\times10^{6}/\mu$L) \\
PCC   & Pus cell clumps (qualitative) \\
BA    & Bacteria (qualitative) \\
BGR   & Random blood glucose (mg/dL) \\
SC    & Serum creatinine (mg/dL) \\
SOD   & Sodium (mEq/L) \\
HEMO  & Hemoglobin (g/dL) \\
PCV   & Packed cell volume (\%) \\
WBCC  & White blood cell count ($\times10^{3}/\mu$L) \\
HTN   & Hypertension status (yes/no) \\
DM    & Diabetes mellitus status (yes/no) \\
APPET & Appetite (good/poor) \\ \bottomrule
\end{tabular}
\end{table}

%% file: ms.bbl
\begin{thebibliography}{}

\bibitem[Ash and Doleans-Dade, 2000]{Ash2000}
Ash, R.~B. and Doleans-Dade, C.~A. (2000).
\newblock {\em Probability and measure theory}.
\newblock Elsevier Science.

\bibitem[Bai and Ghosh, 2018]{BAIGhoshJMVA}
Bai, R. and Ghosh, M. (2018).
\newblock High-dimensional multivariate posterior consistency under
  global--local shrinkage priors.
\newblock {\em Journal of Multivariate Analysis}, 167:157--170.

\bibitem[Banerjee and Ghosal, 2015]{Banerjee2015}
Banerjee, S. and Ghosal, S. (2015).
\newblock {B}ayesian structure learning in graphical models.
\newblock {\em Journal of Multivariate Analysis}, 136:147--162.

\bibitem[Berthold et~al., 2003]{Newton2003}
Berthold, P., Gwinner, E., and Sonnenschein, E., editors (2003).
\newblock {\em Geographical patterns in bird migration}, Berlin, Heidelberg.
  Springer Berlin Heidelberg.

\bibitem[Canale and Dunson, 2011]{CanaleDunson2011}
Canale, A. and Dunson, D.~B. (2011).
\newblock {B}ayesian kernel mixtures for counts.
\newblock {\em Journal of the American Statistical Association},
  106(496):1528--1539.

\bibitem[Carvalho et~al., 2009]{CarvalhoHS}
Carvalho, C.~M., Polson, N.~G., and Scott, J.~G. (2009).
\newblock Handling sparsity via the horseshoe.
\newblock In {\em Proceedings of the 12th International Conference on
  Artificial Intelligence and Statistics}.

\bibitem[Castillo et~al., 2015]{Castillo2015}
Castillo, I., Schmidt-Hieber, J., and van~der Vaart, A. (2015).
\newblock {B}ayesian linear regression with sparse priors.
\newblock {\em The Annals of Statistics}, 43(5):1986--2018.

\bibitem[Chib and Greenberg, 1998]{Chib1998AnalysisOM}
Chib, S. and Greenberg, E. (1998).
\newblock Analysis of multivariate probit models.
\newblock {\em Biometrika}, 85:347--361.

\bibitem[Deshpande et~al., 2019]{Deshpande2019_mSSL}
Deshpande, S.~K., Ro{\v c}kov{\'a}, V., and George, E.~I. (2019).
\newblock Simultaneous variable and covariance selection with the multivariate
  spike-and-slab {LASSO}.
\newblock {\em Journal of Computational and Graphical Statistics},
  28(4):921--931.

\bibitem[Ebiaredoh-Mienye et~al., 2022]{CKDdata}
Ebiaredoh-Mienye, S.~A., Swart, T.~G., Esenogho, E., and Mienye, I.~D. (2022).
\newblock A machine learning method with filter-based feature selection for
  improved prediction of chronic kidney disease.
\newblock {\em Bioengineering}, 9(8).

\bibitem[Ekvall and Molstad, 2022]{ekvall}
Ekvall, K.~O. and Molstad, A.~J. (2022).
\newblock Mixed-type multivariate response regression with covariance
  estimation.
\newblock {\em Statistics in Medicine}, 41(15):2768--2785.

\bibitem[Friedman et~al., 2010]{glmnet2010}
Friedman, J.~H., Hastie, T., and Tibshirani, R. (2010).
\newblock Regularization paths for generalized linear models via coordinate
  descent.
\newblock {\em Journal of Statistical Software}, 33(1):1--22.

\bibitem[George and Ro{\v{c}}kov{\'{a}},
  2020]{GeorgeRockova2020_penalty_mixing}
George, E.~I. and Ro{\v{c}}kov{\'{a}}, V. (2020).
\newblock Comment: regularization via {B}ayesian penalty mixing.
\newblock {\em Technometrics}, 62(4):438--442.

\bibitem[Gessner et~al., 2020]{gessner2020integralsgaussianslineardomain}
Gessner, A., Kanjilal, O., and Hennig, P. (2020).
\newblock Integrals over {G}aussians under linear domain constraints.
\newblock In {\em Proceedings of the 23rd International Conference on
  Artificial Intelligence and Statistics}.

\bibitem[Ghosal et~al., 2000]{Ghosal2000}
Ghosal, S., Ghosh, J.~K., and van~der Vaart, A.~W. (2000).
\newblock Convergence rates of posterior distributions.
\newblock {\em The Annals of Statistics}, 28(2):500--531.

\bibitem[Ghosal and {van der Vaart}, 2007]{Ghosal_2007}
Ghosal, S. and {van der Vaart}, A. (2007).
\newblock Convergence rates of posterior distributions for noniid observations.
\newblock {\em Annals of Statistics}, 35(1):192--223.

\bibitem[Hassanein and Shafi, 2022]{HassaneinCKD}
Hassanein, M. and Shafi, T. (2022).
\newblock Assessment of glycemia in chronic kidney disease.
\newblock {\em BMC Medicine}, 20(1):117.

\bibitem[Hsieh et~al., 2014]{JMLR:v15:hsieh14a}
Hsieh, C.-J., Sustik, M.~A., Dhillon, I.~S., and Ravikumar, P. (2014).
\newblock {QUIC}: Quadratic approximation for sparse inverse covariance
  estimation.
\newblock {\em Journal of Machine Learning Research}, 15(83):2911--2947.

\bibitem[J{\"a}rvinen and V{\"a}is{\"a}nen, 2006]{Jarvinen2006}
J{\"a}rvinen, O. and V{\"a}is{\"a}nen, R. (2006).
\newblock Changes in bird populations as criteria of environmental changes.
\newblock {\em Ecography}, 2:75--80.

\bibitem[Karlsson et~al., 2013]{Karlsson2013T2D}
Karlsson, F.~H., Tremaroli, V., Nookaew, I., Bergstr{\"o}m, G., Behre, C.~J.,
  Fagerberg, B., Nielsen, J., and B{\"a}ckhed, F. (2013).
\newblock Gut metagenome in {European} women with normal, impaired and diabetic
  glucose control.
\newblock {\em Nature}, 498(7452):99--103.

\bibitem[Kostic et~al., 2013]{Kostic2013CRC}
Kostic, A.~D., Chun, E., Robertson, L., Glickman, J.~N., Gallini, C.~A.,
  Michaud, M., Clancy, T.~E., Chung, D.~C., Lochhead, P., Hold, G.~L., El-Omar,
  E.~M., Brenner, D., Fuchs, C.~S., Meyerson, M., and Garrett, W.~S. (2013).
\newblock Fusobacterium nucleatum potentiates intestinal tumorigenesis and
  modulates the tumor-immune microenvironment.
\newblock {\em Cell Host \& Microbe}, 14(2):207--215.

\bibitem[Kowal and Canale, 2020]{KowalCanale2019_star}
Kowal, D.~R. and Canale, A. (2020).
\newblock Simultaneous transformation and rounding {(STAR)} models for
  integer-valued data.
\newblock {\em Electronic Journal of Statistics}, 14(1):1744--1772.

\bibitem[Kumar et~al., 2023]{KumarCKD}
Kumar, M., Dev, S., Khalid, M.~U., Siddenthi, S., Noman, M., John, C.,
  Akubuiro, C., Haider, A., Rani, R., Kashif, M., Varrassi, G., Khatri, M.,
  Kumar, S., and Mohamad, T. (2023).
\newblock The bidirectional link between diabetes and kidney disease:
  mechanisms and management.
\newblock {\em Cureus}, 15:e45615.

\bibitem[Laurent and Massart, 2000]{LaurentMassart2000}
Laurent, B. and Massart, P. (2000).
\newblock Adaptive estimation of a quadratic functional by model selection.
\newblock {\em The Annals of Statistics}, 28(5):1302--1338.

\bibitem[Li et~al., 2023]{LiGhosh2023}
Li, X., ~, Ghosh, J., , and Villarini, G. (2023).
\newblock A comparison of {B}ayesian multivariate versus univariate normal
  regression models for prediction.
\newblock {\em The American Statistician}, 77(3):304--312.

\bibitem[Lindstr{\"o}m et~al., 2015]{Lindstrom15}
Lindstr{\"o}m, {\AA}., Green, M., Husby, M., K{\aa}l{\aa}s, J.~A., and
  Lehikoinen, A. (2015).
\newblock Large-scale monitoring of waders on their {Boreal} and {Arctic}
  breeding grounds in {Northern Europe}.
\newblock {\em Ardea}, 103(1):3--15.

\bibitem[Liu et~al., 2023]{Liu2023}
Liu, J., Huang, X., Chen, C., Wang, Z., Huang, Z., Qin, M., He, F., Tang, B.,
  Long, C., Hu, H., Pan, S., Wu, J., and Tang, W. (2023).
\newblock Identification of colorectal cancer progression-associated intestinal
  microbiome and predictive signature construction.
\newblock {\em Journal of Translational Medicine}, 21(1):373.

\bibitem[Louis et~al., 2014]{Louis2014Butyrate}
Louis, P., Hold, G.~L., and Flint, H.~J. (2014).
\newblock The gut microbiota, bacterial metabolites and colorectal cancer.
\newblock {\em Nature Reviews Microbiology}, 12(10):661--672.

\bibitem[Mandic et~al., 2024]{Mandic24}
Mandic, M., Li, H., Safizadeh, F., Niedermaier, T., Hoffmeister, M., and
  Brenner, H. (2024).
\newblock Correction: Is the association of overweight and obesity with
  colorectal cancer underestimated? {An} umbrella review of systematic reviews
  and meta-analyses.
\newblock {\em European Journal of Epidemiology}, 39(2):231--231.

\bibitem[Nie and Ro{\v c}kov{\'a}, 2023]{NieRockova2021}
Nie, L. and Ro{\v c}kov{\'a}, V. (2023).
\newblock Bayesian bootstrap spike-and-slab {LASSO}.
\newblock {\em Journal of the American Statistical Association},
  118(543):2013--2028.

\bibitem[Ovaskainen and Abrego, 2020]{Ovaskainen2020}
Ovaskainen, O. and Abrego, N. (2020).
\newblock {\em Joint species distribution modelling: with applications in R}.
\newblock Cambridge University Press.

\bibitem[Polson et~al., 2013]{Polson2013_polyagamma}
Polson, N.~G., Scott, J.~G., and Winle, J. (2013).
\newblock {B}ayesian inference for logistic models using {P\'{o}lya-Gamma}
  latent variables.
\newblock {\em Journal of the American Statistical Association},
  108(504):1339--1349.

\bibitem[Qin et~al., 2012]{Qin2012T2D}
Qin, J., Li, Y., Cai, Z., Shenghui, L., Zhu, J., Zhang, F., Liang, S., Zhang,
  W., Guan, Y., Shen, D., Peng, Y., Zhang, D., Jie, Z., Wu, W., Qin, Y., Xue,
  W., Li, J., Han, L., Lu, D., and Wang, J. (2012).
\newblock A metagenome-wide association study of gut microbiota in type 2
  diabetes.
\newblock {\em Nature}, 490:55--60.

\bibitem[Rasmussen and Williams, 2005]{Rasmussen2005}
Rasmussen, C.~E. and Williams, C. K.~I. (2005).
\newblock {\em Gaussian processes for machine learning}.
\newblock The MIT Press.

\bibitem[Razavi et~al., 2024]{Razavi2024}
Razavi, S., Amirmozafari, N., Zahedi~bialvaei, A., Navab-Moghadam, F., Khamseh,
  M.~E., Alaei-Shahmiri, F., and Sedighi, M. (2024).
\newblock Gut microbiota composition and {Type} 2 diabetes: Are these subjects
  linked together?
\newblock {\em Heliyon}, 10(20):e39464.

\bibitem[Ro{\v c}kov{\'a} and George, 2018]{RockovaGeorge2018_ssl}
Ro{\v c}kov{\'a}, V. and George, E.~I. (2018).
\newblock The spike-and-slab {LASSO}.
\newblock {\em Journal of the American Statistical Association},
  113(521):431--444.

\bibitem[Roverato, 2002]{Roverato2002}
Roverato, A. (2002).
\newblock Hyper inverse wishart distribution for non-decomposable graphs and
  its application to {B}ayesian inference for {G}aussian graphical models.
\newblock {\em Scandinavian Journal of Statistics}, 29(3):391--411.

\bibitem[Sagar et~al., 2024]{KNHS}
Sagar, K., Banerjee, S., Datta, J., and Bhadra, A. (2024).
\newblock Precision matrix estimation under the horseshoe-like prior--penalty
  dual.
\newblock {\em Electronic Journal of Statistics}, 18(1):1--46.

\bibitem[Sarkar et~al.,
  2025]{sarkar2024posteriorconsistencymultiresponseregression}
Sarkar, P., Khare, K., and Ghosh, M. (2025).
\newblock Posterior consistency in multi-response regression models with
  non-informative priors for the error covariance matrix in growing dimensions.
\newblock {\em Bernoulli}, 31(3):2403--2433.

\bibitem[Shen and Deshpande, 2025]{shen2022posterior}
Shen, Y. and Deshpande, S.~K. (2025).
\newblock Posterior contraction and uncertainty quantification for the
  multivariate spike-and-slab {LASSO}.
\newblock {\em Journal of Multivariate Analysis}, 210:105493.

\bibitem[Shen et~al., 2024]{shencgssl}
Shen, Y., Sol{\'\i}s-Lemus, C., and Deshpande, S.~K. (2024).
\newblock Estimating sparse direct effects in multivariate regression with the
  spike-and-slab {LASSO}.
\newblock {\em Bayesian Analysis}, 20(3):1031--1055.

\bibitem[Song and Liang, 2023]{song2022nearly}
Song, Q. and Liang, F. (2023).
\newblock Nearly optimal {B}ayesian shrinkage for high dimensional regression.
\newblock {\em Science China Mathematics}, 66(2):409--442.

\bibitem[Virkkala and Lehikoinen, 2014]{Virkkala2014}
Virkkala, R. and Lehikoinen, A. (2014).
\newblock Patterns of climate-induced density shifts of species: poleward
  shifts faster in northern {B}oreal birds than in southern birds.
\newblock {\em Global Change Biology}, 20(10):2995--3003.

\bibitem[Wang et~al., 2025]{wang2023twostep}
Wang, S.-H., Bai, R., and Huang, H.-H. (2025).
\newblock Two-step mixed-type multivariate {B}ayesian sparse variable selection
  with shrinkage priors.
\newblock {\em Electronic Journal of Statistics}, 19(1):397--457.

\bibitem[Watts and Strogatz, 1998]{Watts1998}
Watts, D.~J. and Strogatz, S.~H. (1998).
\newblock Collective dynamics of `small-world'networks.
\newblock {\em Nature}, 393(6684):440--442.

\bibitem[Wei and Ghosal, 2020]{WEI2020215}
Wei, R. and Ghosal, S. (2020).
\newblock Contraction properties of shrinkage priors in logistic regression.
\newblock {\em Journal of Statistical Planning and Inference}, 207:215--229.

\bibitem[Wirbel et~al., 2019]{Wirbel2019meta}
Wirbel, J., Pyl, P.~T., Kartal, E., Zych, K., Kashani, A., Milanese, A., Fleck,
  J.~S., Voigt, A.~Y., Palleja, A., Ponnudurai, R., et~al. (2019).
\newblock Meta-analysis of fecal metagenomes reveals global microbial
  signatures that are specific for colorectal cancer.
\newblock {\em Nature Medicine}, 25:679--689.

\bibitem[Zeller et~al., 2014]{zeller2014potential}
Zeller, G., Tap, J., Voigt, A.~Y., Sunagawa, S., Kultima, J.~R., Costea, P.~I.,
  Amiot, A., B{\"o}hm, J., Brunetti, F., Habermann, N., et~al. (2014).
\newblock Potential of fecal microbiota for early-stage detection of colorectal
  cancer.
\newblock {\em Molecular Systems Biology}, 10(11):766.

\bibitem[Zellner, 1962]{Zellner1962_sur}
Zellner, A. (1962).
\newblock An efficient method for estimating seemingly unrelated regressions
  and tests for aggregation bias.
\newblock {\em Journal of the American Statistical Association},
  57(298):348--368.

\bibitem[Zhang and Huang, 2008]{Zhang2008}
Zhang, C.-H. and Huang, J. (2008).
\newblock The sparsity and bias of the {LASSO} selection in high-dimensional
  linear regression.
\newblock {\em The Annals of Statistics}, 36(4).

\bibitem[Zhang et~al., 2022]{ZHANG2022154}
Zhang, R., Yao, Y., and Ghosh, M. (2022).
\newblock Contraction of a quasi-{B}ayesian model with shrinkage priors in
  precision matrix estimation.
\newblock {\em Journal of Statistical Planning and Inference}, 221:154--171.

\bibitem[Zhao and Yu, 2006]{ZhaoYu2006}
Zhao, P. and Yu, B. (2006).
\newblock On model selection consistency of {LASSO}.
\newblock {\em Journal of Machine Learning Research}, 7(90):2541--2563.

\bibitem[Zito and Miller, 2024]{ZitoMiller2024}
Zito, A. and Miller, J.~W. (2024).
\newblock Compressive {B}ayesian non-negative matrix factorization for
  mutational signatures analysis.
\newblock \href{https://arxiv.org/abs/2404.10974}{\texttt{arXiv}:2404.10964}.

\end{thebibliography}
